\documentclass{asl}
\usepackage{latexsym}
\usepackage{url}
\usepackage{tikz}
\usepackage{rotating}
\usepackage{pdflscape}

\theoremstyle{plain}
\newtheorem{proposition}{Proposition}
\newtheorem{theorem}[proposition]{Theorem}
		
\newtheorem{example}[proposition]{Example}
\newtheorem{corollary}[proposition]{Corollary}

\newtheorem{lemma}[proposition]{Lemma}
\newtheorem{remark}[proposition]{Remark}
\newtheorem{lclaim}{{\sc Claim}}[proposition]

\newtheorem{question}{\sc Question}

\newcommand{\mf}[1]{{\sf #1}}

\newcommand{\kmodels}{\models}  
\newcommand{\varel}[1]{\hat{#1}}
\newcommand{\varset}{\textsf{var}}
\newcommand{\idmap}{\textit{id}}
\newcommand{\newrel}{S}

\newcommand{\OWL}{\textsl{OWL\,2}}
\newcommand{\OWLEL}{\textsl{OWL\,2\,EL}}

\newcommand{\avec}[1]{\boldsymbol{#1}}

\newcommand{\A}{\ensuremath{\mathfrak A}}
\newcommand{\Aa}{\ensuremath{\mathfrak A}}
\newcommand{\vala}{\ensuremath{\mathfrak a}}
\newcommand{\Bb}{\ensuremath{\mathfrak B}}
\newcommand{\V}{\ensuremath{\mathfrak v}}

\newcommand{\T}{\ensuremath{\mathfrak{T}}}
\newcommand{\F}{\ensuremath{\mathfrak F}}
\newcommand{\Ff}{\ensuremath{\mathfrak F}}

\newcommand{\G}{\ensuremath{\mathfrak{G}}}

\newcommand{\Hh}{\ensuremath{\mathfrak{H}}}
\newcommand{\M}{\ensuremath{\mathfrak{M}}}
\newcommand{\N}{\ensuremath{\mathfrak{N}}}

\newcommand{\J}{\ensuremath{\mathcal{J}}}

\newcommand{\R}{\ensuremath{\mathcal{R}}}

\newcommand{\Fc}{\ensuremath{\mathcal{F}}}

\newcommand{\Cc}{\ensuremath{\mathcal{C}}}

\newcommand{\Fslo}{\F^\star}

\newcommand{\Gslo}{\G^\star}

\newcommand{\imp}{\to}  

\renewcommand{\d}{\Diamond}
\newcommand{\dR}{\d_{\!R}\,}
\newcommand{\dRo}{\d_{\!S}\,}
\newcommand{\dS}{\d_{\!S}\,}
\newcommand{\dZ}{\d_{\!Z}}
\newcommand{\dX}{\d_{\!X}}
\newcommand{\dRca}{\d_{\!R}^+}

\newcommand{\dnext}{\d_{\!{\sf next}}\,}
\newcommand{\dstep}{\d_{\!{\sf step}}\,}
\newcommand{\dright}{\d_{\!{\sf right}}}
\newcommand{\dhead}{\d_{{\sf head}}}
\newcommand{\dleft}{\d_{{\sf left}}}
\newcommand{\dq}{\d_{\!q}}
\newcommand{\dqo}{\d_{\!q_0}}
\newcommand{\dqh}{\d_{\!q_h}}
\newcommand{\dqp}{\d_{\!q'}}
\newcommand{\da}{\d_{\!a}}
\newcommand{\dap}{\d_{\!a'}}

\newcommand{\EL}{\mathcal{EL}}
\newcommand{\SP}{spi}
\newcommand{\SPa}{sp}
\newcommand{\SPcap}{Spi}
\newcommand{\SPb}{sp${}^\bot$}
\newcommand{\SPbi}{spi${}^\bot$}

\newcommand{\SLOa}{{\rm SLO}}
\newcommand{\SLOba}{{\rm SLO}^\bot}
\newcommand{\DLOa}{{\rm DLO}}
\newcommand{\BAOa}{{\rm BAO}}
\newcommand{\CAa}{\ensuremath{\mathsf{CA}}}

\newcommand{\SLO}{\ensuremath{\mathsf{SLO}}}
\newcommand{\SLOb}{\ensuremath{\mathsf{SLO}}^\bot}

\newcommand{\BAO}{\ensuremath{\mathsf{BAO}}}
\newcommand{\CA}{\ensuremath{\mathsf{Kr}}}

\newcommand{\FA}{\Fc(\Aa)}

\newcommand{\PFA}{\Fc_{\textit{p}}(\Aa)}

\newcommand{\SPL}{L}
\newcommand{\Ec}{\varSigma}
\newcommand{\Ll}{\varLambda}
\newcommand{\Qc}{\ensuremath{r\!L}}

\newcommand{\q}{\boldsymbol{\rho}}
\newcommand{\e}{\boldsymbol{\iota}}

\newcommand{\ew}{\e_{\textit{width}}}
\newcommand{\eden}{\e_{\textit{dense}}}
\newcommand{\erefl}{\e_{\textit{refl}}}
\newcommand{\etrans}{\e_{\textit{trans}}}
\newcommand{\esym}{\e_{\textit{sym}}}
\newcommand{\efun}{\e_{\textit{fun}}}
\newcommand{\eeuc}{\e_{\textit{eucl}}}
\newcommand{\ewc}{\e_{\textit{wcon}}}
\newcommand{\edep}{\e_{\textit{depth}}}

\newcommand{\Elin}{\mathsf{SPi}_{\textit{lin}}}
\newcommand{\Axlin}{\Ec_{\textit{lin}}}
\newcommand{\Esfour}{\mathsf{SPi}_{\textit{qo}}}
\newcommand{\Axsfour}{\Ec_{\textit{qo}}}
\newcommand{\Esfive}{\mathsf{SPi}_{\textit{equiv}}}
\newcommand{\Axsfive}{\Ec_{\textit{equiv}}}

\newcommand{\Efunn}{\mathsf{SPi}^n_{\textit{equiv}}}
\newcommand{\Axfunn}{\Ec^n_{\textit{equiv}}}

\newcommand{\Efunone}{\mathsf{SPi}^1_{\textit{equiv}}}

\newcommand{\trm}{{\sf for}}

\newcommand{\spterm}{sp-formula}
\newcommand{\spterms}{sp-formulas}
\newcommand{\sptermscap}{Sp-formulas}
\newcommand{\spequation}{sp-im\-pli\-ca\-tion}
\newcommand{\spequations}{sp-im\-pli\-ca\-tions}
\newcommand{\spequationscap}{Sp-im\-pli\-ca\-tions}
\newcommand{\SPi}{\mathsf{SPi}}
\newcommand{\sptheory}{spi-logic}
\newcommand{\sptheories}{spi-logics}
\newcommand{\hornequation}{Horn-implication}
\newcommand{\hornequations}{Horn-im\-pli\-ca\-tions}
\newcommand{\leapfrogequation}{leapfrog implication}
\newcommand{\leapfrogequations}{leapfrog implications}

\newcommand{\spquasiequation}{spi-rule}
\newcommand{\spquasiequations}{spi-rules}
\newcommand{\spquasiequationscap}{Spi-rules}
\newcommand{\qsptheory}{spi-rule logic}
\newcommand{\qsptheories}{spi-rule logics}

\newcommand{\sptype}{sp-type}

\newcommand{\embed}{embeddable}

\newcommand{\qeqcons}{rule-cons\-er\-va\-tive}
\newcommand{\eqcons}{conservative}
\newcommand{\complex}{complex}
\newcommand{\qeqcomp}{globally complete}
\newcommand{\eqcomp}{complete}

\newcommand{\tprof}{tree-profile}

\newcommand{\nterm}{normal form}
\newcommand{\method}{syntactic proxies}
\newcommand{\methodcapital}{Syntactic proxies}

\newcommand{\Le}{\ensuremath{\mathsf{left}}}
\renewcommand{\Re}{\ensuremath{\mathsf{right}}}

\newcommand{\AxM}{\Ec_M}
\newcommand{\TMl}{\mathsf{L}}
\newcommand{\TMr}{\mathsf{R}}
\newcommand{\TMeq}{\e_M}

\title[Completeness of  Strictly Positive Modal Logics]{Kripke Completeness of  Strictly Positive Modal Logics over Meet-semilattices with Operators}

\author[S.~Kikot]{Stanislav~Kikot}\revauthor{Kikot, Stanislav}
\author[A.~Kurucz]{Agi~Kurucz}\revauthor{Kurucz, Agi}
\author[Y.~Tanaka]{Yoshihito~Tanaka}\revauthor{Tanaka, Yoshihito}
\author[F.~Wolter]{Frank~Wolter}\revauthor{Wolter, Frank}
\author[M.~Zakharyaschev]{Michael~Zakharyaschev}\revauthor{Zakharyaschev, Michael}

\address{Department of Computer Science\\
University of Oxford\\
Wolfson Building, Parks Road, Oxford OX1 3QD, U.K.  
}
\address{
Department of Computer Science and Information Systems\\
Birkbeck, University of London\\
Malet Street, London WC1E 7HX, U.K. 
}
\address{
Institute for Information Transmission Problems,\\ 
19 Bolshoy Karetny pereulok, Moscow 127051, Russia 
}
\address{
Moscow Institute of Physics and Technology,\\ 
9 Institutskiy pereulok, Dolgoprudny, 
Moscow Region 141701, Russia
}

\email{staskikotx@gmail.com}

\address{Department of Informatics\\
King's College London\\
Strand Campus, Bush House, 30 Aldwych, London  WC2B 4BG, U.K.}
\email{agi.kurucz@kcl.ac.uk}

\address{Faculty of Economics\\
Kyushu Sangyo University\\
2-3-1 Matsukadai, Higashi-ku, 
Fukuoka 813-8503, Japan}
\email{ytanaka@ip.kyusan-u.ac.jp}

\address{Department of Computer Science\\
University of Liverpool\\
Ashton Building,
Ashton Street, Liverpool L69 3BX, U.K.}
\email{wolter@liverpool.ac.uk}

\address{Department of Computer Science and Information Systems\\
Birkbeck, University of London\\
Malet Street, London WC1E 7HX, U.K.}
\email{michael@dcs.bbk.ac.uk}

\date{ }

\begin{document}
\maketitle

\begin{abstract}
Our concern is the completeness problem for \sptheories, that is, sets of implications between strictly positive 
formulas built from propositional variables, conjunction and modal diamond operators. Originated in logic, algebra and computer science, \mbox{\sptheories{}} have two natural semantics: meet-semilattices with monotone operators providing Birkhoff-style calculi, and first-order relational structures (aka Kripke frames) often used as the intended structures in applications. Here we lay foundations for a completeness theory that aims to answer the question whether the two semantics define the same consequence relations for a given \sptheory.
\end{abstract}


In this paper, we investigate connections between various consequence relations for the fragment of propositional multi-modal logic that comprises implications $\sigma \rightarrow \tau$, where $\sigma$ and $\tau$ are \emph{strictly positive modal formulas}~\cite{Beklemishev14} constructed from propositional variables using conjunction $\land$, unary diamond operators $\d_i$, and the constant `truth' $\top$. We call such formulas $\sigma$ and $\tau$ \emph{sp-formulas} and implications between them \emph{sp-implications\/}.

\section{Background}\label{bg}
Consequence relations for \spequations{} have been studied in knowledge representation, universal algebra, and modal provability logic.


\subsection{Description logic $\EL$}\label{el}
In knowledge representation, ontologies are used to define vocabularies for domains of interest together with logical relationships between the vocabulary terms \cite{DBLP:conf/dlog/2003handbook,DBLP:conf/wollic/LutzW09,newTextBook}. The description logic 
$\EL$~\cite{BaaderKuesters+-IJCAI-1999,BaaderBrandtLutz-IJCAI-05} is a widely used ontology language, in which such relationships are given by means of 
 (notational variants of) \spequations. A typical example of an $\EL$ 
ontology is SNOMED CT \cite{snomed} that provides a standardised medical vocabulary for the healthcare systems of more than twenty countries. SNOMED CT consists of about 300,000 \spequations{} covering most aspects of medicine and healthcare. For example, the \spequation{}
$$
\textit{Viral\_pneumonia} \rightarrow 
\Diamond_{\textit{causative\_agent}}\textit{Virus} \wedge \Diamond_{\textit{finding\_site}} \textit{Lung}
$$
says that viral pneumonia is caused by a virus and found in lungs.  
$\EL$ is the logical underpinning of the profile \OWLEL{} of the Web Ontology Language
 \OWL{} \cite{owl} designed by W3C for writing up ontologies. Under the $\EL$ semantics, \spequations{} are interpreted in relational structures known as Kripke frames in modal logic. Important
 reasoning problems are whether an \spequation{} is valid under this semantics and, more generally,
whether it follows from a finite set of \spequations{}. The former is called the subsumption problem, its generalisation is
the subsumption problem relative to a TBox. In modal logic, they correspond to the local and, respectively, global consequence relation
(restricted to \spequations). 
The computational complexity of these problems has been extensively studied.
Both were shown to be {\sc PTime}-complete in general~\cite{BaaderKuesters+-IJCAI-1999,BaaderBrandtLutz-IJCAI-05} as well as under additional relational constraints and extensions to the language~\cite{BaaderBrandtLutz-IJCAI-05,stokkermans2008}, for example, over transitive Kripke frames and, more generally, frames satisfying implications of the form $R_{1}\circ \dots \circ R_{n} \subseteq R$, for binary relations $R_1,\dots,R_n,R$.
{\sc PTime}/{\sc coNP} dichotomy results for the subsumption problem under
some universally first-order definable relational constraints were obtained  in~\cite{islands10}, while 
\cite{Baader03} gave an example of a constraint under which subsumption becomes undecidable.


\subsection{Semilattices with monotone operators}

Following the algebraic approach to giving  semantics to propositional logics~\cite{Rasiowa&Sikorski63}, we can regard strictly positive modal formulas as terms of the algebraic language with a binary function $\land$, unary functions $\d_i$ and constant $\top$. If $\land$ is a semilattice operation, then an \spequation{} $\sigma\imp\tau$ becomes an `inequality' of the form $\sigma\le\tau$, which is equationally expressible as $\sigma\land\tau\approx\sigma$. Conversely, any algebraic equation $\sigma\approx\tau$ between strictly positive `terms' is equivalent to the pair $\sigma\imp\tau$ and $\tau\imp\sigma$ of \spequations. Thus, semilattices with additional operators provide another natural semantics for  
\spequations.

Semilattices with operators have been studied in universal algebra. An important example
is their use in McKenzie's undecidability proof for Tarski's finite basis problem 
\cite{doi:10.1142/S0218196796000040}.
There has been extensive research on generalising
natural dualities for algebras with various kinds of (semi)lattice reducts to algebras with   operators 
\cite{Priestley70,Urquhart78,Goldblatt89,Allwein-Dunn93,gehrke-jonsson1994,Hartonas-Dunn97,ghilardi-meloni1997,gehrke-jonsson2000,Sofronie-Stokkermans00a,gehrke-harding2001,gehrke-jonsson2004,Davey2007}.

The relational semantics for the description logic $\EL$ mentioned above has been
connected to the uniform word problem (aka quasiequational theory) of varieties of \emph{semilattices with monotone}%
\footnote{A unary operator $\d_i$ in an algebra $\A$ is called \emph{monotone} if $\A$ validates
$\d_i(p\land q)\le\d_i q$. This is the same as to say that $a\le b$ implies 
$\d_i a\le \d_i b$, for any $a,b$ in $\A$.}
\emph{unary operators\/} (\SLOa s, for short) in \cite{Sofronie-Stokkermans01,stokkermans2008}.
Varieties of closure semilattices, that is, \SLOa s with a single operator $\d$ validating $p\le\d p$ and 
$\d\d p\le\d p$, have been investigated in \cite{jackson2004}.  They are also connected
to the closure algebras of McKinsey and Tarski \cite{McKinsey&Tarski44}.


\subsection{Sub-propositional modal logics and Reflection Calculus $\textsf{RC}$}

Sp-impli\-cations have also been investigated in the context of provability logic~ \cite{Beklemishev12,dashkov2012positive,Beklemishev14,Beklemishev15b,Beklemishev2017}.
The main motivation for considering them was the observation that, while syntactical modal reasoning in Japaridze's multi-modal provability logic $\mf{GLP}$~\cite{jap,boolos1995logic} cannot be characterised by any class of Kripke frames, its restriction $\textsf{RC}$ 
to \spequations{} does have such a characterisation~\cite{dashkov2012positive}. In particular, \spequations{} are regarded in $\textsf{RC}$ as sequents connecting two strictly positive formulas, and the developed syntactic calculus mimics the algebraic \SLOa-axioms and the axioms and rules of Birkhoff's equational
calculus \cite{Birkhoff35} (see \S\ref{algcalc} below).
Note also that $\textsf{RC}$ allows more general arithmetic interpretations than $\mf{GLP}$ \cite{Beklemishev14} and, similarly to the subsumption problem in $\EL$, reasoning in $\textsf{RC}$ is {\sc PTime}-complete~\cite{dashkov2012positive} (whereas $\mf{GLP}$ is {\sc PSpace}-complete~\cite{DBLP:conf/aiml/Shapirovsky08}).

Other sub-propositional fragments of full modal logic that contain \spterms{} 
have also been 
considered in the literature, both in the modal and description logic setting and under various relational constraints. For example, results on the computational complexity of the fragment with formulas built from literals using $\land$ and both diamond and box modalities can be found in~\cite{sss91,dhlnws92,HemaspaandraS01}.
The above mentioned dualities have also been investigated
from the modal logic perspective 
in order to find extensions of Kripke semantics that match the corresponding 
algebraic semantics; see~\cite{Dunn95,Celani-Jansana97,Celani-Jansana99,Sofronie-Stokkermans00b} for the negation-free fragment and~\cite{gehrke-nagahashi-venema2005} for its extension with $\land/\lor$-swapping operators.


In this paper, our concern is somewhat `orthogonal' to duality theory: instead of modifying/extending the relational semantics to `match' it with the algebraic one, we aim to understand the relationship between the (often intended) relational and (syntactic) algebraic consequence relations for \spequations.

\section{Research problems and results}\label{rproblem}
Following the modal logic tradition, we define the \emph{\sptheory{} axiomatised by a set $\Ec$ of spi-implications}
as the closure of $\Ec$ under the axioms and rules of a syntactic calculus capturing the algebraic semantics of \spequations. 
We denote this logic by $L = \SPi + \Ec$, indicating that $\SPi$ comprises the \spequations{} that are valid in all SLOs.
%
%
%
%

Our primary concern is the (Kripke) completeness problem for spi-logics. More precisely, we would like to
\begin{description}
\item[(\emph{completeness})]
identify \sptheories{} $\mathsf{SPi} + \Ec$ that are \emph{complete} in the sense that the 
two consequence relations $\Ec\models_{\CA}$ and $\Ec\models_{\SLO}$ coincide, where
for any \spequation{} $\e$,
\begin{align*}
\Ec\models_{\CA}\e &\quad\mbox{iff}\quad\mbox{$\e$ is valid in every Kripke frame validating $\Ec$;}\\
\Ec\models_{\SLO}\e &\quad\mbox{iff}\quad\mbox{$\e$ is valid in every \SLOa\ validating $\Ec$.}
\end{align*}
\end{description}
\spequationscap{} are modal Sahlqvist formulas~\cite{Sahlqvist75}. So, by the completeness part 
of Sahlqvist's theorem, 
the full Boolean normal modal logic $\mf{K} \oplus \Ec$ axiomatised (using the standard   calculus of normal modal logic%
\footnote{It has the modal axioms $\Box_i(\varphi\to\psi)\to(\Box_i\varphi\to\Box_i\psi)$ 
and the rules of substitution, modus ponens
and necessitation $\varphi/\Box_i\varphi$, for each modal operator $\Box_i$.}%
) by the \spequations{} in $\Ec$ is Kripke complete, that is, for every modal formula $\varphi$,
\begin{equation}\label{rq}
\Ec\models_{\CA}\varphi
 \ \ \mbox{iff}\ \ \varphi\in \mf{K} \oplus \Ec
\ \ \mbox{iff}\ \ \mbox{$\varphi\approx\top$ is valid in every \BAOa\ validating $\Ec$,}
\end{equation}
where \BAOa\ stands for \emph{Boolean algebra with normal and $\vee$-additive unary operators}%
\footnote{A \BAOa\ is an
algebra of the form $\A =(A,\land,\lor,-,\bot,\top,\d_i)_{i\in I}$, where $(A,\land,\lor,-,\bot,\top)$
is a Boolean algebra, $\d_i\bot=\bot$ and $\d_i(a\lor b)=\d_i a\lor \d_i b$, for all $a,b\in A$ and $i\in I$.}%
~\cite{Jonsson&Tarski51}.
Note that, by \eqref{rq}, the completeness problem is equivalent to 
%
\begin{description}
\item[(\emph{\SP-axiomatisability})]
the problem whether $\Ec$ \emph{\SP-axiomatises} the \SP-fragment
of the modal logic $\mathsf{K} \oplus \Ec$, that is,
 $\e\in  \SPi + \Ec$ iff $\e\in  \mf{K} \oplus \Ec$, for any \spequation{} $\e$
(in other words, the problem whether the \sptheory{} $\SPi+\Ec$ has a \emph{modal companion}  \cite{Beklemishev15b}); and also to

\item[(\emph{conservativity})]
the purely algebraic problem of whether the consequence relation 
$\Ec\models_{\BAO}$ is \emph{conservative} over $\Ec\models_{\SLO}$ with respect to algebraic equations between \spterms,
that is, $\Ec\models_{\SLO} \sigma\approx\tau$ iff $\Ec\models_{\BAO} \sigma\approx\tau$, for any 
\spterms{} $\sigma$ and $\tau$. 
\end{description}
In Boolean modal logic, the completeness problem has been actively and thoroughly investigated since the invention of the Kripke semantics in the 1950--60s. Nearly all standard modal logics were proved to be Kripke complete by showing that they either are canonical or have the finite model property, and it took a while to \emph{construct} first examples of incomplete logics~\cite{Fine74,Thomason74}. 
In contrast, incomplete \sptheories{} are easy to find, with two simplest ones being $\SPi + \{\Diamond p \imp p\}$ and $\SPi + \{\d p\imp\d q\}$ (Examples~\ref{e:simple} and \ref{e:simple2}). It is readily seen that both of them  have the finite model (but not  finite frame) property. By Sahlqvist's theorem, all Boolean modal logics with \spequation al axioms are canonical. Thus, the classical completeness theory appears to be of little help in understanding completeness of \sptheories. New tools and techniques are required to investigate this phenomenon. 

In this paper, we develop and apply two general methods for establishing completeness of \sptheories.  

The first
one is based on the fact that an \sptheory{} $\SPL$ 
 is complete whenever every \SLOa\ validating $\SPL$ can be embedded into the (\SLOa-reduct of the) full complex algebra of some Kripke frame for $\SPL$. Following the terminology of Goldblatt \cite{Goldblatt89}, we call such \sptheories{} $\SPL$ \emph{\complex\/}. 
Proving that $\SPL$ is \complex\ can be regarded as a generalisation of the canonical model technique from modal logic: for every \BAOa\ $\A$ validating an \sptheory{} $\SPL$, its ultrafilter-frame $\A_+$ validates $\SPL$ as well.
Unfortunately, no such `canonical' Kripke frame construction is available for \SLOa s. 
Instead, we suggest two `templates' that provide a range of embeddings of \SLOa s into the \SLOa-reducts of complex algebras of appropriate frames, one generalising the embedding of~\cite{jackson2004}, and another one using filters in \SLOa s (see \S\ref{embed}). We employ these templates to obtain two general sufficient conditions for complexity 
(and so completeness) of \sptheories{} (Theorems~\ref{Th:6} and \ref{t:exist}), 
and also show complexity of numerous concrete \sptheories{} defining familiar classes of Kripke frames. 
Our conditions cover earlier results of
Sofronie-Stokkermans~\cite{Sofronie-Stokkermans01,stokkermans2008} who proved that \spequations{} of the form \mbox{$\d_1\dots\d_n p\imp\d_0 p$} axiomatise \complex{}   
\sptheories, 
and those of Jackson ~\cite{jackson2004} who showed that the \sptheory{} $\Esfour=\SPi+\{p\imp\d p,\d\d p\imp\d p\}$ (whose axioms 
$\Axsfour=\{p\imp\d p,\d\d p\imp\d p\}$ 
define the class of all quasiorders---frames of the modal logic $\mf{S4}$) is \complex. 
We delimit the scope of the method by providing many examples of incomplete \sptheories, in particular, pairs of complete and incomplete \sptheories{} sharing the same Kripke frames, and develop a general technique for constructing incomplete \sptheories{} (Theorem~\ref{thm:genericinc}).

As mentioned above, Boolean modal logics with \spequation al axioms are always complex. In contrast, we show a few natural and simple \spequations{} that axiomatise complete but not complex \sptheories, for example, those 
expressing $n$-functionality, for $n\geq 2$, and linearity (Theorems~\ref{t:altncomplex} and \ref{t:notdlocons}).
For such \sptheories, we develop another general technique, called the method of \emph{\method}, that mimics Kripke frame reasoning with the help of the syntactic  
Birkhoff-type calculus for \SLOa s (see \S\ref{sec:interp}).
We use this method to prove one more general sufficient condition for completeness (Theorem~\ref{thm:interpolant}) and apply it to a number of concrete \sptheories{} that are not \complex\ (Theorems~\ref{t:altn}, \ref{t:s5alt}, \ref{t:lincomp}).
Syntactic proxies can also be used to establish completeness of all but two
proper extensions of the \sptheory{} $\Esfive=\SPi + \{p\imp\d p,\d\d p\imp\d p,q\land\d p\imp\d(p\land\d q)\}$\ (whose axioms define the class of all equivalence relations---frames of the modal logic $\mf{S5}$), the two exceptions being in fact incomplete. 
%
Jackson~\cite{jackson2004} fully described the lattice of extensions of $\Esfive$; it follows from his proofs that most of them are $\models_{\BAO}$--to--$\models_{\SLO}$ conservative.



One feature that \sptheories{} do share with Boolean modal logics is that---apart from a few simple cases (such as extensions of $\Esfive$ and $\mf{S5}$)---complete and effective classifications of logics according to their non-trivial properties are hardly possible. 
In \S\ref{sec:undec}, we prove by reduction of the halting problem for Turing machines 
that, given a finite set $\Ec$ of \spequations, no algorithm can recognise \eqcomp ness 
or complexity of the \sptheory{} $\SPi+\Ec$. 
The proof is more direct compared to the known constructions from modal logic~\cite{Thomason82,DBLP:journals/jsyml/ChagrovZ93,DBLP:conf/aiml/ChagrovC06} because  very simple incomplete \sptheories{} are available.

Having laid foundations for a completeness theory in the strictly positive context, we are naturally 
interested in the byproducts it may have for two related problems, viz., the computational complexity (in particular, decidability) of \sptheories{} and the definability problem. Recall that tractability of reasoning was one of the main motivations for considering \sptheories. 

As far as \textbf{\emph{computational complexity}} is concerned, we observe that \sptheories{} with universally definable classes of Kripke frames have the polynomial finite frame property%
\footnote{An \sptheory{} $\SPL$ has the \emph{polynomial} \emph{finite frame property} if every 
\spequation{} $\e$ 
that fails in some frame for $\SPL$ also
fails in a frame for $\SPL$ of polynomial size in $\e$.}
%
and are decidable in {\sc coNP} if finitely axiomatisable
and \eqcomp{} (Theorem~\ref{t:subframedec}); moreover, those complete ones whose frames are definable by equality-free universal Horn sentences are actually tractable (Theorem~\ref{thm:horncomplete}). The latter applies to the \sptheories{} in the scope of completeness Theorems~\ref{Th:6}, \ref{thm:interpolant} and \ref{t:comp}.
(Note that Boolean modal logics axiomatised by the same \spequations{} can be computationally very complex, even undecidable~\cite{KikotSZ14}). 
We also show tractability of several finitely axiomatisable complete \sptheories{} defining universal \emph{non}-Horn frame conditions 
such as the \sptheory{} $\Efunn$ 
whose frames 
are equivalence relations with classes of size $\leq n$, for $n\geq 2$ (Theorem~\ref{t:s5altP}), and 
the \SP-fragment $\Elin$ of the modal logic $\mf{S4.3}$ (Theorem~\ref{t:linP}). 
On the other hand, we observe that the completeness criterion of Theorem~\ref{t:exist} has 
the \SP-fragments of all modal grammar logics~\cite{farinasdelCerroPenttonen88} in its scope, and so there exist
finitely axiomatisable and undecidable 
complete \sptheories~\cite{Tseitin58,Shehtman82,Chagrov-Shehtman95,Baader03,Beklemishev15b}.

A class $\Cc$ of Kripke frames is called \textbf{\emph{\SP-definable}} 
if $\Cc=\{\F \mid \F \models \Ec\}$ for some set $\Ec$ of \spequations.
%
The correspondence part of Sahlqvist's theorem~\cite{Sahlqvist75} says that \SP-definability (unlike modal definability) always implies definability by first-order $\forall\exists$-sentences. Many standard properties of frames turn out to be \SP-definable (see Table~\ref{t:SP-logics}). On the other hand, such well-known logics as ${\sf K4.1}$, ${\sf K4.2}$ and ${\sf K4.3}$ are typical examples of Kripke 
complete modal logics whose frames are not \SP-definable (see Table~\ref{t:SP-logicsundef}). To obtain such non-\SP-definability results, we give a general necessary condition for \SP-definability (in  \S\ref{eldef}), and also show that \SP-definable properties of quasiorders must be universal.


The remainder of the article is organised as follows. Having defined in \S\ref{pr:comp} the required basic notions, in \S\ref{sec:tools} we introduce the two general methods for establishing  \eqcomp ness, which are applied in 
\S\S\ref{Sahlqvist}--\ref{sec:disj} and complemented by multiple examples of incomplete \sptheories.
We systematise our completeness results for \sptheories{} according to the form of the first-order correspondents of their axioms: \spequations{} with universal Horn, existential and disjunctive correspondents are discussed in 
\S\ref{Sahlqvist}, \S\ref{sec:exist} and \S\ref{sec:disj}, respectively. 
In \S\ref{sec:undec} we prove that it is undecidable whether a given finite set of \spequations{} axiomatises a \eqcomp\ or \complex{} \sptheory.
A few related problems are briefly discussed in \S\ref{sec:other}: 
in \S\ref{eldef} we deal with non-\SP-definability; 
in \S\ref{bot} we consider \SPb-implications that may also contain the constant $\bot$ standing for `falsehood'   in Kripke frames and for the $\le$-smallest element in \SLOa s; 
in \S\ref{qeqcomp} we have a brief look at \emph{\qsptheories}
(quasiequational theories in the algebraic setting). 
In particular, we characterise  \complex{} \qsptheories{} $\Qc$ as those for which $\Qc\models_{\CA}\q$ coincides with $\Qc\models_{\SLO}\q$, for all \spquasiequations{} $\q$.
Finally, in \S\ref{concl} we suggest further research directions; a few open questions are also scattered throughout the paper. 
%
%
\begin{table}[th]
\caption{\SPcap-definable first-order properties.}\label{t:SP-logics}
\centering 
\begin{tabular}{|l|l|c|}\hline
first-order property & \spequation(s) & notation\\\hline\hline
reflexivity & $p \imp \d p$ & $\erefl$\\\hline
transitivity & $\d\d p \imp \d p$ & $\etrans$\\\hline
symmetry & $q \land \d p \imp \d (p \land \d q)$ & $\esym$\\\hline
$\forall x,y,z\,\bigl(R(x,y) \land R(x,z) \to R(y,z)\bigr)$ & $ \d p \land \d q \imp \d (p \land \d q)$ & $\eeuc$\\
Euclideanness & & \\\hline
quasiorder & $\{\erefl,\etrans\}$ & $\Axsfour$\\\hline
equivalence & $\{\erefl,\etrans,\esym\}$ & $\Axsfive$\\
& $\{\erefl,\etrans,\eeuc\}$ & $\Axsfive'$\\\hline
$\forall x,y,z\,\big[R(x,y)\land R(x,z)\to$ & $\d(p\land q)\land \d(p\land r)\imp $ & $\ewc$\\[-2pt]
$\ \ \bigl(R(y,y)\land R(y,z)\bigr)\lor \bigl(R(z,z)\land R(z,y)\bigr)\big]$ & $\qquad\quad\d (p\land \d q\land \d r)$ &\\\hline
linear quasiorder\footnotemark & $\{\erefl,\etrans,\ewc\}$ & $\Axlin$ \\\hline
$\forall x,y\, \big[R(x,y) \to \exists z \, \bigl(R(x,z) \land R(z,y)\bigr)\big]$ & $\d p\imp\d\d p$ & $\eden$ \\
density & & \\\hline
$\forall x,y,z\,\bigl(R(x,y) \land R(x,z) \to (y=z)\bigr)$ & $\d p \land \d q \imp \d(p\land q)$ & $\efun$ \\
functionality & & \\\hline
\end{tabular}
\end{table}%

\begin{table}[th]
\caption{Non-\SP-definable but modally definable first-order properties.}\label{t:SP-logicsundef}
\centering 
\begin{tabular}{|l|l|c|}\hline
first-order property & modal formula(s) & notation\\\hline\hline
$\forall x,y,z\,\bigl(R(x,y)\land R(y,z)\to $ &  &  \\
\hspace*{2.3cm}$R(x,z)\lor (x=z)\bigr)$ & $\d\d p \imp p \vee \d p$ & $\varphi_{\it ptrans}$ \\
pseudo-transitivity & & \\\hline
pseudo-equivalence & $\e_{\it sym},\ \varphi_{\it ptrans}$ & $\mf{Diff}$\\\hline
weak connectedness$^5$  & $\d p \land \d q \imp \d (p \land q) \,\lor$ & $\varphi_{\it wcon}$ \\
& $\   \d (p \land \d q) \lor \d (q \land \d p)$ & \\\hline
transitivity and weak connectedness  & $\e_{\it trans},\ \varphi_{\it wcon}$ & $\mf{K4.3}$\\\hline
$\forall x,y,z\,\bigl(R(x,y)\land R(x,z)\to $ &  & \\
\hspace*{1.8cm}$\exists u\,\bigl(R(y,u)\land R(z,u)\bigr)\bigr)$ & $\d\Box p\imp\Box\d p$ & $\varphi_{\it conf}$ \\
confluence & & \\\hline
transitivity and confluence  & $\e_{\it trans},\ \varphi_{\it conf}$ & $\mf{K4.2}$\\\hline
transitivity and & $\e_{\it trans},\ \Box\d p\imp\d\Box p$ & $\mf{K4.1}$\\
$\forall x\, \exists y\bigl(R(x,y)\land \forall z\,\bigl( R(y,z)\to (y=z)\bigr)\bigr)$ &  & \\\hline
\end{tabular}
\end{table}%
\footnotetext{A reflexive and transitive relation $R$ is called a \emph{linear quasiorder\/} if
$R$ is \emph{weakly connected\/}:
$\forall x,y,z\; \bigl(R(x,y)\land R(x,z)\to R(y,z)\lor R(z,y)\lor (y=z)\bigr)$.
Linear quasiorders are the frames of the modal logic $\mf{S4.3}$.}


\section{Preliminaries}\label{pr:comp}

We begin by giving definitions of the basic notions and discussing the problems we deal with in this paper.

\subsection{\sptermscap{} and \spequations} 

Let $\R$ be 
a non-empty set called a \emph{signature}. An \emph{\spterm} (\emph{of signature\/} $\R$) 
is a multi-modal formula constructed from 
propositional variables $p$ from some countably infinite set $\varset$ and constant $\top$ using 
conjunction $\land$ and unary diamond operators $\dR$, for \mbox{$R \in \R$}.
We omit the subscript $R$ in the unimodal case $\R=\{R\}$.  

An \emph{\spequation\/} $\e$ (\emph{of signature\/} $\R$) 
 is an expression of the form $\sigma \imp \tau$, where $\sigma $ and $\tau$ are \spterms{} of signature $\R$. 
%


As argued in \S\S\ref{bg}--\ref{rproblem}, we aim to connect two types of semantics for \spequations: one
based on first-order relational structures, known as Kripke frames in modal logic, and an algebraic one, based on meet-semilattices with monotone operators. We begin with the latter.


\subsection{Algebraic semantics}\label{algsem}

A structure $\A =(A,\land,\top,\dR)_{R\in \R}$ is an \emph{\sptype{} algebra} (\emph{of signature} $\R$) if $A \ne \emptyset$, 
$\top\in A$, $\land$ is a binary and each $\dR$ a unary function (operator) on $A$. 
This way \spterms{} can be regarded as algebraic \emph{\sptype{} terms\/}. (The overloading of $\land$, $\top$ and $\dR$ should not confuse the reader as it will always be clear from context whether we deal with algebraic operations or logic connectives.) 
An \emph{\sptype{} equation} is of the form $\sigma\approx\tau$, where $\sigma$ and $\tau$ are \sptype{} terms
(that is, \spterms).
A \emph{valuation in} $\Aa$ is a function $\vala$ mapping the variables $p\in\varset$ to elements in $A$. The \emph{value} $\tau[\vala]\in A$ of an \sptype{} term $\tau$ under $\vala$ is defined inductively as usual. If the variables occurring in $\tau$ are among $p_1,\dots,p_n$ and $\vala(p_i)=a_i$, then we also write $\tau[a_1,\dots,a_n]$ in place of $\tau[\vala]$. Given an \sptype{} equation $\sigma\approx\tau$,
we set $\Aa\models(\sigma\approx\tau)[\vala]$ if $\sigma[\vala] = \tau[\vala]$, and $\Aa\models(\sigma\approx\tau)$ if $\Aa\models(\sigma\approx\tau)[\vala]$ for every valuation $\vala$ in $\Aa$, in which case we say that $\Aa$ \emph{validates} $\sigma\approx\tau$. 

A \emph{meet-semilattice with monotone operators} (\SLOa, for short) is an \sptype{} algebra 
validating the following \sptype{} equations:
\begin{align}
\label{idem}
& p\land p\approx p,\\
& p\land q\approx q\land p,\\
\label{assoc}
& p\land (q\land r)\approx (p\land q)\land r,\\
& p\land \top\approx p,\\
\label{dmon}
& \dR(p\land q)\land \dR q\approx \dR(p\land q), \quad\mbox{for $R\in\R$}.
\end{align}
In a \SLOa{} $\A$, the partial order $\le$ is defined as usual by taking $a\le b$ iff 
 $a\land b=a$, for all $a,b$ in $\A$.
 It is readily seen that $\land$ and $\dR$ are 
 \emph{monotone} with respect to $\leq$: if $a\leq b$ then $a\land c\leq b\land c$ and $\dR a\leq \dR b$, for
 all $a,b,c$ in $\A$ and $R\in \R$.
 
 By regarding any \spequation{} $\e = (\sigma\imp\tau)$ as an \sptype{} `inequality' $\sigma\leq\tau$ (which is a shorthand for the \sptype{} equation $\sigma\land\tau\approx\sigma$), we set $\Aa\models\e[\vala]$ if $\sigma[\vala] \leq \tau[\vala]$, and $\Aa\models\e$ if $\Aa\models\e[\vala]$ for every valuation $\vala$ in $\Aa$, in which case we say that $\Aa$ \emph{validates} $\e$. The set of \spequations{} 
 that are validated by all SLOs is denoted by $\SPi$. 
 
We say that a SLO $\Aa$ \emph{validates}  a set $\Ec$ of \spequations{} and write $\Aa\models\Ec$ if  
$\Aa\models\e$ for all $\e$ in $\Ec$. 
%
%
We denote by $\SLO_{\Ec}$ the class---in fact, variety---of all \SLOa s 
validating $\Ec$. 
%
In particular, \SLO{} denotes the variety of all \SLOa{}s.
We define a consequence relation $\Ec\models_{\SLO}$ by taking, for any \spequation{} $\e$,   
%
\[
\Ec\models_{\SLO}\e \quad\mbox{iff}\quad\mbox{$\A\models\e$\ \ for every  $\Aa\in\SLO_{\Ec}$.}
\]
We write $\models_{\SLO}\e$ for $\emptyset\models_{\SLO}\e$.
As a \SLOa{}  clearly validates $\sigma\approx\tau$ iff 
it validates both $\sigma\to\tau$ and $\tau\to\sigma$, we write
 $\Ec\models_{\SLO}\sigma\approx\tau$ whenever both $\Ec\models_{\SLO}\sigma\imp\tau$
 and $\Ec\models_{\SLO}\tau\imp\sigma$ hold.


\subsection{Spi-logics}\label{algcalc}

As \spequations{} are special cases of algebraic \sptype{} equations, the consequence relation $\Ec\models_{\SLO}$
can be characterised syntactically by Birkhoff's equational calculus \cite{Birkhoff35,Graetzer79}.
Using a Lindenbaum--Tarski-algebra type argument, it is readily seen  that $\Ec\models_{\SLO}$
can also be captured by a calculus using only \spequations{} in its derivations.
Namely, it is not hard to show that
\begin{equation}\label{birkhoff}
\Ec\models_{\SLO}\e  \quad\mbox{iff}\quad\Ec\vdash_{\SLO}\e,
\end{equation}
where $\Ec\vdash_{\SLO}\e$ means that there is a finite sequence $\e_0,\dots,\e_n$ of 
\spequations{} such that $\e_n=\e$ and each $\e_i$, for $i\le n$, is  either 
a substitution instance of some \spequation{} in $\Ec$ or a substitution instance of one of the axioms
\begin{equation}\label{axioms}
p\imp p,\qquad
p\imp \top,\qquad
p\land q\imp q\land p,\qquad 
p\land q\imp p,
\end{equation}
or obtained from earlier members of the sequence using one of the rules 
\begin{equation}\label{rules}
\frac{\sigma \imp \tau \ \ \tau \imp \varrho}{\sigma \imp \varrho}, \qquad 
\frac{\sigma \imp \tau \ \ \sigma \imp \varrho}{\sigma \imp \tau \land \varrho}, \qquad 
\frac{\sigma \imp \tau}{\dR\sigma \imp \dR \tau} \ \ \ (R\in\R)
\end{equation}
(see also the Reflection Calculus $\textbf{RC}$ of \cite{Beklemishev12,dashkov2012positive}).
In fact, throughout we shall only use the $\Leftarrow$ (soundness) direction of \eqref{birkhoff}.
We write $\vdash_{\SLO}\e$ for $\emptyset\vdash_{\SLO}\e$.
We write
 $\Ec\vdash_{\SLO}\sigma\approx\tau$ whenever both $\Ec\vdash_{\SLO}\sigma\imp\tau$
 and $\Ec\vdash_{\SLO}\tau\imp\sigma$ hold.

%
%
%
%

For any set $\Ec$ of \spequations,
we define the \emph{\sptheory{} $\SPi + \Ec$ axiomatised by $\Ec$} as 
%
\[
\SPi + \Ec ~=~ \{\e\mid \mbox{$\e$ is an \spequation{} and }\Ec\vdash_{\SLO}\e\}.
\]
If $\SPL=\SPi + \Ec$, for some set $\Ec$ of \spequations, then we call $\SPL$ an \emph{\sptheory}.


%
%


\subsection{Kripke semantics}

A \emph{Kripke model} (\emph{of signature\/} $\R$) 
is a pair of the form $\M = (\F,\V)$, where $\F=(W,R^\F)_{R\in \R}$ 
 is a \emph{frame} (\emph{of signature\/} $\R$) 
 with domain $W \ne\emptyset$ and binary (accessibility) relations $R^\F$, for $R \in \R$, and $\V$ is a \emph{valuation} associating a subset $\V(p) \subseteq W$ with any variable $p$. 
 The 
 truth relation $\M,w\kmodels\tau$ for $w\in W$ and an \spterm{} $\tau$ is defined by induction: 
 $\M,w\kmodels\top$,
 $\M,w\kmodels p$ iff $w\in\V(p)$,
 $\M,w\kmodels\tau' \land \tau''$ iff $\M,w\kmodels\tau'$ and $\M,w\kmodels\tau''$,
 and for each $R\in\R$,
%
\[
\M,w,\kmodels\dR \tau'\quad\mbox{iff}\quad\M,w'\models\tau'\mbox{ for some $w'$ with $(w,w')\in R^\F$.}
\]
For an \spequation{} $\e = (\sigma \imp \tau)$ and $w\in W$, we write $\M,w\models\e$ if 
$\M,w\kmodels\sigma$  implies $\M,w\kmodels\tau$. 
We say that $\e$ \emph{holds} in $\M$ (or $\M$ is a \emph{model of} $\e$) and write $\M \models \e$, if 
 $\M,w\models\e$ holds for every $w\in W$. 
 We also write $\F,w\models\e$ if $\M,w\models\e$ holds for every Kripke model $\M$ based on $\F$, and $\F\models\e$ if $\F,w\models\e$ for every $w \in W$ (equivalently, if $\M\models\e$ for every model $\M$ based on $\F$); in this case, we say that $\F$ \emph{validates} $\e$. Finally, we say that $\F$ \emph{validates} (or is a \emph{frame for}) a
 set $\Ec$ of \spequations{}
and write $\F\models\Ec$, if 
$\F\models\e$ for every $\e$ in $\Ec$. 
The class of frames for $\Ec$ is denoted by $\CA_{\Ec}$.
By the correspondence part of Sahlqvist's theorem, $\CA_{\Ec}$ 
is first-order definable in the language with binary predicate symbols $R$, for $R\in\R$, and equality. Any such first-order theory  defining $\CA_{\Ec}$ is called a \emph{correspondent of\/} $\Ec$; see, e.g., \cite{Blackburnetal01,Chagrov&Z97}. (All correspondents of $\Ec$ are equivalent.) If $\{\Psi\}$ is a correspondent of $\{\e\}$,  we say that $\Psi$ is a \emph{correspondent of\/} $\e$.

Given 
a set $\Ec$ of \spequations,
we define a consequence relation $\Ec\models_{\CA}$ by taking, for any \spequation{} $\e$, 
%
%
%
%
\[
\Ec\models_{\CA}\e \quad\mbox{iff}\quad\mbox{$\F\models\e$ for every frame $\F\in\CA_{\Ec}$.}
\]
We write $\models_{\CA}\e$ for $\emptyset\models_{\CA}\e$.


\subsection{Completeness}\label{eqcompness}

Every frame $\F=(W,R^\F)_{R\in\R}$ gives rise to a \SLOa\
\[
\Fslo=(2^{W},\cap,W,\Diamond^{+}_{R})_{R\in\R},
\]
where, for all $R\in \mathcal{R}$ and $X\subseteq W$,
\[
\Diamond_{R}^{+}X = \{ w\in W \mid (w,v)\in R^\F\mbox{ for some }v\in X\}
\]
(that is, $\Fslo$ is the \sptype{} reduct of the \emph{full complex algebra of} $\F$ \cite{Goldblatt89}).
As Kripke models over $\F$ and valuations in $\Fslo$ are the same thing, for every \spequation{}
$\e$, we have $\F\models\e$ iff $\Fslo\models\e$. Therefore, for every \sptheory{} $\SPi + \Ec$, 
\begin{equation}\label{soundness}
\Ec\models_{\SLO}\e \quad \Longrightarrow \quad \Ec\models_{\CA}\e,\quad
\mbox{for any }\e,
\end{equation}
%
and so, by \eqref{birkhoff},
\begin{equation}\label{soundness2}
\CA_{\Ec}=\CA_{\SPi+\Ec}.
\end{equation}
%
An \sptheory{} $\SPL=\SPi + \Ec$ 
is called \emph{\eqcomp} if, for every \spequation{} $\e$, 
%
\[
\Ec\models_{\CA}\e  \quad\mbox{iff}\quad\Ec\models_{\SLO}\e.
\]
%
Note that \eqcomp ness of $\SPL$ does not depend on its axioms: if $\SPL=\SPi+\Ec=\SPi+\Ec'$ then $\SLO_{\SPL}=\SLO_{\Ec}=\SLO_{\Ec'}$, and so  
$\CA_{\SPL}=\CA_{\Ec}=\CA_{\Ec'}$ by \eqref{soundness2}.


As discussed in \S\ref{rproblem}, $\SPi + \emptyset$ and $\Esfour$ are simple examples of \eqcomp\ \sptheories{}
\cite{Sofronie-Stokkermans01,jackson2004} (see also Theorem~\ref{t:valid} and its proofs in
\S\ref{embed} and \S\ref{sec:interp}, and Corollary~\ref{t:ss}). 
The following two examples show incomplete ones. 

\begin{example}\label{e:simple}\em 
Consider 
 the \spequation{} $\d p\imp p$.
On the one hand, a frame $\F = (W,R^\F)$ validates $\d p\imp p$ iff $\F \models \forall x,y \, (R(x,y) \to (x = y))$. Thus, it is easy to see that $\{\d p\imp p\}\models_{\CA}\e$, where $\e = (p\land \d \top\imp \d p)$.
%
%
On the other hand, $\{\d p\imp p\} \not\models_{\SLO} \e$ as the \SLOa\ $\Aa$ with 3 elements 
$b\le a\le \top$ such that $\d a=\d b=b$ and $\d \top=a$ validates 
$\d p\imp p$ and refutes $\e$, since $a\land \d \top=a\ne b=\d a$
(see Fig.~\ref{f:firstex}~(a)). 
So, the \sptheory{} $\SPi + \{\d p\imp p\}$ is incomplete.
\end{example}

\begin{example}\label{e:simple2}\em 
Consider the 
\spequation{} $\d p\imp \d q$.
On the one hand, 
a frame $\F = (W,R^\F)$ validates $\d p\imp \d q$ iff $R^\F=\emptyset$,
and so $\{\d p\imp \d q\}\models_{\CA}\d \top\imp p$. 
On the other hand, $\{\d p\imp \d q\} \not\models_{\SLO} \d\top\imp p$ as the \SLOa\ $\Aa$ with two elements 
$a\le \top$ such that $\d a=\d \top=\top$ validates $\d p\imp \d q$ 
and refutes $\d\top\imp p$, since $\d\top=\top\not\leq a$.
Therefore, the \sptheory{} $\SPi + \{\d p\imp \d q\}$ is incomplete.
\end{example}

\subsection{Drawing \SLOa s.}\label{draw}

In our examples, depending on the context, we depict \SLOa s in two different ways. One way is to  represent the semilattice structure by its Hasse diagram and use arrows labelled by $R$ to indicate the $\dR$ functions. In the unimodal case, we represent the elements $x$ with $\d x=x$ by hollow circles, and indicate $\d$ by unlabelled arrows otherwise; see Fig.~\ref{f:firstex}~(a). 

Another way is to draw a \SLOa{} as a subalgebra $\A$ of some suitable $\Fslo$
(which always exists by Theorem~\ref{t:sloembeddable}). 
We represent the underlying $\F = (W,R^\F)_{R\in\R}$ as a labelled directed multigraph (omitting the edge labels in the unimodal case) and indicate the non-empty subsets of $W$ that belong to $\A$.
This representation makes it easier for the `modal logic minded' reader to check whether the given \SLOa{} validates an \spequation{} $\e$: it suffices to verify that $\M\models \e$ for every 
$\A$-\emph{admissible} Kripke model $\M$ based on $\F$, in which all $\M(p)$ belong to the indicated subsets of $\F$ (cf.\ \emph{general frames} in modal logic~\cite{Goldblatt76a,Chagrov&Z97}). In Fig.~\ref{f:firstex}~(b), showing such a drawing of the \SLOa{} $\A$ from Example~\ref{e:simple}, $\M\models\d p\imp p$ for all $\A$-admissible Kripke models over the depicted $\F$ (the model $(\F,\V)$ with $\V(p)=\{2\}$ is not $\A$-admissible), while $\M',1\not\models p\land \d \top\imp \d p$ for $\M'=(\F,\V')$ with $\V'(p)=\{1\}$. 
\begin{figure}[ht]
\centering
\begin{tikzpicture}[scale=.5]
\draw [very thick] (1,1) circle [radius=0.15];
\node[left] at (0.9,1) {$b$};
\draw [fill] (1,3) circle [radius=0.1];
\node[left] at (.9,3) {$a$};
\draw [fill] (1,5) circle [radius=0.1];
\node[left] at (1,5) {$\top$};
\draw [very thin] (1,1.3) -- (1,2.8);
\draw [very thin] (1,3.2) -- (1,4.7);
\draw [very thick,->] (1.2,2.9) to [out=-45, in=45] (1.3,1);
\draw [very thick,->] (1.2,5) to [out=-45, in=45] (1.2,3);
\node at (1,-.5) {(a)};
\end{tikzpicture}
\hspace*{5cm}
\begin{tikzpicture}[scale=.5]

\draw [fill=gray] (1,1.5) circle [radius=.25];
\node at (1,1.5)  {\textcolor{white}{\bf\footnotesize 1}}; 
\draw [fill=gray] (1,4) circle [radius=.25];
\draw [fill=gray] (1,4) circle [radius=.25];
\node at (1,4) {\textcolor{white}{\bf\footnotesize 2}}; 
\draw [thick,->] (1,1.9) to [out=90, in=-90] (1,3.6);
\draw[thick,dotted,rounded corners=8] (.2,2.3) -- (1.8,2.3) -- (1.8,.7) -- (.2,.7) -- cycle;
\draw[thick,dotted,rounded corners=10] (-.2,5) -- (2.2,5) -- (2.2,.3) -- (-.2,.3) -- cycle;

\node at (1,-.5) {(b)};
\end{tikzpicture}
\caption{Two ways of depicting the \SLOa{} $\A$ of Example~\ref{e:simple}.}\label{f:firstex}
\end{figure}
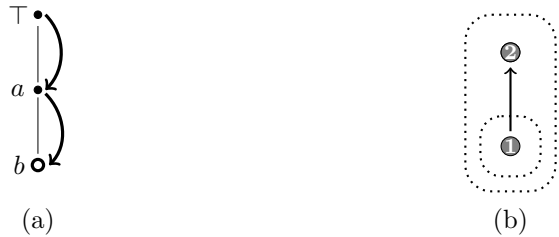


\section{Tools and techniques for proving completeness}\label{sec:tools}

In this section, we introduce 
two general methods for proving completeness of \sptheories. Both methods
will be illustrated by many examples throughout the paper.

\subsection{Embedding SLOs into complex algebras of frames}\label{embed}

Adopting the terminology of Goldblatt~\cite{Goldblatt89},
we call an \sptheory{} $\SPL$
\emph{\complex} if every $\Aa$ in $\SLO_{\SPL}$  
is embeddable%
\footnote{Given \sptype{} algebras $\Aa$ and $\Bb$ of the same signature, 
a function $\eta \colon \Aa\to\Bb$ is 
an \emph{\SPa-homomorphism} 
if it preserves all the \SPa-operations. A one-to-one \SPa-homomorphism is an \emph{\SPa-embedding\/}.
 $\Aa$ is \emph{embeddable into} $\Bb$ if there exists an  \SPa-embedding $\eta \colon \Aa\to\Bb$
(that is, if $\Aa$ is isomorphic to a subalgebra of $\Bb$). 
For universal algebra basics, we refer the reader to~\cite{Graetzer79}.} 
into $\Fslo$, 
for some frame $\F$ for $\SPL$. As \spequations{} are preserved under taking subalgebras, we always have that  
\[
\text{$\SPL$ is \complex} \quad \Longrightarrow\quad \text{$\SPL$ is \eqcomp.}
\]
Theorems~\ref{t:altncomplex} and \ref{t:notdlocons} give examples where the converse implication does not hold.

It is well-known that every \BAOa\ is embeddable into the full complex algebra of its ultrafilter frame~\cite{Jonsson&Tarski51}. As shown in~\cite{Sofronie-Stokkermans01}, 
a similar result also holds for \SLOa s:
 
\begin{theorem}\label{t:sloembeddable}
Every \SLOa{} is embeddable into $\Fslo$, for some frame $\F$.
\end{theorem}

As an immediate consequence, we obtain:

\begin{theorem}\label{t:valid} 
The \sptheory{} $\SPi + \emptyset$ 
is \complex, and so \eqcomp\textup{.} 
\end{theorem}

The simple proposition below provides us with infinitely many \complex\ \sptheories.
Call an \spequation{} $\sigma\imp\tau$ \emph{variable-free} if both $\sigma$ and $\tau$ are built up from 
$\top$ using $\land$ and the $\dR$.

\begin{proposition}\label{p:varfree}
If $\SPi + \Ec$ is a \complex\ \sptheory{} and $\Ec_0$ a set of variable-free \spequations{}, 
then $\SPi + (\Ec\cup\Ec_0)$ is \complex.
\end{proposition}
\begin{proof}
By possibly adding `dummy' \spequations{} to $\Ec$,
we may assume that every $\dR$ occurring in $\Ec_0$ also occurs in $\Ec$. 
Suppose $\Aa\in\SLO_{\Ec\cup\Ec_0}$. As $\SPi + \Ec$ is \complex, $\Aa$ is (isomorphic to) a subalgebra of $\Fslo$, for some frame $\F\models\Ec$.
As $\Aa\models\Ec_0$, there is an \mbox{$\Aa$-admissible} Kripke model $\M\models\Ec_0$ based on $\F$. Since $\Ec_0$ is variable-free, we also have $\F\models\Ec_0$. 
\end{proof}

%
%
%
%
%
%

\begin{question}\em
Does Proposition~\ref{p:varfree} hold with `\eqcomp' in place of `\complex'?
\end{question}

In the remainder of \S\ref{embed}, 
we show two different ways of proving Theorem~\ref{t:sloembeddable} and discuss connections between them. 


 \subsubsection{Embeddings via elements of SLOs.}\label{Jembed}
These are variants of the embedding used by Jackson~\cite{jackson2004} for closure algebras.
We embed a \SLOa{} $\A = (A, \land, \top,\dR)_{R\in\R}$ into the \SLOa{} $\Fslo$, for some frame  $\F=(A,R^\F)_{R\in\R}$, using the map
 \[
 \eta \colon a \mapsto \{b\in A\mid b\le a\}.
 \]
Clearly, 
$\eta(\top)=A$ and $\eta(a\land b)=\eta(a)\cap \eta(b)$.
We show now that to preserve the $\dR$,  
it is enough
if $R^\F$ satisfies the following two conditions, for all $R\in\R$:
 \begin{align}
 \label{Jemb1}
& \forall a,b\ \bigl[(a, b)\in R^\F \ \  \Rightarrow\ \  a\le\dR b\bigr],\\
 \label{Jemb2}
&  \forall a,b\ \bigl[ a\le \dR b \  \ \Rightarrow\ \ \exists c\, \bigl(c\le b \ \mbox{ and }\  (a,c)\in R^\F\bigr)\bigr].
 \end{align}
First we establish $\eta(\dR a)\subseteq \dRca\eta(a)$. Let $b\le \dR a$.
 By~\eqref{Jemb2}, there is $c\in A$ with $c\le a$ and $(b,c)\in R^\F$. It follows that $c\in\eta(a)$, and so $b\in \dRca\eta(a)$.
To show $\dRca\eta(a)\subseteq \eta(\dR a)$, take any $b\in A$  such that $(b,x)\in R^\F$ for some $x\in \eta(a)$. Then $x\le a$ and, by \eqref{Jemb1}, $b\le \dR x$.
By the monotonicity of $\dR$, $\dR x\le \dR a$, and so $b\le \dR a$, that is, $b\in  \eta(\dR a)$. 
(In fact, it is easy to see that \eqref{Jemb1} and \eqref{Jemb2} are actually equivalent to
$\forall a\,\eta(\dR a)= \dRca\eta(a)$.) 
Finally, we check that $\eta$ is injective. 
 If $a,b\in A$ and $a\ne b$ then we may assume that $a\not\le b$, in which case $a\in\eta(a)$ but $a\notin\eta(b)$.

For example, an $R^\F$ satisfying \eqref{Jemb1} and \eqref{Jemb2} can be defined by taking
\begin{equation}\label{JRclassic}
(a,b)\in R^\F\  \   \Longleftrightarrow\ \ a\le \dR b.
\end{equation}
We use this definition in the proofs of Theorems~\ref{Th:6} and \ref{t:exist}.
However, the proofs of Theorems~\ref{Th:2}, \ref{t:comp}, \ref{t:simple1} and \ref{t:fun} require  different $R^\F$ satisfying \eqref{Jemb1} and \eqref{Jemb2}.
 
 
 \subsubsection{Embeddings via filters}\label{Tembed}
 
 Let $\A = (A, \land, \top,\dR)_{R\in\R}$ be a \SLOa. 
 For any $U\subseteq A$ and $R\in\R$, we set 
 \[
 \dR[U] = \{\dR a \mid a\in U\}.
 \]
We remind the reader that a nonempty subset $U\subseteq A$ is  a
\emph{filter} (\emph{of} $\Aa$) if it is \emph{up-closed} (in the sense that $a\in U$ and $a\le b$ imply $b\in U$) and
$\land$-\emph{closed} (that is, $a\land b\in U$ for any $a,b\in U$). 
We denote by $\FA$ the set of all filters of $\Aa$.

We embed $\A$ into  $\Gslo$, for some frame $\G=(\FA,R^\G)_{R\in\R}$,  using the map
 \[
 f \colon a \mapsto \{U \in \FA \mid a \in U\}.
 \]
Clearly, 
$f(\top)=\FA$ and $f(a\land b)=f(a)\cap f(b)$  for all $a,b \in A$.
 Also, it is readily seen that to ensure $f\bigl(\dR a\bigr)=\dRca f(a)$ for all $a$, we can equivalently require that the following two conditions hold for all $U\in\FA$ and $R\in\R$:
 \begin{align}
 \label{emb1}
&  \forall V\ \bigl((U,V)\in R^\G\  \   \Rightarrow\ \ \dR [V]\subseteq U\bigr),\\
 \label{emb2}
&  \forall a\ \bigl[ \dR a\in U \  \ \Rightarrow\ \ \exists V \bigl(a\in V \ \mbox{ and }\  (U,V)\in R^\G\bigr)\bigr].
 \end{align}
To check that $f$ is injective, let $a\ne b$. We may assume that $a\not\le b$ and take the filter $\{a\}\!\!\uparrow~=\{b \mid a\le b\}$ 
(the \emph{principal filter} generated by $a$).  Then $\{a\}\!\!\uparrow\,\in f(a)$ but $\{a\}\!\!\uparrow\,\notin f(b)$.

For example, one can define $R^\G$ by taking
 \begin{equation}\label{Rclassic}
  	(U,V)\in R^\G\  \   \Longleftrightarrow\ \ \dR [V]\subseteq U.
 \end{equation}
Again, in general, there can be different $R^\G$ satisfying \eqref{emb1} and \eqref{emb2}; see, e.g., the proofs of Theorems~\ref{t:comp0} and \ref{t:fun}~$(i)$.


\subsubsection{Connection between the two embeddings}

For an arbitrary \SLOa\ $\Aa$, with the `classical' definitions
of $R^\F$ and $R^{\G}$ via \eqref{JRclassic} and  \eqref{Rclassic}, respectively, we have the following: 
\begin{proposition}
The frame $(A,R^\F)_{R\in\R}$ is isomorphic to a \textup{(}not necessarily generated\textup{)}
subframe of $(\FA,R^\G)_{R\in\R}$.
For finite $\Aa$, these frames are isomorphic.
 \end{proposition}

\begin{proof}
For all $a,b\in A$, we have 
$(a,b)\in R^\F$ iff $a\le\d_R b$ iff $a\le \d_R c$ for all $c\geq b$ iff  
 $\bigl(\{a\}\!\!\uparrow,\{b\}\!\!\uparrow\bigr)\in R^\G$. 
 If $\Aa$ is finite, then all filters of $\Aa$ are principal.
\end{proof}


\subsection{Completeness by \method}\label{sec:interp} 
To introduce our second method for proving completeness, we establish some connections between \spterms{} and Kripke models.

Given Kripke models $\M_i = (\F_i,\V_i)$ based on frames $\F_i=(W_i,R_i)_{R\in \R}$, for $i=1,2$, a map $h \colon W_1 \to  W_2$ is called an $\F_1\to\F_2$ \emph{homomorphism} 
if $(x,y)\in R_1$ implies $\bigl(h(x),h(y)\bigr)\in R_2$, for any $x,y\in W_1$ and $R\in\R$.
If in addition $x\in\V_1(p)$ implies $h(x)\in\V_2(p)$, for any $x\in W_1$ and variable $p$, then 
$h$ is called an $\M_1\to\M_2$ \emph{homomorphism\/}. 
Clearly, \spterms{} are preserved under homomorphisms in the sense that 
$\M_1,x\kmodels \varrho$ implies $\M_2,h(x)\kmodels \varrho$,
for any $x\in W_1$ and \spterm{} $\varrho$.

\subsubsection{Kripke models from \spterms}\label{sec:treemodel}
We say that a frame $\F=(W,R^\F)_{R\in \R}$ is \emph{tree-shaped} (or simply a \emph{tree}) \emph{with root} $r$ if $(W,\bigcup_{R\in \R} R^\F)$ is a finite directed tree with root $r$ such that $R_1^\F\cap R_2^{\F}=\emptyset$ for all $R_1 \ne R_2$. (In particular, $(W,R^\F)_{R\in \R}$ is irreflexive and intransitive.) 

We use the following notions and notation throughout the paper.
Given an \spterm{} $\varrho$,  
 we define by induction  a Kripke model 
 \[
 \mbox{$\M_\varrho = (\T_\varrho,\V_\varrho)$ based on a finite tree 
 $\T_\varrho=(W_\varrho,R_\varrho)_{R\in \R}$ with root $r_\varrho$.}
 \]
 For $\varrho = \top$, $\T_\varrho$ consists of  a single irreflexive point $r_\varrho$ with $\V_\varrho(p)=\emptyset$ for all variables $p$. For $\varrho = p$, $\T_\varrho$ consists of a single irreflexive point $r_\varrho$, $\V_\varrho(p)=\{r_\varrho\}$, and $\V_\varrho(q)=\emptyset$ for $q \ne p$. For $\varrho = \varrho_1 \land \varrho_2$, we first construct disjoint $\M_{\varrho_1}$ and $\M_{\varrho_2}$, and then merge their roots $r_{\varrho_1}$ and $r_{\varrho_2}$ into $r_\varrho$ such that $r_\varrho\in \V_{\varrho}( q)$ iff 
$r_{\varrho_i}\in \V_{\varrho_i}( q)$, for some $i=1,2$.  Finally, for  $\varrho = \dR\varrho'$, we 
add a fresh point $r_\varrho$ to $W_{\varrho'}$, 
and set $R_\varrho=R_{\varrho'}\cup\{(r_\varrho,r_{\varrho'})\}$
and $\V_{\varrho}(p)=\V_{\varrho'}(p)$ for all variables $p$. We refer to $\M_\varrho$ as the \emph{$\varrho$-tree model\/}.
Note that $\M_\varrho$ and $\varrho$ are of the same size as the points in $W_\varrho$ are in one-to-one correspondence with the subformulas of $\varrho$.

\begin{proposition}\label{p:hom}
For any \spterm{} $\varrho$, Kripke model $\M$ and point $w$ in $\M$, we have 
$\M,w\kmodels\varrho$
iff there is a homomorphism \mbox{$h \colon \M_\varrho \to \M$} with $h(r_\varrho) = w$.
\end{proposition}
\begin{proof}
By a straightforward induction on the construction of $\varrho$.
\end{proof}

The connection between the validity of \spequations{} and homomorphisms between models proved below was first observed in~\cite{BaaderKuesters+-IJCAI-1999}.

\begin{corollary}\label{c:krchar} 
$(i)$
For any \spequation{} $\e = (\sigma \imp \tau)$, Kripke model $\M$ and point $w$ in $\M$, the following conditions are equivalent\textup{:}
\begin{itemize}
\item[--] $\M,w \models \e$;

\item[--] for every homomorphism $h_\sigma\colon\M_\sigma\to\M$ with  $h_\sigma(r_\sigma)=w$, there is a homomorphism  $h_\tau\colon\M_\tau\to\M$ with $h_\tau(r_\tau)=w$.
\end{itemize}

$(ii)$
For any  
\spterms{} $\sigma$ and $\tau$, we have $\models_{\CA}\sigma\imp\tau$ iff 
$\M_\sigma,r_\sigma\kmodels\tau$.
\end{corollary}
\begin{proof}
Claim $(i)$ is an immediate consequence of Proposition~\ref{p:hom}.

$(ii)$
($\Rightarrow$)
As the identity map on $\M_\sigma$ is a homomorphism, 
$\M_\sigma,r_\sigma\kmodels\sigma$
by Proposition~\ref{p:hom},
and so 
$\M_\sigma,r_\sigma\kmodels\tau$
 by the assumption.

($\Leftarrow$)
Suppose 
$\M,w\kmodels\sigma$,
for some Kripke model $\M$ based on a frame $\F$.
By Proposition~\ref{p:hom}, 
there is a homomorphism $h \colon \M_\sigma \to \M$ with $h(r_\sigma) = w$. Thus, 
$\M,w\kmodels\tau$ follows from $\M_\sigma,r_\sigma\kmodels\tau$. 
\end{proof}

 
\subsubsection{\sptermscap{} from Kripke models}\label{mterms} 
Suppose $\N = (\F,\V)$ is a Kripke model such that 
$\V(p) \ne \emptyset$ for finitely many variables $p$ only, and $\F = (W, R^\F)_{R\in\R}$ is a finite frame with root $r$ that contains no directed cycles.
We inductively associate with $\N$ an \spterm{} $\trm(\N)=\trm_r^\N$ by setting, for every $w\in W$, 
\[
\trm_w^\N = \bigwedge_{w\in\V(p)}p \ \land\bigwedge_{(w,v)\in R^\F}\dR\trm_v^\N.
\]
Clearly, 
$\N,w\kmodels\trm_w^\N$.
Observe that if $\F$ is a directed tree then 
$\trm_w^\N$ is the unique (modulo \SLOa-axioms \eqref{idem}--\eqref{assoc}) \spterm{} $\varrho$ such
that the $\varrho$-tree model $\M_\varrho$ is the submodel of $\N$ generated by $w$.
Thus, in this case $\M_{\trm(\N)}$ is the same as $\N$. In particular, 
$\trm(\M_\sigma)=\sigma$, for any \spterm{} $\sigma$.
In general, the $\trm(\N)$-tree model $\M_{\trm(\N)}$ is what is known in modal logic as the $r$-\emph{unravelling of\/} $\N$, and so:

\begin{proposition}\label{p:unravel}
For every \spterm{} $\tau$, 
$\N,r\kmodels\tau$ iff $\M_{\trm(\N)},r\kmodels\tau$.
\end{proposition}

We also note the following important fact:  

\begin{proposition}\label{p:htot}
If $h\colon\N_1\to\N_2$ is a homomorphism, then, for every $w$ in $\N_1$, we have 
$\vdash_{\SLO}\trm_{h(w)}^{\N_2}\imp\trm_w^{\N_1}$. 
\end{proposition}

%


\subsubsection{\methodcapital}

The above observations give another completeness proof for the \sptheory{} $\SPi + \emptyset$
(cf.\ Theorem~\ref{t:valid}).
Indeed, suppose $\models_{\CA}\sigma\imp\tau$. 
Then, by Corollary~\ref{c:krchar}~$(ii)$, we have 
$\M_\sigma,r_\sigma\kmodels\tau$,
and so by Proposition~\ref{p:hom}, there is a homomorphism $h \colon \M_\tau \to \M_\sigma$ with 
$h(r_\tau) = r_\sigma$. Thus, $\vdash_{\SLO}\sigma\imp\tau$ follows by Proposition~\ref{p:htot}, and so
$\models_{\SLO}\sigma\imp\tau$ by \eqref{birkhoff}.

This proof is a special case of the following general method of establishing  completeness of \sptheories, which we call the \emph{method of \method\/}.
In order to prove that an \sptheory{} $\SPi + \Ec$ is \eqcomp\ (without knowing whether it is \complex{} or  not), 
we do the following, for any given \spequation{} $\sigma\imp\tau$:
\begin{itemize}
\item[$(i)$]
transform one of the \spterms{} $\sigma$ or $\tau$ into some $\Ec\vdash_{\SLO}$-equivalent normal form 
resulting in an \spequation{} $\alpha\imp\beta$, called 
a $\Ec$-\emph{proxy} for $\sigma \imp \tau$;

\item[$(ii)$]
show that $\Ec\models_{\CA}\alpha\imp\beta$ is reducible to $\Ec^-\models_{\CA}\alpha\imp\beta$, for some 
subset $\Ec^-$ of $\Ec$ such that $\SPi + \Ec^-$ is \eqcomp\ and has the finite frame property.
\end{itemize}
The concrete $\Ec$-normal form used in this method depends on $\Ec$ and reflects the structure of its frames. Say, for $\Ec=\{\dR p\imp\dS p\}$ that defines the property
$\Phi=\forall x, y\, (R(x,y)\to S(x,y))$, we transform $\sigma$ into a $\Ec\vdash_{\SLO}$-equivalent \spterm{} describing the
$\Phi$-closure of the finite $\sigma$-tree model $\M_\sigma$, and take $\Ec^- = \emptyset$ (see Theorem~\ref{thm:interpolant}). 
For $\Axlin$ defining linear quasiorders, we transform $\tau$ into a conjunction of \spterms, each of which describes a linearly ordered full branch
of the finite $\tau$-tree model $\M_\tau$, and take $\Ec^-=\Axsfour$ (see Theorem~\ref{t:lincomp}).


We use the method of \method{} to obtain a number of \eqcomp ness results: 
Theorem~\ref{thm:interpolant}, which is a general \eqcomp ness criterion
(where we do not know whether all the covered \sptheories{} are \complex),
and Theorems~\ref{t:altn}, \ref{t:s5alt} and \ref{t:lincomp} (where the \sptheories{} in question are \emph{not} 
\complex).

In the next three sections, we apply the tools and techniques developed above to investigate completeness properties of \sptheories, systematising our results according to the form 
of the first-order correspondents of their axioms. 


\section{Completeness of \sptheories{} with universal Horn correspondents}\label{Sahlqvist}
We begin by recalling that, by the correspondence part of Sahlqvist's theorem~\cite{Sahlqvist75,Blackburnetal01}, 
 a first-order correspondent $\Psi_{\e}$ of any  \spequation{} $\e = (\sigma \imp \tau)$ can be constructed as follows, using the tree models $\M_\sigma$ and $\M_\tau$ from \S\ref{sec:treemodel}. 
Suppose $W_\sigma=\{v_0,v_1,\dots,v_{n_\sigma}\}$ with $v_0=r_\sigma$, and 
$W_\tau=\{u_0,u_1,\dots,u_{n_\tau}\}$ with $u_0=r_\tau$. With each point $w$ in $W_\sigma\cup W_\tau$, we associate a variable $\varel{w}$, and set
\begin{multline}\label{eq:corr1}
\Psi'_{\e}(\varel{v}_0) ~=~ \forall \varel{v}_1,\dots,\varel{v}_{n_\sigma} \, \Big( \hspace*{-2mm} \bigwedge_{\stackrel{i,j\leq n_{\sigma},\,R\in\R}{(v_i,v_j)\in R_\sigma}} \hspace*{-3mm} R(\varel{v}_i,\varel{v}_j)  \to \\
\exists \varel{u}_0,\dots,\varel{u}_{n_\tau} \, \big( (\varel{v}_0 = \varel{u}_0) \land \hspace*{-3mm} \bigwedge_{\stackrel{i,j\leq n_\tau,\,R\in\R}{(u_i,u_j)\in R_\tau}} \hspace*{-3mm} R(\varel{u}_i,\varel{u}_j) 
\land \hspace*{-3mm} \bigwedge_{\stackrel{i\leq n_\tau,\,p\in\varset}{u_i \in\V_\tau(p)}} \bigvee_{\stackrel{j\leq n_\tau}{v_j \in\V_\sigma(p)}}\!\!(\varel{u}_i=\varel{v}_j) \big)\Big).
\end{multline}
Then (as actually follows from Corollary~\ref{c:krchar}~$(i)$), for any frame $\F$ and any point $w$ in it,   
$\F,w \models \e$ iff $\F \models \Psi'_{\e}(\varel{v}_0)[w]$.
 The formula $\Psi'_{\e}(\varel{v}_0)$ with one free variable $\varel{v}_0$ is called a \emph{local  correspondent of} $\e$. The sentence $\Psi_{\e} = \forall \varel{v}_0\, \Psi'_{\e}(\varel{v}_0)$ is then a (\emph{global}) \emph{correspondent of} $\e$, that is, for every frame $\F$,
\begin{equation}\label{globalcorr}
\F\models \e \quad \Longleftrightarrow\quad \F \models \Psi_{\e}.
\end{equation}
The left-hand side of the implication in $\Psi_{\e}$ is just the diagram of the tree-shaped frame $\T_\sigma$ constructed from the atoms $R(\varel{v}_i,\varel{v}_j)$ with $(v_i,v_j)\in R_\sigma$. The right-hand side has a more complex structure that involves equality, disjunction and existential quantifiers. In some cases, 
$\Psi_{\e}$ is equivalent to a first-order sentence without some of these.
For example, reflexivity, transitivity or symmetry can clearly be defined without using any of
$=$, $\lor$ and $\exists$ on the right-hand side. 
On the other hand, $=$ is required to define functionality, $\lor$ is needed for linearity, and $\exists$ for density.
Note that if $\Psi_{\e}$ is equivalent to a universal sentence, then every subframe of a frame in $\CA_{\{\e\}}$ is also in  $\CA_{\{\e\}}$.
We call an \sptheory{} $\SPL$ a \emph{subframe logic} if $\CA_{\SPL}$ is closed under taking subframes.

\begin{theorem}\label{t:subframedec}
Every subframe \sptheory{} $\SPL$ has the polynomial finite frame property, and
is decidable in {\sc coNP} if \eqcomp{} and finitely axiomatisable.  
\end{theorem}

\begin{proof}
Decidability in {\sc coNP} follows from completeness and finite axiomatisability,
using the polynomial finite frame property. To show it, suppose $\e \notin L$ and $\e= (\sigma \imp \tau)$. Then
there is a Kripke model $\M\not\models \e$ based on some $\mathfrak{F}\in \CA_{\SPL}$, that is,  
$\M,w\kmodels\sigma$ and $\M,w\not\kmodels\tau$, for some point $w$.
By Proposition~\ref{p:hom}, there is a homomorphism
$h \colon \mathfrak{M}_{\sigma} \rightarrow \mathfrak{M}$ with $h(r_{\sigma})=w$.
Take the restrictions $\mathfrak{F}'$ and $\mathfrak{M}'$ of, respectively, $\mathfrak{F}$ and $\mathfrak{M}$ to $\{ h(w) \mid w\in W_{\sigma}\}$. Then $\mathfrak{M}'\not\models \e$ and $\mathfrak{F}'\in \CA_{\SPL}$ is a subframe of $\T_\sigma$, and so it is of polynomial size in $\e$.
\end{proof}


\subsection{Equality-free universal Horn correspondents.}\label{tprof}
By a \emph{profile} we mean a quadruple $\pi= (\G,\newrel,u,v)$,  
where $\G=(\Delta,R^\G)_{R\in\R}$ is a finite rooted frame with  $u,v\in\Delta$, $\newrel \in \R$ and $(u,v)\notin \newrel^\G$. Let $\Delta=\{x_0,\dots,x_n\}$. The profile $\pi$ represents the universal Horn sentence 
$$
\Phi_\pi ~=~  \forall \varel{x}_0,\dots,\varel{x}_n \, \big( \hspace*{-1mm} \bigwedge_{\stackrel{i,j\leq n,\,R\in\R}{(x_i,x_j)\in R^\G}} \hspace*{-2mm} R(\varel{x}_i,\varel{x}_j) \to \newrel(\varel{u},\varel{v}) \big). 
$$
We call $\e$ a \emph{\hornequation} if 
its correspondent $\Psi_{\e}$ is equivalent to $\Phi_\pi$ for some profile $\pi$, in which case we 
%
%
say that $\pi$ \emph{is a profile of\/} $\e$ or $\e$ \emph{has profile\/} $\pi$. 
Since $(u,v)\notin \newrel^\G$, we have $\G\not\models\Phi_\pi$, and so
$\G\not\models\Psi_{\e}$. Thus,
\begin{equation}\label{pinote}
\mbox{if $\pi= (\G,\newrel,u,v)$ is a profile of $\e$, then $\G\not\models\e$.}
\end{equation}
Given a set $\Pi$ of profiles and a frame $\F=(W,R^\F)_{R\in \R}$, 
we denote by $\Pi(\F)$ the $\Pi$-\emph{closure of\/} $\F$, that is, the smallest frame $\Hh$
extending $\F$ such that $\Hh\models\Phi_\pi$, for $\pi \in \Pi$.
If $\Pi = \{\pi\}$, we write $\pi(\F)$ instead of $\Pi(\F)$. Thus, $\pi(\F)$ contains the same points as $\F$ but possibly more $\newrel$-arrows between them.
For a Kripke model $\M = (\F,\V)$, we set $\Pi(\M) = \bigl(\Pi(\F),\V\bigr)$. 
Clearly, if both $\Pi$ and $\F=(W,R^\F)_{R\in\R}$ are finite, we can construct $\Pi(\F)$ step-by-step by defining
a finite sequence
\begin{equation}\label{stepbystep}
\F=\F^0,\dots,\F^i=(W,R^{\F^i})_{R\in\R},\dots,\F^n=\Pi(\F)
\end{equation}
of frames such that $n \le |\R| \cdot |W|^2$ and, for every $i<n$, there exist a profile $\pi^i= (\G^i,\newrel^i, u^i,v^i)$ in $\Pi$
and a homomorphism $h^i\colon\G^i\to\F^i$ with
\begin{equation}\label{stepbystep2}
R^{\F^{i+1}}=\left\{
\begin{array}{ll}
R^{\F^i}\cup\bigl\{\bigl(h^i(u^i),h^i(v^i)\bigr)\bigr\}, & \mbox{ if $R= \newrel^i$},\\
R^{\F^i}, & \mbox{ otherwise.}
\end{array}
\right.
\end{equation}
To put it another way, $\Pi(\F)$ is the result of applying the datalog program  with rules $\{\Phi_\pi \mid \pi \in \Pi\}$ to the input database $\{R^\F \mid R\in \R \}$, 
which can be done in polynomial time in $\F$ for a fixed finite $\Pi$~\cite{DBLP:journals/csur/DantsinEGV01}. 
In general, using a similar step-by-step construction for successor ordinals and taking the union for limits, one can show that, for any frames $\F,\F'$ and set $\Pi$ of profiles, 
\begin{equation}\label{homoext}
\text{any homomorphism $f\colon\F\to\F'$ is  a $\Pi(\F)\to\Pi(\F')$ homomorphism.}
\end{equation}

We have the following generalisation of Corollary~\ref{c:krchar}~$(ii)$: 

\begin{proposition}\label{p:Hclosure} 
Let $\Ec$ be a set of \hornequations{} and $\Pi_{\Ec}=\{\pi_{\e}\mid\e\in\Ec\}$ their profiles. Then $\Ec\models_{\CA}\sigma\imp\tau$ iff 
$\Pi_{\Ec}(\M_\sigma),r_\sigma\kmodels\tau$, 
for any \spterms{} $\sigma$ and $\tau$.
\end{proposition}
\begin{proof}
$(\Rightarrow)$ 
As $\Pi_{\Ec}(\M_\sigma)$ extends the $\sigma$-tree model $\M_\sigma$,
the identity map is an $\M_\sigma\to \Pi_{\Ec}(\M_\sigma)$ homomorphism, and so 
$\Pi_{\Ec}(\M_\sigma),r_\sigma\kmodels\sigma$ 
by Proposition~\ref{p:hom}.
As $\Pi_{\Ec}(\T_\sigma)\models\Phi_{\pi_e}$ for every $\pi_{\e}\in\Pi_{\Ec}$, we have
$\Pi_{\Ec}(\T_\sigma)\models\Psi_{\e}$ for every $\e\in \Ec$, and so $\Pi_{\Ec}(\T_\sigma)\models\Ec$.
Therefore, $\Pi_{\Ec}(\T_\sigma)\models\sigma\imp\tau$, and so 
$\Pi_{\Ec}(\M_\sigma),r_\sigma\kmodels\tau$. 

$(\Leftarrow)$ 
Suppose 
$\M,w\kmodels\sigma$
for some Kripke model $\M$ based on a frame $\F \in\CA_\Ec$.
By Proposition~\ref{p:hom}, 
there is a homomorphism $h \colon \M_\sigma \to \M$ with $h(r_\sigma) = w$. 
By \eqref{homoext},
$h$ is a homomorphism from $\Pi(\M_\sigma)$ to $\Pi(\M)$. As $\F \models \Ec$, we have 
$\F\models\Phi_{\pi_{\e}}$ for any $\pi_{\e}\in\Pi_{\Ec}$. Thus, $\Pi_{\Ec}(\M)=\M$, and so
$\M,w\kmodels\tau$ follows from $\Pi_{\Ec}(\M_\sigma),r_\sigma \kmodels\tau$,
as required.  
\end{proof}

As the Kripke model $\Pi_{\Ec}(\M_\sigma)$ has $|W_\sigma|$-many points and can be constructed in 
polynomial time in $\sigma$, we obtain the following consequence of Proposition~\ref{p:Hclosure}:

\begin{theorem}\label{thm:horncomplete}
For any finite set $\Ec$ of \hornequations, $\SPi + \Ec$ has the polynomial finite frame property,
and is decidable in {\sc PTime} if \eqcomp. 
\end{theorem}
%


Note that full Boolean normal multi-modal logics axiomatisable by  \hornequations{}
can be very complex. For example, it is shown in~\cite{KikotSZ14} that 
$\mf{K}\oplus\Sigma$ is undecidable for
\[
\Sigma= \{\dR\d_{\!P\,}\dR p \imp \d_{\!P\,} p , \  
\d_{\!Q\,}\dR p \imp \d_{\!Q\,} p, \ 
\d_{\!Q\,}\d_{\!P\,} p \imp \d_{\!P\,} p\},
\]
On the other hand, by Corollary~\ref{t:ss}  below, the \sptheory{} $\SPi+\Sigma$
is \eqcomp, and so decidable in {\sc PTime} by Theorem~\ref{thm:horncomplete}.
For more decidability and complexity results for  modal logics 
of Horn definable classes of frames, 
the reader is referred to~\cite{DBLP:conf/stacs/HemaspaandraS08,DBLP:conf/lics/MichaliszynO12}.

In the remainder of this section, we provide a few general  sufficient conditions for  
completeness of \sptheories{} axiomatisable by \hornequations, and also give a number of counterexamples illustrating their boundaries. 

We say that  $\pi= (\G,\newrel,u,v)$ is a \emph{\tprof} if $\G$ is a tree with root $r_\G$. 

\begin{proposition}\label{p:eandpi}
Suppose that a \hornequation{} $\e = (\sigma \imp \tau)$ has a tree-profile $(\G,\newrel,u,v)$.
Then the following hold:

$(i)$
there exist a homomorphism $f\colon \T_\sigma\to\G$ and a homomorphism $g\colon\G\to\T_\sigma$; 

$(ii)$
for any homomorphism $h \colon \T_\sigma \to \G$, we have $h(r_\sigma) = r_\G$.
\end{proposition}

\begin{proof}
$(i)$ 
By \eqref{pinote} $\G \not \models \e$, 
 and so
there is a homomorphism $f\colon \T_\sigma\to\G$. Since $\not\models_{\CA}\e$,  
by Corollary~\ref{c:krchar}~$(ii)$ we obtain 
$\M_\sigma,r_\sigma\not\kmodels\tau$,
 from which $\T_\sigma \not\models \e$. Therefore, $\T_\sigma \not\models \Phi_\pi$, and so
there is a homomorphism $g\colon\G\to\T_\sigma$. 

$(ii)$ 
Suppose $h\colon \T_\sigma\to\G$ is a homomorphism. Then
the composition of $g$ and $h$ is a homomorphism from the finite tree $\G$ to itself, which gives $h\bigl(g(r_\G)\bigr)=r_\G$, and so
$g(r_\G)=r_\sigma$ must hold as well.
\end{proof}

A profile $\pi = (\G,\newrel,u,v)$ is \emph{minimal} if there is no profile $\pi'= (\G',\newrel',u',v')$ such that $|\G'| < |\G|$ and $\Phi_\pi$ is  equivalent to $\Phi_{\pi'}$.
As shown in~\cite{Kikot11}, for any minimal profile $\pi$, the class of frames validating $\Phi_\pi$ is modally definable iff $\pi$ is a \tprof. (Thus, every \hornequation{} has a correspondent $\Phi_\pi$ given by a minimal tree-profile $\pi$.) Moreover, any such modally definable class is 
in fact definable by a single \spequation{} $\e_\pi$ constructed in the following way.   

Suppose 
$y_0 R_1^\G y_1\dots  y_{\ell-1} R_\ell^\G y_{\ell}$
 is the unique path in the tree-shaped frame $\G=(\Delta,R^\G)_{R\in\R}$ from the root $y_0$
to $y_\ell=u$, 
for some $\ell<\omega$.
We introduce a propositional variable $p_x$ for each $x\in\{y_1,\dots,y_\ell,v\}$.
Let $\Delta=\{x_0, \dots, x_n\}$ be 
such that $x_0=y_0$ is the root of $\G$, and $(x_i, x_j)\in R^\G$ implies $i < j$, for all $i,j\le n$ and $R\in\R$. 
By induction on $i$ from $n$ to $0$, we set
\begin{equation}\label{sigmai}
\sigma_i =
\left\{
\begin{array}{ll}
p_{x} \land\hspace*{-.2cm}\displaystyle\bigwedge_{(x_i,x_j)\in R^\G}\hspace*{-.2cm} \dR \sigma_j, & \mbox{if $x_i=x$ for some $x\in\{y_1,\dots,y_\ell,v\}$},\\[5pt]
\top \land \hspace*{-.2cm}\displaystyle\bigwedge_{(x_i,x_j)\in R^\G}\hspace*{-.2cm} \dR \sigma_j, & \mbox{otherwise}
\end{array}
\right.
\end{equation}
and 
%
\begin{equation}\label{epi}
\e_\pi=\Bigl(
\sigma_0 \imp \d_{\!R_1\,} \bigl(p_{y_1} \land \d_{\!R_2\,} (p_{y_2} \land \dots \land \d_{\!R_\ell\,} (p_{u} \land 
\dRo p_{v}) \dots )\bigr)\Bigr).
\end{equation}
It is readily checked that $\e_\pi$ is a \hornequation{} and $\pi$ is a profile of $\e_\pi$. 
%


\subsubsection{Horn-implications with rooted \tprof s.}\label{rootedp}
We say that a tree-profile $\pi=(\G,\newrel,u,v)$ is \emph{rooted} if $u$ is the root of $\G$, in which case 
\[
\e_\pi ~=~ (\sigma_0 \imp \dRo p_v),
\]
and the only variable that occurs in $\sigma_0$ is $p_v$.
A few examples of \tprof s $\pi$ with their $\Phi_\pi$ and $\e_\pi$ 
are given in 
Table~\ref{t:default}, where the first two \tprof s (for reflexivity and transitivity) are rooted, and the last two (for symmetry and Euclideanness) are non-rooted.

\begin{table}[htp]
\begin{center}
\begin{tabular}{|c|c|l|}\hline
profile $\pi$ & $\Phi_\pi$ & $\qquad\e_\pi$\\\hline\hline
\begin{tikzpicture}[>=latex, point/.style={circle,draw=black,minimum size=1mm,inner sep=0pt},hm/.style={dashed}, xscale=1.0, yscale=1.0]\scriptsize
\node (r) at (0,0) [point,fill=black] {};
\draw[->,loop,dashed,looseness=20] (r) to node {} (r);
\end{tikzpicture} & 
{\small $\forall x \, R(x,x)$} & {\small $\erefl:\  p \imp \d p$}\\[3pt]\hline
%
\begin{tikzpicture}[>=latex, point/.style={circle,draw=black,minimum size=1mm,inner sep=0pt},hm/.style={dashed}, xscale=1.0, yscale=1.0]\scriptsize
\node (u) at (0,0) [point,fill=black] {};
\node (w) at (1,0) [point,fill=black] {};
\node (v) at (2,0) [point,fill=black] {};
\draw[->,] (u) to node {} (w);
\draw[->,] (w) to node {} (v);
\draw[->,dashed,bend left,looseness=0.9] (u) to node {} (v);
\end{tikzpicture}
& {\small $\forall x,y,z \, \bigl(R(x,y) \land R(y,z) \to R(x,z)\bigr)$} & {\small $\etrans:\ \d\d p \imp \d p$}\\[3pt]\hline
&& {\small $\esym:$}\\
%
\begin{tikzpicture}[>=latex, point/.style={circle,draw=black,minimum size=1mm,inner sep=0pt},hm/.style={dashed}, xscale=1.0, yscale=1.0]\scriptsize
\node (u) at (0,0) [point,fill=black] {};
\node (v) at (1,0) [point,fill=black] {};
\draw[->,] (u) to node {} (v);
\draw[->,dashed,bend right,looseness=0.9] (v) to node {} (u);
\end{tikzpicture}
&
{\small $\forall x,y \, \bigl(R(x,y) \to R(y,x)\bigr)$} & {\small $\ \ q \land \d p \imp \d (p \land \d q)$}\\[3pt]\hline
&& {\small $\eeuc:$}\\[-8pt]
\begin{tikzpicture}[>=latex, point/.style={circle,draw=black,minimum size=1mm,inner sep=0pt},hm/.style={dashed}, xscale=1.0, yscale=1.0]\scriptsize
\node (r) at (0,0) [point,fill=black] {};
\node (u) at (-0.5,0.5) [point,fill=black] {};
\node (v) at (0.5,0.5) [point,fill=black] {};
\draw[->,] (r) to node {} (u);
\draw[->,] (r) to node {} (v);
\draw[->,dashed] (u) to node {} (v);
\end{tikzpicture}%
& {\small $\forall x,y,z \, \bigl(R(x,y) \land R(x,z) \to R(y,z)\bigr)$} & {\small $\ \d p \land \d q \imp \d (p \land \d q)$}\\[3pt]\hline
\end{tabular}
\end{center}
\caption{Examples of rooted and non-rooted tree-profiles.}\label{t:default}
\end{table}%

\begin{theorem}\label{Th:2}
Any \sptheory{} axiomatised by \spequations{} $\e_\pi$, for some rooted \tprof{}s  $\pi$, is \complex, and so \eqcomp.
\end{theorem}

A generalisation of this theorem (Theorem~\ref{t:exist}) will be proved in \S\ref{sec:exist}.
%
%
%
%
%
%
%
%
%
Note that as a consequence we obtain the following:

\begin{corollary}[\cite{stokkermans2008}]\label{t:ss}
Any \sptheory{} axiomatised by \spequations{} of the form $\d_{\!R_1\,} \dots \d_{\!R_n\,} p \imp \d_{\!R_0\,} p$, for $n\ge 0$, is \complex, and so \eqcomp.
In particular, $\SPi + \{\erefl\}$, $\SPi + \{\etrans\}$, and $\Esfour$ are all \complex{} and \eqcomp.
\end{corollary}

In general, there can be different \spequations{} $\e$ with the same rooted \tprof\ $\pi$.
Since for each such $\e$, $\Psi_{\e}$ is equivalent to $\Phi_\pi$, $\e$ and $\e_\pi$ are valid in the same  frames. However, we do not necessarily have $\text{\SLO}_{\{e\}}=\text{\SLO}_{\{\e_\pi\}}$, and so Theorem~\ref{Th:2} cannot be generalised to all such \spequations, as shown by the following examples.

\begin{example}\label{Ex:4}\em
Consider first the rooted \tprof\ $\pi$ for reflexivity in Table~\ref{t:default}.  It is not hard to see that the \spequation{} $\e = \bigl(p \imp \d\d (p \land \d p)\bigr)$ also has $\Phi_\pi$ as its correspondent, and so $\e_\pi=\erefl$ is valid in exactly the same frames as $\e$. On the other hand, $\{\e\} \not\models_{\text{\SLO}} \e_\pi$ because the \SLOa{} in Fig.~\ref{f:pic3}~(a) validates $\e$ but refutes 
$\e_\pi$ when $p$ is $\{1,3\}$. Therefore, $\SPi + \{\e\}$ is not \eqcomp.
\end{example}

\begin{example}\label{Ex:4b}\em
Let $\pi = (\G,R,v_1,v_4)$, where $\G$ is an $R$-chain of 
$v_1,v_2,v_3,v_4$.
It is not hard to check that  
$\Phi_\pi = \forall \avec{x} \, \bigl(  R(x_1,x_2) \land R(x_2,x_3) \land R(x_3,x_4)  \to R(x_1,x_4)\bigr)$ is a correspondent of 
the \spequation{} $\e' = (\d\d p \land \d\d\d p \imp \d p)$. The \SLOa{} in 
Fig.~\ref{f:pic3}~(b)  validates $\e'$ but refutes $\e_\pi = (\d \d \d p \imp \d p) $ when $p$ is $\{4\}$.
Therefore, $\{\e'\} \not\models_{\text{\SLO}} \e_\pi$, and so $\SPi + \{\e'\}$ is not complete.
\end{example}
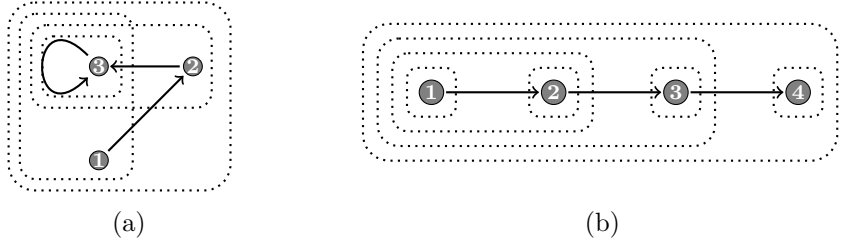
\begin{figure}[ht]
\centering
\begin{tikzpicture}[scale=.5]
\tikzset{every loop/.style={thick,min distance=15mm,in=-135,out=135,looseness=10}}
\tikzset{place/.style={circle,thin,draw=white,fill=white,scale=1.2}}

\draw [fill=gray] (1,1) circle [radius=.25];
\node at (1,1)  {\textcolor{white}{\bf\footnotesize 1}}; 
\draw [fill=gray] (3.5,3.5) circle [radius=.25];
\node at (3.5,3.5)  {\textcolor{white}{\bf\footnotesize 2}}; 
\node[place] (foo) at (1,3.5) {}; 
\draw [fill=gray] (1,3.5) circle [radius=.25];
\node at (1,3.5)  {\textcolor{white}{\bf\footnotesize 3}}; 
\draw [thick,->] (1.25,1.25) to [out=45, in=-135] (3.25,3.25);
\draw [thick,->] (3.15,3.5) to [out=180, in=0] (1.3,3.5);
\path[->] (foo) edge [loop] node {} ();
\draw[thick,dotted,rounded corners=4] (-.5,4.3) -- (1.6,4.3) -- (1.6,2.7) -- (-.5,2.7) -- cycle;
\draw[thick,dotted,rounded corners=6] (-.8,4.6) -- (4,4.6) -- (4,2.4) -- (-.8,2.4) -- cycle;
\draw[thick,dotted,rounded corners=6] (-1.1,4.9) -- (1.9,4.9) -- (1.9,.5) -- (-1.1,.5) -- cycle;
\draw[thick,dotted,rounded corners=10] (-1.4,5.2) -- (4.5,5.2) -- (4.5,.2) -- (-1.4,.2) -- cycle; 

\node at (1.8,-.7) {(a)};
\end{tikzpicture}
\hspace*{1.5cm}
\begin{tikzpicture}[scale=.65]
\draw [fill=gray] (1,2) circle [radius=.25];
\node at (1,2)  {\textcolor{white}{\bf\footnotesize 1}}; 
\draw [fill=gray] (3.5,2) circle [radius=.25];
\node at (3.5,2)  {\textcolor{white}{\bf\footnotesize 2}}; 
\draw [fill=gray] (6,2) circle [radius=.25];
\node at (6,2)  {\textcolor{white}{\bf\footnotesize 3}}; 
\draw [fill=gray] (8.5,2) circle [radius=.25];
\node at (8.5,2)  {\textcolor{white}{\bf\footnotesize 4}}; 
\draw [thick,->] (1.3,2) to [out=0, in=180] (3.2,2);
\draw [thick,->] (3.8,2) to [out=0, in=180] (5.7,2);
\draw [thick,->] (6.3,2) to [out=0, in=180] (8.2,2);
\draw[thick,dotted,rounded corners=4] (.5,2.5) -- (1.5,2.5) -- (1.5,1.5) -- (.5,1.5) -- cycle;
\draw[thick,dotted,rounded corners=4] (3,2.5) -- (4,2.5) -- (4,1.5) -- (3,1.5) -- cycle;
\draw[thick,dotted,rounded corners=4] (5.5,2.5) -- (6.5,2.5) -- (6.5,1.5) -- (5.5,1.5) -- cycle;
\draw[thick,dotted,rounded corners=4] (8,2.5) -- (9,2.5) -- (9,1.5) -- (8,1.5) -- cycle;
\draw[thick,dotted,rounded corners=6] (.2,2.8) -- (4.3,2.8) -- (4.3,1.2) -- (.2,1.2) -- cycle;
\draw[thick,dotted,rounded corners=8] (-.1,3.1) -- (6.8,3.1) -- (6.8,.9) -- (-.1,.9) -- cycle;
\draw[thick,dotted,rounded corners=10] (-.4,3.4) -- (9.3,3.4) -- (9.3,.6) -- (-.4,.6) -- cycle;

\node at (4.5,-.7) {(b)};
\end{tikzpicture}
\caption{The \SLOa s of Examples~\ref{Ex:4} and \ref{Ex:4b}.}\label{f:pic3}
\end{figure}

We say that a rooted \tprof\ $\pi= (\G,\newrel,u,v)$ is \emph{leapfrog} if $(u,w)\notin \newrel^\G$ for any $w$ in $\G$; and we refer to a \hornequation{} of the form $\varrho \imp \dRo p$ having a leapfrog profile as a \emph{\leapfrogequation}.

\begin{theorem}\label{Th:6}
Any \sptheory{} axiomatised by \leapfrogequations{} is \complex, and so \eqcomp.
\end{theorem}
\begin{proof}
Suppose $\e= (\varrho \imp \dRo p)$ is a \hornequation{} with a leapfrog profile $\pi= (\G,\newrel,u,v)$.
Recall the finite tree $\T_\varrho=(W_\varrho,R_\varrho)_{R\in \R}$ with root $r_\varrho$ from
\S\ref{sec:treemodel}.
 By Proposition~\ref{p:eandpi}, we obtain that
 \begin{equation}\label{leapfrog}
\mbox{there is no $z$ with $(r_\varrho,z)\in \newrel_\varrho$.}
 \end{equation}
 
 \begin{lclaim}\label{c:homo}
 $(i)$
For every $y\in\V_\varrho (p)$, there is a homomorphism $h^y\colon \T_\varrho \to \G$ with $h^y(r_\varrho)=u$ and $h^y(y) = v$.

 $(ii)$
 There is a homomorphism $h \colon \T_\varrho \to \G$ such that $h(r_\varrho)=u$ and $h(y) = v$, for all $y\in\V_\varrho (p)$.
 \end{lclaim}
 
 \begin{proof}
 $(i)$  
 Fix some $y\in\V_\varrho (p)$ and consider  the rooted tree-profile 
 $\pi_{\varrho,y} = (\T_\varrho, \newrel,r_\varrho, y)$.  
 With each point $x$ in $W_\varrho$ we associate a variable $\varel{x}$.
 As 
\begin{align*}
& \Phi_{\pi_{\varrho,y}} =\  \forall \varel{\avec{x}} \, \big( \bigwedge_{\stackrel{x,x'\in W_\varrho,\,R\in\R,}{(x,x')\in R_\varrho}} R(\varel{x},\varel{x}')  \to  \newrel(\varel{r}_\varrho,\varel{y}) \big),
\quad\mbox{and}\\
& \Phi_\pi\ \leftrightarrow\ \Psi_{\e}\ \leftrightarrow\  \forall \varel{\avec{x}} \, \big( \bigwedge_{\stackrel{x,x'\in W_\varrho,\,R\in\R,}{(x,x')\in R_\varrho}}
R(\varel{x},\varel{x}')  \to \bigvee_{y \in\V_\varrho(p)} \newrel(\varel{r}_\varrho,\varel{y}) \big),
\end{align*}
$\Phi_{\pi_{\varrho,y}}$ implies $\Phi_\pi$. Take the $\pi_{\varrho,y}$-closure $\pi_{\varrho,y}(\G)$ of $\G$.
As $\pi_{\varrho,y}(\G)\models \Phi_{\pi_{\varrho,y}}$, we have $\pi_{\varrho,y}(\G)\models\Phi_\pi$.
As the identity map is a homomorphism from $\G$ to $\pi_{\varrho,y}(\G)$,
\begin{equation}\label{uv}
(u,v)\in R^{\pi_{\varrho,y}(\G)}.
 \end{equation}
Next, consider the step-by-step construction \eqref{stepbystep}--\eqref{stepbystep2} of $\pi_{\varrho,y}(\G)$. We show by induction that, for every $i<n$,
$(a)$
the homomorphism $h^i\colon\T_\varrho\to\F^i$ used to obtain $\F^{i+1}$ from $\F^i$ is in fact a
$\T_\varrho\to\G$ homomorphism, and so, by Proposition~\ref{p:eandpi}~$(ii)$,
$(b)$
the new pair in $\newrel^{\F^{i+1}}$ is $\bigl(u,h^i(y)\bigr)$.
Indeed, for $i=0$ this follows from $\F^0=\G$. Now suppose inductively that $(a)$ and $(b)$ hold
for all $j\leq i$, and take the homomorphism $h^{i+1}\colon\T_\varrho\to\F^{i+1}$.
Since by IH all the $\newrel$-pairs in $\F^{i+1}$ that are not in $\G$ are of the form $(u,z)$, for some $z$,
\eqref{leapfrog} implies that $h^{i+1}$ is a $\T_\varrho\to\G$ homomorphism, proving $(a)$.
Now by \eqref{uv} and $(a)$, there is $i<n$ such that $h^i(r_\varrho)=u$ and $h^i(y)=v$, 
for the homomorphism $h^i\colon\T_\varrho\to\G$, as required.

$(ii)$ 
We define a homomorphism $h \colon \T_\varrho \to \G$ as follows.
First, define $h$ on the \emph{trunk} of $\T_\varrho$ comprising the points that lie on the paths from $r_\varrho$ to some $y\in\V_\varrho (p)$. Namely, for each $z$ on the trunk, we take any $y$ such that $z$ lies on the path from $r_\varrho$ to $y$ and set $h(z) = h^y(z)$ 
(which is well-defined since $\G$ is a tree, and so all the $y$ are located at the same distance from $r_\varrho$).
Next, for any $d$ on the trunk, we take the \emph{branch} with \emph{base} $d$ (containing all non-trunk descendants of $d$), fix some $y$ such that $y\in\V_\varrho (p)$ and $d$ lies on the path from $r_\varrho$ to $y$, and set $h(z) = h^y(z)$ for any $z$ on that branch. It is readily seen that $h$ is as required.
\end{proof}

Now, let $\A = (A, \land, \top, \dR)_{R\in \R}$ be a SLO validating $\e$.
It is shown in \S\ref{Jembed} that $\Aa$ can be embedded
into $\Fslo$, for the frame $\F=(A,R^\F)_{R\in\R}$ with $R^\F$ given by \eqref{JRclassic}. We show that $\F\models\Phi_\pi$, and so $\F\models\Psi_{\e}$, as required. 
To begin with,
take the tree-shaped frame $\G=(\Delta,R^\G)_{R\in\R}$ and suppose that $\Delta=\{x_1, \dots, x_n\}$ 
such that $x_1=u$ is the root,
and $(x_i, x_j)\in R^\G$ implies $i < j$. For each $i=1,\dots,n$, take some $a_{x_i}\in A$ such that
$(a_{x_i},a_{x_j})\in R^\F$ whenever $(x_i,x_j)\in R^\G$. 
We need to show that $(a_u,a_v)\in \newrel^\F$, that is,
$a_u\le\dRo a_v$.
Take the \spterms{} $\sigma_i$ defined in \eqref{sigmai}.
We prove by induction on $i=n,\dots,0$ that
\begin{equation}\label{treeok}
a_{x_i}\le\sigma_i[a_v].
\end{equation}
Indeed, as $x_n$ is a leaf in $\G$, $\sigma_n$ is either $\top$ (if $x_n\ne v$) or $p_v$ (if $x_n=v$), and so 
in either case \eqref{treeok} holds for $a_{x_n}$.
Now suppose inductively that \eqref{treeok} holds for every $j$, $i<j\leq n$.
We have $a_{x_i}\le\dR a_{x_j}$ for every $x_j$ with $(x_i,x_j)\in R^\G$.
So, by IH and monotonicity, we have 
\[
a_{x_i}\le a_{x_i}\land \bigwedge_{(x_i,x_j)\in R^\G} \dR \sigma_j[a_v].
\]
Since
\[
\sigma_i[a_v]=\left\{
\begin{array}{ll}
\displaystyle
\top\land \bigwedge_{(x_i,x_j)\in R^\G} \dR \sigma_j[a_v], & \mbox{ if $x_i\ne v$},\\[5pt]
\displaystyle
a_v\land \bigwedge_{(x_i,x_j)\in R^\G} \dR \sigma_j[a_v], & \mbox{ if $x_i= v$},
\end{array}
\right.
\]
 \eqref{treeok} follows. In particular, we have
$a_u=a_{x_1}\le\sigma_1[a_v]$. 

Now, take the following valuation $\vala$ in $\Aa$, for any variable $q$:
\[
\vala(q)=\left\{
\begin{array}{ll}
a_v, & \mbox{ if $q=p$},\\
\top, & \mbox{ otherwise,}
\end{array}
\right.
\]
and take the homomorphism $h$ from Claim~\ref{c:homo}~$(ii)$.
For any $y$ in $\T_\varrho$, take the \spterm{} $\trm_y^{\M_\varrho}$
defined in \S\ref{mterms}.
One can readily show by induction that
%
\[
\mbox{$h(y) = x_i$ \ implies \ $\sigma_i[a_v]\le \trm_y^{\M_\varrho}[\vala]$}.
\]
Indeed, if $y$ is a leaf in $\T_\varrho$ and $y\notin\V_\varrho(p)$, then $\trm_y^{\M_\varrho}[\vala]=\top$.
If $y$ is a leaf and $y\in\V_\varrho(p)$, then $h(y)=v$,
and so $\sigma_i[a_v]\le a_v=\trm_y^{\M_\varrho}[\vala]$.
If $y\in\V_\varrho(p)$ and has $\ell$ successors $y_{0},\dots,y_{\ell-1}$ with $(y,y_{j})\in R^j_\varrho$, then by IH and monotonicity, we have
\[
\sigma_i[a_v]\le a_v\land\bigwedge_{\substack{x_k=h(y_j)\\ \mbox{\tiny for some } j<\ell}}\d_{\!R^j\,}\sigma_{k}[a_v]\le
a_v\land\bigwedge_{j<\ell} \d_{\!R^j\,}\trm_{y_j}^{\M_\varrho}[\vala]=
\trm_y^{\M_\varrho}[\vala].
\]
The case $y\notin\V_\varrho(p)$ is similar.
In particular, we have $\sigma_1[a_v]\le\varrho[\vala]$. Finally, as
$\Aa \models \varrho \imp \dRo p$, we obtain $\sigma_1[a_v]\le\dRo a_v$, 
and so $a_u\le \dRo a_v$  by \eqref{treeok}. 
\end{proof}
 

\subsubsection{Horn-implications with arbitrary \tprof s.} 
We consider next \hornequations{} with \tprof s $\pi = (\G,\newrel,u,v)$ such that $u$ is not 
necessarily the root of 
the tree $\G$. Here again there are both positive and negative results. We begin by proving a general sufficient condition for completeness.

A set $\Pi$ of tree-profiles is called \emph{stable} if, for any $\pi=(\G, \newrel,u,v)$ in $\Pi$ and any tree $\T$, every homomorphism $h \colon \G \to \Pi(\T)$ is also a homomorphism from $\G$ to $\T$. 
To illustrate, $\{\pi_1\}$ and $\{\pi_2\}$ in Fig.~\ref{f:stable} are stable, while $\{\pi_3\}$
is not 
(take the `linear' frame $\T$ with $S^\T=\{(u_1,u_2)\}$ and $R^\T=\{(u_2,u_3),(u_3,u_4)\}$). 
We say that a tree-profile $\pi=(\G, \newrel,u,v)$ is \emph{forward-looking} if $u <_\G v$, where $<_\G$ is the transitive closure of $\bigcup_{R\in \R} R^\G$.


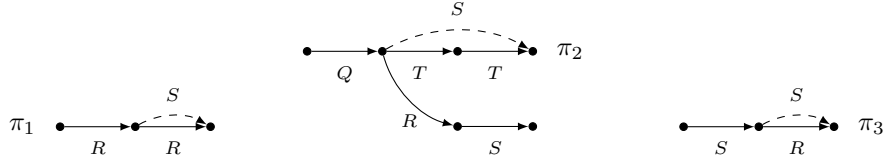
\begin{figure}[ht]
\centering
\begin{tikzpicture}[>=latex, point/.style={circle,draw=black,minimum size=1mm,inner sep=0pt},hm/.style={dashed}, xscale=1.0, yscale=1.0]\scriptsize
\node (p) at (-0.5,0) {\normalsize$\pi_1$};
\node (u) at (0,0) [point,fill=black] {};
\node (w) at (1,0) [point,fill=black] {};
\node (v) at (2,0) [point,fill=black] {};
\draw[->,] (u) to node [label=below:{\scriptsize$R$}] {} (w);
\draw[->,] (w) to node [label=below:{\scriptsize$R$}] {} (v);
\draw[->,dashed,bend left,looseness=0.9] (w) to node [label=above:{\scriptsize$S$}] {} (v);
\end{tikzpicture}
\qquad\quad 
\begin{tikzpicture}[>=latex, point/.style={circle,draw=black,minimum size=1mm,inner sep=0pt},hm/.style={dashed}, xscale=1.0, yscale=1.0]\scriptsize
\node (x) at (-1,0) [point,fill=black] {};
\node (u) at (0,0) [point,fill=black] {};
\node (w) at (1,0) [point,fill=black] {};
\node (v) at (2,0) [point,fill=black] {};
\node (w1) at (1,-1) [point,fill=black] {};
\node (v1) at (2,-1) [point,fill=black] {};
\node (p) at (2.5,0) {\normalsize$\pi_2$};
\draw[->,] (x) to node [label=below:{\scriptsize$Q$}] {} (u);
\draw[->,] (u) to node [label=below:{\scriptsize$T$}] {} (w);
\draw[->,] (w) to node [label=below:{\scriptsize$T$}] {} (v);
\draw[->,bend right,looseness=0.9] (u) to node [label=below:{\scriptsize$R$}] {} (w1);
\draw[->,] (w1) to node [label=below:{\scriptsize$S$}] {} (v1);
\draw[->,dashed,bend left,looseness=0.9] (u) to node [label=above:{\scriptsize$S$}] {} (v);
\end{tikzpicture}
\qquad\quad 
\begin{tikzpicture}[>=latex, point/.style={circle,draw=black,minimum size=1mm,inner sep=0pt},hm/.style={dashed}, xscale=1.0, yscale=1.0]\scriptsize
\node (u) at (0,0) [point,fill=black] {};
\node (w) at (1,0) [point,fill=black] {};
\node (v) at (2,0) [point,fill=black] {};
\node (p) at (2.5,0) {\normalsize$\pi_3$};
\draw[->,] (u) to node [label=below:{\scriptsize$S$}] {} (w);
\draw[->,] (w) to node [label=below:{\scriptsize$R$}] {} (v);
\draw[->,dashed,bend left,looseness=0.9] (w) to node [label=above:{\scriptsize$S$}] {} (v);
\end{tikzpicture}
\caption{Stable and unstable tree-profiles.}\label{f:stable}
\end{figure}

Suppose a tree-profile $\pi=(\G, \newrel,u,v)$ is forward-looking and $\G=(\Delta, R^\G)_{R\in\R}$. 
We define an \spequation{} $\e'_\pi$ as follows. 
For every $x\in \Delta$, we take a propositional variable $p_x$, and denote by $\V$ the valuation given by $\V(p_x) = \{x\}$. Let $\M = (\G, \V)$ and $\M'=(\G',\V)$, where  $\G'=(\Delta, R^{\G'})_{R\in\R}$ with $R^{\G'}=R^\G$, for $R\neq \newrel$, and
$\newrel^{\G'} = \newrel^\G \cup \{(u,v)\}$.  Since $\pi$ is forward-looking, $\G'$ does not contain directed cycles, and so both \spterms{} $\trm(\M)$ and $\trm(\M')$ are
defined (see \S\ref{mterms}), with $\trm(\M')$ obtained by substituting $p_u \land \dRo \trm_v^{\M}$ for $p_u$ in  $\trm(\M)$.
We set
\[ 
\e'_\pi = \bigl(\trm(\M) \imp \trm(\M')\bigr).
\]
It is readily checked that $\pi$ is a profile of $\e'_\pi$. The difference between $\e'_\pi$ and the
\spequation{} $\e_\pi$ defined by \eqref{epi}
is that the former contains propositional variables for all points in $\G$, while the latter only for $v$ and for the points on the path from the root of $\G$ to $u$. For example, for the transitivity profile $\pi$ from Table~\ref{t:default}, we have 
$$
\e'_{\pi} = \big( p_1 \land \d(p_2 \land \d p_3) \imp p_1 \land \d(p_2 \land \d p_3) \land \d p_3 \big) \quad \text{and} \quad \e_\pi = \big(\d\d p \imp \d p\big).
$$
The extra variables make it possible to obtain the following: 
\begin{theorem}\label{thm:interpolant}
For any stable set $\Pi$ of forward-looking tree-profiles, the \sptheory{} $\SPi + \Ec'_\Pi$, for $\Ec'_\Pi = \{\e'_\pi \mid \pi \in \Pi\}$,  
is \eqcomp. 
\end{theorem}
\begin{proof}
The proof uses the \method\ method from \S\ref{sec:interp}. 
Given an \spterm{} $\sigma$, we take the $\Pi$-closure $\Pi(\M_\sigma)$ of its tree-model $\M_\sigma$.
As every $\pi\in\Pi$ is forward-looking, $\Pi(\T_\sigma)$ does not contain directed cycles, and so
the \spterm{}
$\trm\bigl(\Pi(\M_\sigma)\bigr)$ is defined in \S\ref{mterms}.
We show that $\varrho_\sigma=\trm\bigl(\Pi(\M_\sigma)\bigr)$ 
has the following properties: 
\begin{itemize}
\item[$(i)$]
for any \spterm{} $\tau$, if $\Ec'_\Pi\models_{\CA}\varrho_\sigma\imp\tau$ then $\models_{\CA}\varrho_\sigma\imp\tau$,
\item[$(ii)$]
$\Ec'_\Pi\vdash_{\SLO}\varrho_\sigma\imp\sigma$ and $\Ec'_\Pi\vdash_{\SLO}\sigma\imp\varrho_\sigma$,
\end{itemize}
which clearly imply that $\SPi + \Ec'_\Pi$ is complete.

$(i)$ If $\Ec'_\Pi\models_{\CA}\varrho_\sigma\imp\tau$ then 
$\Pi(\T_\sigma)\models \varrho_\sigma\imp\tau$. As 
$\Pi(\M_\sigma),r_\sigma\kmodels\varrho_\sigma$,
we obtain that 
$\Pi(\M_\sigma),r_\sigma\kmodels\tau$, and so $\M_{\varrho_\sigma},r_\sigma\kmodels\tau$ 
by Proposition~\ref{p:unravel}.
Now, take any Kripke model $\M$ and a point $w$ in it with 
$\M,w\kmodels\varrho_\sigma$.
By Proposition~\ref{p:htot}, there is a homomorphism $h \colon \M_{\varrho_\sigma} \to \M$ with 
$h(r_\sigma) = w$, and so 
$\M,w\kmodels\tau$, 
as required.

$(ii)$ As $\Pi(\M_\sigma)$ extends $\M_\sigma$, the identity map is a homomorphism from
$\M_\sigma$ to $\Pi(\M_\sigma)$, from which $\vdash_{\SLO}\varrho_\sigma\imp\sigma$ follows by Proposition~\ref{p:htot}.
To prove  that $\Ec'_\Pi\vdash_{\SLO}\sigma\imp\varrho_\sigma$,
we construct $\Pi(\T_\sigma)$ step-by-step as in 
\eqref{stepbystep}--\eqref{stepbystep2}.  As every $\pi\in\Pi$ is forward-looking, the interim $\F^i$ do not contain directed cycles, but they are not necessarily trees. However,
as $\Pi$ is stable, at each step the homomorphism 
$h^i\colon\G^i\to\F^i$ we use to obtain $\F^{i+1}$ from $\F^i$ is actually a $\G^i\to\T_\sigma$ 
homomorphism, and so we can arrange the steps in such a way that 
the depth of $h^i(u^i)$ in $\T_\sigma$ is not smaller than the depth of $h^{i+1}(u^{i+1})$ in $\T_\sigma$.
This means that, for any $i< n$,
\begin{equation}\label{path}
\mbox{there is a unique path in $\F^i$ from $r_\sigma$ to $h^i(u^i)$.}
\end{equation}
Let $\M^i=(\F^i,\V_\sigma)$, for $i\le n$ (so $\M^0=\M_\sigma$ and $\M^n=\Pi(\M_\sigma)$).
We claim that
\begin{equation}\label{midterm}
\Ec\vdash_{\SLO}\trm(\M^i)\imp\trm(\M^{i+1}),\quad\mbox{for every $i<n$}.
\end{equation}
Indeed, fix some $i<n$ and suppose $r^i$ is the root of $\G^i$. By \eqref{path}, 
$\trm_{h^i(r^i)}^{\M^{i+1}}$ differs from $\trm_{h^i(r^i)}^{\M^i}$
in an extra conjunct $\d_{\!\newrel^i} \trm_{h^i(v^i)}^{\M^i}$ at the unique place corresponding to the point $h^i(u^i)$.
Therefore, the \spequation{} $\trm_{h^i(r^i)}^{\M^i}\imp \trm_{h^i(r^i)}^{\M^{i+1}}$ is in fact 
a substitution instance of $\e'_{\pi^i}$ obtained by replacing each $p_x$ in $\e'_{\pi^i}$ with
\[
\bigwedge_{h^i(x)\in\V_\sigma(p)}\hspace*{-.4cm}p\ \land
\bigwedge_{(y,R)\in A^i_x}\dR\trm_y^{\M^i},
\]
where
\[
A^i_x=\bigl\{(y,R)\mid \bigl(h^i(x),y\bigr)\in R^{\F^i},\mbox{ but $y\ne h^i(x')$ for any $x'$ with $(x,x')\in R^{\G^i}$}\bigr\}.
\] 
It remains to notice that
$\bigl\{\trm_{h^i(r^i)}^{\M^i}\imp \trm_{h^i(r^i)}^{\M^{i+1}}\bigr\}\vdash_{\SLO}\trm(\M^i)\imp\trm(\M^{i+1})$,
which proves \eqref{midterm}.
Finally, as 
\[
\trm(\M^0)=\trm(\M_\sigma)=\sigma\quad \mbox{ and }\quad\trm(\M^n)=\trm\bigl(\Pi(\M_\sigma)\bigr)=\varrho_\sigma,
\]
 we obtain $\Ec'_\Pi\vdash_{\SLO}\sigma\imp\varrho_\sigma$,
as required.
\end{proof}

\begin{question}
Does Theorem~\ref{thm:interpolant} hold for 
$\Ec_\Pi=\{\e_\pi \mid \pi \in \Pi\}$ in place of $\Ec'_\Pi$?
\end{question}

We do not know whether the \sptheories{} covered by Theorem~\ref{thm:interpolant} are complex. The next theorem 
indicates that showing this may require tricky embeddings.

\begin{theorem}\label{t:comp0}
The \sptheory{} $\SPi + \{\e_{\pi_1}'\}$ with $\pi_1$ from Fig.~\ref{f:stable} is \complex.
\end{theorem}
\begin{proof}
Suppose  $\A = (A, \land, \top, \dR,\dS)$ is a \SLOa\ validating the \spequation{} $\e_{\pi_1}'=\bigl(\dR(p\land \dR q)\imp\dR(p\land \dS q)\bigr)$. Take
the set $\FA$ of all filters of $\Aa$ and set, for $U,V\in\FA$,
\begin{align*}
& (U,V)\in R^\G\ \Longleftrightarrow\  \dR[V]\subseteq U,\ \mbox{and $\dR a\in V$ implies $\dS a\in V$ for every $a$};\\
& (U,V)\in S^\G\ \Longleftrightarrow\ \dS[V]\subseteq U.
\end{align*}
Then $\G=(\FA,R^\G,S^\G)$ clearly validates $\Phi_{\pi_1}$.
Also, $S^\G$ satisfies both \eqref{emb1} and \eqref{emb2}, and $R^\G$ satisfies \eqref{emb1}.  We show 
that $R^\G$ satisfies \eqref{emb2} as well. Then,
as shown in \S\ref{Tembed}, $\Aa$ would embed into $\Gslo$.
So suppose $\dR a\in U$ for some $a$. We need to find a $V\in\FA$ such that $a\in V$ and
$(U,V)\in R^\G$. To this end, for any $X\subseteq A$, we let 
$X\!\!\uparrow=\{y\mid y\geq x\mbox{ for some }x\in X\}$, $V_0=\{a\}\!\!\uparrow$ and,  for every $n<\omega$, 
\[
V_{n+1}=\{x\land\dS y_1\land\dots\land \dS y_m\mid x\land\dR y_1\land\dots\land \dR y_m\in V_n\}\!\!\uparrow.
\]
It can be shown by induction that, for every $n<\omega$,
\begin{itemize}
\item[--]
$V_n$ is a filter;
\item[--]
$\dR b\in V_n$ implies $\dS b\in V_{n+1}$, for every $b\in A$,
\item[--]
$\dR [V_n]\subseteq U$.
\end{itemize}
We show that last item only. For $n=0$, it holds because of the monotonicity of $\dR$. If
$b\geq x\land\dS y_1\land\dots\land \dS y_m$, for some $x\land\dR y_1\land\dots\land \dR y_m\in V_n$,
then by monotonicity and $\A \models \dR(p\land \dR q)\le\dR(p\land \dS q)$, we have
\[
\dR b\geq \dR(x\land\dS y_1\land\dots\land \dS y_m)\geq \dR(x\land\dR y_1\land\dots\land \dR y_m).
\]
Since $\dR(x\land\dR y_1\land\dots\land \dR y_m)\in U$ by IH, $\dR b\in U$ follows.

As $V_0\subseteq\dots\subseteq V_n\subseteq\dots$, their union $V=\bigcup_{n<\omega}V_n$ is
the required filter.

Note that Theorem~\ref{t:comp0} cannot be proved using the simpler embedding of \S\ref{Jembed}.
Indeed, take the infinite \SLOa\ $\A= (A, \land, \top,\d)$ with the elements
\[
\top=a_0> a_1>\dots > a_n>\dots > g,
\]
$\dR g=\dS g = g$, $\dR a_n=\top$,
 and $\dS a_n=a_{n+1}$, for $n<\omega$. Then clearly $\A\models \e_{\pi_1}'$.
On the other hand, we claim that there are no $R^\F,S^\F\subseteq A\times A$ that
both satisfy \eqref{Jemb1}--\eqref{Jemb2} and validate $\Phi_{\pi_1}$. 
Indeed, suppose otherwise.
As $a_0\le\dR a_0$, we have $(a_0,x)\in R^\F$ for some $x\le a_0$ by \eqref{Jemb2}.
As $a_0\not\le\dR g$, it follows by \eqref{Jemb1} that $x\ne g$, and so $(a_0,a_n)\in R^\F$
for some $n<\omega$.
As $a_n\le\dR a_n$, we have $(a_n,y)\in R^\F$ for some $y\le a_n$ by \eqref{Jemb2}.
As $a_n\not\le\dR g$, it follows by \eqref{Jemb1} that $y\ne g$, and so $(a_n,a_k)\in R^\F$
for some $k$ with $n\le k<\omega$. Thus, $\Phi_{\pi_1}$ implies that $(a_n,a_k)\in S^\F$, and so
$a_n\le \dS a_k$ by \eqref{Jemb1}, which is a contradiction. 
\end{proof}

The next example shows that the stability condition is essential in Theorem~\ref{thm:interpolant}.
\begin{example}\label{e:nonstable}\em 
Consider
 the unstable set $\{\pi_3\}$ with the forward-looking profile $\pi_3$ from Fig.~\ref{f:stable}.
It is easy to see that $ \{\e_{\pi_3}'\}\models_{\CA}\e$, where
 \[
\e=\bigl(
 \dS\bigl(q\land\dR (p\land\dR r)\bigr)\imp \dS\bigl(q\land\dR (p\land\dS r)\bigr)\bigr).
 \]
On the other hand, the \SLOa{} in  Fig.~\ref{f:pic7} validates $\e_{\pi_3}'$ but refutes
$\e$ when $q$ is $\{2\}$, $p$ is $\{3,4\}$, and $r$ is $\{5,6,7\}$. Therefore, $\SPi + \{\e_{\pi_3}'\}$ is not
\eqcomp.
\end{example}
\begin{figure}[ht]
\centering
\begin{tikzpicture}[scale=.7]
\draw [fill=gray] (3.5,1) circle [radius=.25];
\node at (3.5,1)  {\textcolor{white}{\bf\footnotesize 1}}; 
\draw [fill=gray] (6,5) circle [radius=.25];
\node at (6,5)  {\textcolor{white}{\bf\footnotesize 4}}; 
\draw [fill=gray] (3.5,3.5) circle [radius=.25];
\node at (3.5,3.5)  {\textcolor{white}{\bf\footnotesize 2}}; 
\draw [fill=gray] (1,5) circle [radius=.25];
\node at (1,5)  {\textcolor{white}{\bf\footnotesize 3}}; 
\draw [fill=gray] (1,7.5) circle [radius=.25];
\node at (1,7.5)  {\textcolor{white}{\bf\footnotesize 5}}; 
\draw [fill=gray] (4.5,7.5) circle [radius=.25];
\node at (4.5,7.5)  {\textcolor{white}{\bf\footnotesize 6}}; 
\draw [fill=gray] (7.5,7.5) circle [radius=.25];
\node at (7.5,7.5)  {\textcolor{white}{\bf\footnotesize 7}}; 
\draw [thick,->] (3.9,3.65) to (5.7,4.8);
\draw [thick,->] (3.2,3.65) to (1.3,4.8);
\draw [thick,->] (3.5,1.35) to (3.5,3.15);
\draw [thick,->] (1,5.35) to (1,7.15);
\draw [thick,->] (5.8,5.35) to (4.6,7.15);
\draw [thick,->] (6.2,5.35) to (7.4,7.15);

\node[right] at (2.9,2) {${}^S$};
\node[right] at (1.5,4) {${}^R$};
\node[right] at (4.4,3.8) {${}^S$};
\node[right] at (.4,6.3) {${}^R$};
\node[right] at (4.3,6.55) {${}^R$};
\node[right] at (7,6.55) {${}^S$};

\draw[thick,dotted,rounded corners=4] (5.3,5.6) -- (6.7,5.6) -- (6.7,4.4) -- (5.3,4.4) -- cycle;
\draw[thick,dotted,rounded corners=4] (2.8,4.1) -- (4.1,4.1) -- (4.1,2.9) -- (2.8,2.9) -- cycle;
\draw[thick,dotted,rounded corners=4] (2.8,1.6) -- (4.1,1.6) -- (4.1,.4) -- (2.8,.4) -- cycle;
\draw[thick,dotted,rounded corners=6] (2.6,4.3) -- (4.3,4.3) -- (4.3,.2) -- (2.6,.2) -- cycle;
\draw[thick,dotted,rounded corners=6] (5.2,5.8) -- (6.9,5.8) -- (6.9,4.4) -- (4.5,2.7) -- (2.4,2.7) -- (2.4,4.4) -- cycle;
\draw[thick,dotted,rounded corners=6] (.55,5.6) -- (5,6) -- (7.1,6) -- (7.1,4.2) -- (5.1,4.2) -- (5.1,4.5) -- (.55,4.5) -- cycle;
\draw[thick,dotted,rounded corners=8] (5.2,6.2) -- (7.3,6.2) -- (7.3,4.2) -- (4.5,2.2) -- (4.5,0) -- (2.2,0) -- (2.2,4.7) -- cycle;
\draw[thick,dotted,rounded corners=8] (.35,5.8) -- (5.3,6.4) -- (7.5,6.4) -- (7.5,4) -- (5.4,2.5) -- (1.9,2.5) -- (.35,4) -- cycle;
\draw[thick,dotted,rounded corners=8] (.5,8) -- (8,8) -- (8,6.9) -- (.5,6.9) -- cycle;
\draw[thick,dotted,rounded corners=10] (.1,8.3) -- (8.2,8.3) -- (8.2,-.2) -- (.1,-.2) -- cycle; 
\end{tikzpicture}
\caption{The \SLOa{} of Example~\ref{e:nonstable}.}\label{f:pic7}
\end{figure}

However, 
\hornequations{} with forward-looking but unstable profiles (such as $\etrans$) can still 
axiomatise \complex{} \sptheories.
Likewise, \sptheories{} axiomatised by \hornequations{} having non-forward-looking profiles such as $\esym$ can also be
 \eqcomp{} and even \complex:

\begin{theorem}\label{t:comp}
The following \sptheories{} are \complex, and so \eqcomp:
\begin{itemize}
\item[$(i)$] 
$\SPi + \{\esym\}$;

\item[$(ii)$] 
$\Esfive= \SPi + \Axsfive=\SPi + \Axsfive'$, where $\Axsfive=\{\erefl,\etrans,\esym\}$
and $\Axsfive'=\{\erefl,\etrans,\eeuc\}$.
%
\end{itemize}
\end{theorem}
\begin{proof}
$(i)$ Let $\A = (A, \land, \top, \d)$ be a SLO with  $\A \models \esym$.
For $a,b\in A$, let 
 \begin{equation}\label{JRsymm}
  (a,b)\in R^\F\  \   \Longleftrightarrow\ \ a\le \d b\ \mbox{ and }\ b\le \d a.
 \end{equation}
Then $R^\F$ is clearly symmetric and satisfies \eqref{Jemb1}. We show that it satisfies
\eqref{Jemb2} as well, and so, as shown in \S\ref{Jembed}, $\Aa$ embeds
to $\Fslo$, for $\F=(A,R^\F)$. To this end, fix some $a\in A$ and let $x$ be such that
$a \le \d x$. Then, by $\A \models \esym$, we have
\[
a = a\land \d x\le \d \bigl(\d a\land x\bigr).
\]
Let $b=\d a\land x$. Then $a\le \d b$, $b\le x$ and $b\le \d a$; so $(a,b)\in R^\F$, as required 
in \eqref{Jemb2}.

$(ii)$ 
It is easy to see that 
$\{\erefl,\eeuc\}\vdash_{\SLO}\esym$ and $\{\etrans,\esym\}\vdash_{\SLO}\eeuc$,
and so $\SPi + \Axsfive=\SPi + \Axsfive'$.
It is straightforward to check that if $\Aa\models\erefl$ and $\Aa\models\etrans$, then
the $R^\F$ defined in \eqref{JRsymm} is reflexive and transitive
as well.
%
%
%
%
%
Note that Jackson \cite{jackson2004} proves  \eqcomp ness of  $\Esfive$ by showing that 
$\Axsfive'\models_{\BAO}$ is conservative over  $\Axsfive'\models_{\SLO}$. 
\end{proof}

The next two examples show incomplete \sptheories{} axiomatised by \spequations{} with non-rooted, non-forward looking and unstable \tprof s.

\begin{example}\label{e:fourstep}\em 
The \spequation{}
$
\e =\bigl(q \land \d \d p \imp \d \d (p \land \d q)\bigr)
$
has the non-rooted \tprof\  %
\raisebox{-0.5mm}{
\begin{tikzpicture}[>=latex, point/.style={circle,draw=black,minimum size=1mm,inner sep=0pt},hm/.style={dashed}, xscale=1.0, yscale=1.0]\scriptsize
\node (u) at (0,0) [point,fill=black] {};
\node (w) at (1,0) [point,fill=black] {};
\node (v) at (2,0) [point,fill=black] {};
\draw[->,] (u) to node {} (w);
\draw[->,] (w) to node {} (v);
\draw[->,dashed,bend right,looseness=0.9] (v) to node {} (u);
\end{tikzpicture}
}
.
It is easy to see that $\{\e\}\models_{\CA} \d\d\d\d p \imp \d p$. On the other hand, the \SLOa{} in 
Fig.~\ref{f:pic2} validates $\e$ but refutes $\d\d\d\d p \imp \d p$ when $p$ is $\{5\}$. 
Therefore, $\SPi + \{\e\}$ is not \eqcomp.
\end{example}
\begin{figure}[ht]
\centering
\begin{tikzpicture}[scale=.6]
\draw [fill=gray] (1,2) circle [radius=.25];
\node at (1,2)  {\textcolor{white}{\bf\footnotesize 1}}; 
\draw [fill=gray] (3.5,2) circle [radius=.25];
\node at (3.5,2)  {\textcolor{white}{\bf\footnotesize 2}}; 
\draw [fill=gray] (6,2) circle [radius=.25];
\node at (6,2)  {\textcolor{white}{\bf\footnotesize 3}}; 
\draw [fill=gray] (8.5,2) circle [radius=.25];
\node at (8.5,2)  {\textcolor{white}{\bf\footnotesize 4}}; 
\draw [fill=gray] (11,2) circle [radius=.25];
\node at (11,2)  {\textcolor{white}{\bf\footnotesize 5}}; 
\draw [thick,->] (1.3,2) to [out=0, in=180] (3.2,2);
\draw [thick,->] (3.8,2) to [out=0, in=180] (5.7,2);
\draw [thick,->] (6.3,2) to [out=0, in=180] (8.2,2);
\draw [thick,->] (8.8,2) to [out=0, in=180] (10.7,2);
\draw [thick,->] (5.9,1.7) to [out=200, in=-20] (1.1,1.7);
\draw [thick,->] (10.9,1.7) to [out=200, in=-20] (6.1,1.7);
\draw [thick,->] (8.4,2.3) to [out=-200, in=20] (3.6,2.3);
\draw[thick,dotted,rounded corners=4] (5.5,2.5) -- (6.5,2.5) -- (6.5,1.5) -- (5.5,1.5) -- cycle;
\draw[thick,dotted,rounded corners=4] (8,2.5) -- (9,2.5) -- (9,1.5) -- (8,1.5) -- cycle;
\draw[thick,dotted,rounded corners=4] (10.5,2.5) -- (11.5,2.5) -- (11.5,1.5) -- (10.5,1.5) -- cycle;
\draw[thick,dotted,rounded corners=4] (.5,2.8) -- (.5,.7) -- (9.3,.7) -- (9.3,2.8) -- (7.7,2.8) -- (7.7,1.1) -- (1.5,1.1) -- (1.5,2.8) -- cycle;
\draw[thick,dotted,rounded corners=4] (3,3.5) -- (3,1.5) -- (4,1.5) -- (4,3.1) -- (10.3,3.1) -- (10.2,1.1) -- (11.8,1.1) -- (11.8,3.5) -- cycle;
\draw[thick,dotted,rounded corners=8] (0,4) -- (12.2,4) -- (12.2,0.3) -- (0,0.3) -- cycle;
\end{tikzpicture}
\caption{The \SLOa{} of Example~\ref{e:fourstep}.}\label{f:pic2}
\end{figure}

\begin{example}\label{e:eucl}\em
Consider next the \spequation{} $\eeuc$ (see Table~\ref{t:default}).
It is not hard to see that $\{\eeuc\}\models_{\CA} \d\d p \land \d q \imp \d (q \land \d p)$.  On the other hand, the \SLOa{}  in Fig.~\ref{f:pic6_4}~(a)
validates $\eeuc$ but refutes 
$\d\d p \land \d q \imp \d (q \land \d p)$ 
when $p$ is $ \{5\}$ and $q$ is $ \{3,4\}$.
Therefore, $\SPi + \{\eeuc\}$ is not \eqcomp.
\end{example}
\begin{figure}[ht]
\centering
\begin{tikzpicture}[scale=.6]
\tikzset{every loop/.style={thick,min distance=8mm,in=45,out=135,looseness=8}}
\tikzset{place/.style={circle,thin,draw=white,fill=white,scale=1.2}}

\draw [fill=gray] (3.5,1) circle [radius=.25];
\node at (3.5,1)  {\textcolor{white}{\bf\footnotesize 1}}; 
\node[place] (foo5) at (6,6) {}; 
\draw [fill=gray] (6,6) circle [radius=.25];
\node at (6,6)  {\textcolor{white}{\bf\footnotesize 5}}; 
\node[place] (foo2) at (3.5,3.5) {}; 
\draw [fill=gray] (3.5,3.5) circle [radius=.25];
\node at (3.5,3.5)  {\textcolor{white}{\bf\footnotesize 2}}; 
\node[place] (foo3) at (1,3.5) {}; 
\draw [fill=gray] (1,3.5) circle [radius=.25];
\node at (1,3.5)  {\textcolor{white}{\bf\footnotesize 3}}; 
\node[place] (foo4) at (1,6) {}; 
\draw [fill=gray] (1,6) circle [radius=.25];
\node at (1,6)  {\textcolor{white}{\bf\footnotesize 4}}; 
\draw [thick,->] (3.9,3.6) to [out=45, in=-135] (5.8,5.7);
\draw [thick,->] (6,5.6) to [out=-110, in=20] (3.9,3.5);
\draw [thick,->] (3.2,3.6) to [out=135, in=-45] (1.2,5.7);
\draw [thick,->] (1,5.6) to [out=-70, in=160] (3.1,3.5);
\draw [thick,->] (3.15,1.2) to [out=135, in=-45] (1.15,3.2);
\draw [thick,->] (3.5,1.35) to [out=90, in=-90] (3.5,3.15);
\draw [thick,->] (1.35,3.3) to [out=0, in=180] (3.15,3.3);
\draw [thick,->] (1.35,6.1) to [out=10, in=170] (5.65,6.1);
\draw [thick,->] (5.65,5.9) to [out=-170, in=-10] (1.35,5.9);
\path[->] (foo5) edge [loop] node {} ();
\path[->] (foo3) edge [loop] node {} ();
\path[->] (foo2) edge [loop] node {} ();
\path[->] (foo4) edge [loop] node {} ();
\draw[thick,dotted,rounded corners=4] (.2,7.3) -- (1.8,7.3) -- (1.8,5.4) -- (.2,5.4) -- cycle;
\draw[thick,dotted,rounded corners=4] (5.2,7.3) -- (6.8,7.3) -- (6.8,5.4) -- (5.2,5.4) -- cycle;;
\draw[thick,dotted,rounded corners=6] (-.2,7.7) -- (2.2,7.7) -- (2.2,2.8) -- (-.2,2.8) -- cycle;
\draw[thick,dotted,rounded corners=8] (-.6,8.1) -- (7.2,8.1) -- (7.2,5) -- (5.1,2.9) -- (2.5,2.9) -- (-.6,6) -- cycle;
\draw[thick,dotted,rounded corners=10] (-1,8.5) -- (7.6,8.5) -- (7.6,.2) -- (-1,.2) -- cycle; 

\node at (3.4,-.9) {(a)};
\end{tikzpicture}
\hspace*{1.8cm}
\begin{tikzpicture}[scale=.65,baseline={(0,-1.8)}]
\tikzset{every loop/.style={thick,min distance=15mm,in=45,out=135,looseness=10}}
\tikzset{place/.style={circle,thin,draw=white,fill=white,scale=1.2}}

\draw [fill=gray] (3.5,1) circle [radius=.25];
\node at (3.5,1)  {\textcolor{white}{\bf\footnotesize 1}}; 
\node[place] (foo) at (3.5,3.5) {}; 
\draw [fill=gray] (3.5,3.5) circle [radius=.25];
\node at (3.5,3.5)  {\textcolor{white}{\bf\footnotesize 2}}; 
\draw [fill=gray] (1,3.5) circle [radius=.25];
\node at (1,3.5)  {\textcolor{white}{\bf\footnotesize 3}}; 
\draw [thick,->] (3.25,1.25) to [out=135, in=-45] (1.25,3.25);
\draw [thick,->] (3.5,1.35) to [out=90, in=-90] (3.5,3.15);
\path[->] (foo) edge [loop right] node {\ ${}^S$} ();
\draw[thick,dotted,rounded corners=4] (2.7,5) -- (4.5,5) -- (4.5,2.8) -- (2.7,2.8) -- cycle;
\draw[thick,dotted,rounded corners=4] (2.7,1.6) -- (4.5,1.6) -- (4.5,.5) -- (2.7,.5) -- cycle;
\draw[thick,dotted,rounded corners=6] (.2,5.3) -- (4.8,5.3) -- (4.8,2.4) -- (.2,2.4) -- cycle;
\draw[thick,dotted,rounded corners=6] (2.4,5.6) -- (5.1,5.6) -- (5.1,.2) -- (2.4,.2) -- cycle;
\draw[thick,dotted,rounded corners=10] (-.6,5.9) -- (5.5,5.9) -- (5.5,-.1) -- (-.6,-.1) -- cycle; 

\node[right] at (1.8,1.9) {${}^R$};
\node[right] at (3.4,1.8) {${}^S$};
\node at (2.7,-2.7) {(b)};
\end{tikzpicture}
\caption{The \SLOa s of Examples~\ref{e:eucl} and \ref{e:bekl}.}\label{f:pic6_4}
\end{figure}
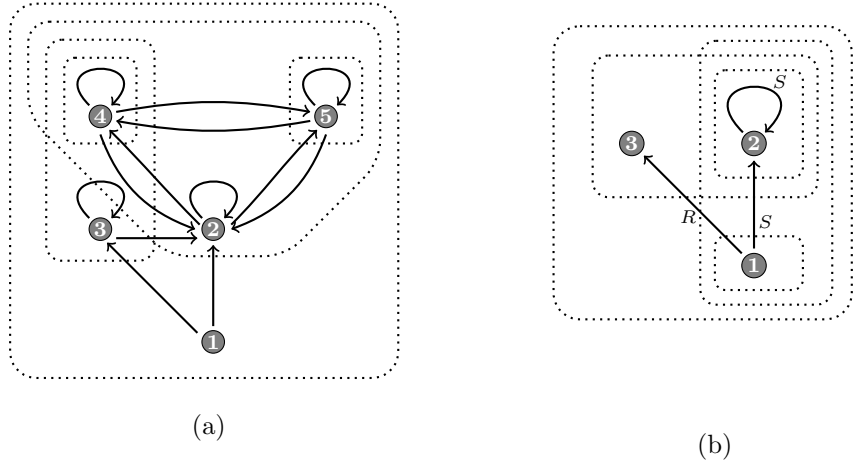

\begin{example}[\cite{Beklemishev15}]\label{e:bekl}\em 
Consider
$
\e=\bigl(  \dS p \imp \dS (p \land \dS p)\bigr)
$
 with non-rooted \tprof{} \hspace*{-6mm} \raisebox{-3mm}{
\begin{tikzpicture}[>=latex, point/.style={circle,draw=black,minimum size=1mm,inner sep=0pt},hm/.style={dashed}, xscale=1.0, yscale=1.0]\scriptsize
\node (u) at (0,0) [point,fill=black] {};
\node (v) at (1,0) [point,fill=black] {};
\draw[->,] (u) to node [label=above:{\qquad\qquad\qquad\scriptsize$S$}] {} (v);
\draw[->,loop,dashed,looseness=20] (v) to node  {} (v);
\node at (.4,-.2) {\scriptsize$S$};
\end{tikzpicture}}
\hspace*{-3mm} and 
$\e' = (\dR p \imp \dS p)$ with rooted \tprof{} \raisebox{-3mm}{
\begin{tikzpicture}[>=latex, point/.style={circle,draw=black,minimum size=1mm,inner sep=0pt},hm/.style={dashed}, xscale=1.0, yscale=1.0]\scriptsize
\node (u) at (0,0) [point,fill=black] {};
\node (v) at (1,0) [point,fill=black] {};
\draw[->,] (u) to node [label=below:{\scriptsize$R$}] {} (v);
\draw[->,dashed,bend left,looseness=0.9] (u) to node [label=above:{\scriptsize$S$}] {} (v);
\end{tikzpicture}} %
.
Then $\{\e,\e'\}\models_{\CA}\dR p \imp \dR(p \land \dS p)$. However, the \SLOa{} in Fig.~\ref{f:pic6_4}~(b)  
validates both $\e$ and $\e'$, but refutes $\dR p \imp \dR(p \land \dS p)$ when $p$ is $\{2,3\}$.
Therefore, $\SPi + \{\e,\e'\}$ is not \eqcomp.
%
\end{example}

This example generalises to the following theorem: 

\begin{theorem}\label{thm:genericinc}
For any \hornequation{} $\e$ with a non-rooted \tprof, there is a \hornequation{} $\e'$ with a rooted 
\tprof\  \textup{(}and a fresh diamond operator\textup{)} such that the \sptheory{} $\SPi + \{\e, \e'\}$ is not \eqcomp.
\end{theorem}
\begin{proof}
Suppose $\pi = (\G,\newrel,u,v)$ is the non-rooted profile of $\e = (\sigma \imp \tau)$. Denote by $r$ the root of 
$\G=(\Delta,R^\G)_{R\in\R}$ and by $w$ the successor of $r$ on the branch from $r$ to $u$ with, say,  $(r,w)\in R^\G$ for some $R\in \R$. 
Define $\G'$ to be a tree whose points are copies $x'$ of the points $x$ in $\G$, and the arrows between them are the same as in $\G$ except that we replace the $R^{\G'}$-arrow from $r'$ to $w'$ with an $R_\dag^{\G'}$-arrow, for some fresh $R_\dag\notin\R$. 
Let $\pi' = (\G',\newrel,u',v')$ and let $\e' = (\d_{\!R_\dag\,} p \imp \dR p)$. It is readily seen that any 
frame validating $\{\e,\e'\}$ also validates the \spequation{} $\e_{\pi'}$. 

We now construct a \SLOa\ $\Aa$ validating $\{\e,\e'\}$ but refuting $\e_{\pi'}$.
Consider the Horn closure $\pi(\G)$ of $\G$. Clearly, $\pi(\G) \models \Phi_\pi$, from which $\pi(\G) \models \Psi_{\e}$ and 
\begin{equation}\label{psieholds}
\pi(\G)\models\e.
\end{equation}
Now let $\F$ be the result of merging the roots $r$ of $\pi(\G)$
 and $r'$ of $\G'$ into a single point. We define $\Aa$ 
as the subalgebra of $\Fslo$ with domain 
$$
A = \{X \cup X' \mid X \subseteq \G,\ X' \subseteq \G',\ X' \subseteq' X\},
$$
where $X' \subseteq' X$ iff $x'\in X'$ implies $x \in X$. Then $\emptyset$ and the domain of $\F$ clearly belong to $A$. Also, $A$ is closed under intersections because we clearly have  $(X \cup X') \cap (Y \cup Y') = (X \cap Y) \cup (X' \cap Y')$; here we use the fact that $r=r'$. Furthermore, $\d_{\!R_\dag\,}^+ (X \cup X') = \emptyset$ if $w' \notin X'$,  
$\d_{\!R_\dag\,}^+ (X \cup X') = \{r\}$ if $w'\in X'$, and $\d_{\!Q\,}^+ (X \cup X') = \d_{\!Q\,}^+ X \cup \d_{\!Q\,}^+ X'$ with $\d_{\!Q\,}^+ X' \subseteq' \d_{\!Q\,}^+ X$, for any $Q$ different from $R_\dag$. Thus, $\A$ is a \SLOa.
Observe also that, for every $X\cup X'\in A$, we have 
$\d_{\!R_\dag}^+(X\cup X')\subseteq\d_{\!R}^+(X\cup X')$, and so $\A\models\e'$.

Next, we show that $\A \not\models \e_{\pi'}$. Indeed, 
suppose $\e_{\pi'}=(\alpha\imp\d_{\!R_\dag\,}\beta)$ (cf.\ \eqref{epi}).
We have $\G' \not \models \e_{\pi'}$ by \eqref{pinote}, and so there exist a Kripke model $\M=(\G',\V)$ 
and some $w$ in it such that 
$\M,w\kmodels\alpha$ but $\M,w\not\kmodels\d_{\!R_\dag\,}\beta$.
We define a valuation $\vala$ in $\A$  
by taking
\[ 
\vala(p)=\V(p)\cup \{ x \mid x' \in \V(p)\},\ \mbox{ for every variable $p$.}
\]
It is easy to see that 
$\varrho[\vala]\cap\Delta=\{w\mid \M,w\kmodels\varrho\}$, 
for every \spterm{} $\varrho$.
Then 
$\alpha[\vala]\supseteq\{w\mid \M,w\kmodels\alpha\}$ 
and $(\d_{\!R_\dag\,}\beta)[\vala]=
\d_{\!R_\dag\,}^+\bigl(\beta[a]\bigr)=\d_{\!R_\dag\,}^+\bigl(\{w\mid\M,w\kmodels\beta\}\bigr)=
\{w\mid \M,w\kmodels \d_{\!R_\dag\,}\beta\}$, from which $\A\not\models \e_{\pi'}[\vala]$.

It remains to establish $\A \models \e$. As $\A$ is a subalgebra of $\Fslo$, it is enough
to show that $\F\models\e$. Take any Kripke model $\M=(\F,\V)$ and
suppose 
$\M,x\kmodels\sigma$, 
for some point $x$ in $\F$. By Proposition~\ref{p:hom}, 
there is a homomorphism $h \colon \M_\sigma \to \M$ with $h(r_\sigma)=x$ for the root $r_\sigma$ of
 $\T_\sigma$. We show that 
 $\M,x\kmodels\tau$. 
 Indeed,
note first that $x$ cannot be a non-root point in $\G'$ because otherwise we would have a homomorphism 
$f\colon\T_\sigma\to\G$ with $f(r_\sigma)\ne r$, contradicting Proposition~\ref{p:eandpi}~$(ii)$.
Thus, $x$ is a point in $\pi(\G)$. We define a map $h' \colon \T_\sigma \to \pi(\G)$ by taking 
\[
h'(y)=\left\{
\begin{array}{ll}
h(y), & \mbox{ if $h(y)$ is in $\pi(\G)$},\\
z, & \mbox{ if $h(y)=z'$ for some $z'$ in $\G'$}.
\end{array}
\right.
\]
As $\sigma$ does not contain $\d_{\!R_\dag}$, it is easy
to see that $h'$ is a homomorphism from $\M_\sigma$ to the Kripke model $\M^-=\bigl(\pi(\G), \V\mathop{\restriction} \pi(\G)\bigr)$ with $h'(r_\sigma)=h(r_\sigma)=x$, and so 
$\M^-,x\kmodels\sigma$ 
by Proposition~\ref{p:hom}.
Then we have 
$\M^-,x\kmodels\tau$ and so $\M,x\kmodels\tau$
by \eqref{psieholds}, as required.
\end{proof}


\subsection{Universal Horn correspondents with equality.} 

Example~\ref{e:simple} showed that the \sptheory{} $\SPi + \{\d p\imp p\}$ with the correspondent 
 \[
\forall x,y \, (R(x,y) \to (x = y))
\]
 is incomplete. It is easy to find an extension of this \sptheory{} that is \complex:
 
 \begin{theorem}\label{t:simple1}
  The \sptheory{} 
   $\SPi + \{\erefl,\d p\imp p\}=\SPi + (\Axsfour\cup\{\d p\imp p\})=$ $\SPi + (\Axsfive\cup\{\d p\imp p\})=\SPi + (\Axsfive'\cup\{\d p\imp p\})$
 is \complex, and so \eqcomp. 
\end{theorem}
\begin{proof}
It is easy to see that $\{\d p\imp p\}\vdash_{\SLO}\etrans$ and $\{\erefl,\d p\imp p\}\vdash_{\SLO}\eeuc$,
and so all four \sptheories{} are the same. 
The correspondent of this \sptheory{} is 
\begin{equation}\label{idarrow}
\Phi=\ \forall x, y \, \bigl(R(x,y)\leftrightarrow x=y\bigr).
\end{equation}
Let $\A = (A, \land,  \top, \d)$ be a \SLOa\ with \mbox{$\A \models \{\erefl,\d p\imp p\}$.} 
For all $a,b\in A$, we set $(a,b)\in R^\F$ iff $a=b$.
Then $R^\F$ clearly satisfies $\Phi$,  \eqref{Jemb1} and  \eqref{Jemb2}.
So, as is shown in \S\ref{Jembed}, $\Aa$ embeds to $\Fslo$ for $\F=(A,R^\F)$.
\end{proof}


Our next example is the \spequation{}  
\[
\efun=\bigl( \d p \land \d q \imp \d(p\land q)\bigr)
\]
saying that $\d$ is a semilattice homomorphism.%
\!\footnote{In \cite{jackson2004}, any $\A\in\SLO_{\Esfour}$ validating $\efun$ is called \emph{entropic\/}.}
The first-order correspondent of $\efun$ is \emph{functionality\/}:
\[
\forall x,y,z \, \bigl(R(x,y) \land R(x,z) \to y=z\bigr).
\]
It is easy to see that 
$\{\erefl,\d p\imp p\}\vdash_{\SLO}\efun$ and $\{\erefl,\efun\}\vdash_{\SLO}\eeuc$,
and so
\begin{multline*}
\SPi + \{\erefl,\efun,\d p\imp p\} = 
\SPi + (\Axsfour\cup\{\efun,\d p\imp p\}) =\\
\SPi + (\Axsfive\cup\{\efun,\d p\imp p\})=
\SPi + (\Axsfive'\cup\{\efun,\d p\imp p\})
\end{multline*}
is the same \sptheory{} as in Theorem~\ref{t:simple1}.

\begin{theorem}\label{t:fun}
\begin{itemize}
\item[$(i)$]
The \sptheory{} $\SPi + \{\efun\}$ is \complex, and so \eqcomp.
\end{itemize}
On the other hand, the following \sptheories{} are incomplete:
\begin{itemize}
\item[$(ii)$]
$\SPi + \{\d p\imp p,\efun\}$;

\item[$(iii)$]
$\SPi + \{\erefl,\efun\}=\SPi + \{\erefl,\eeuc,\efun\}$;

\item[$(iv)$]
$\SPi + (\Axsfour\cup\{\efun\})=
\SPi + (\Axsfive\cup\{\efun\})=
\SPi + (\Axsfive'\cup\{\efun\})$;

\item[$(v)$]
$\SPi + \{\esym,\efun\}$.
\end{itemize}
\end{theorem}

\begin{proof}
$(i)$
Let $\A = (A, \land, \top, \d)$ be a \SLOa\ such that $\A \models \efun$ and let $\FA$ be the set  of all filters of $\Aa$.
We claim that in this case
 $\d^{-1}[U]$ is either empty or a filter, for every $U\in\FA$. Indeed, $\d^{-1}[U]$ is up-closed by the
 monotonicity of $\d$, and $\land$-closed by $\efun$. Now
we set, for $U,V\in\FA$,
\[
(U,V)\in R^\G\quad\Longleftrightarrow\quad V=\d^{-1}[U].
\]
Then $R^\G$ is clearly functional and satisfies \eqref{emb1}. It is readily seen that it satisfies \eqref{emb2} as well. So, as shown in \S\ref{Tembed}, $\Aa$ embeds into
$\Gslo$ for $\G=(\FA,R^\G)$.


$(ii)$ 
The proof in Example~\ref{e:simple} again works. Note that the \SLOa{} in Fig.~\ref{f:fivepointslo} shows that the \sptheories{} $\SPi + \{\d p\imp p,\efun\}$ and $\SPi + \{\d p\imp p\}$ are not the same.
\begin{figure}[ht]
\centering
\begin{tikzpicture}[scale=.5]
\draw [very thick] (1,1) circle [radius=0.15];
\draw [fill] (1,3) circle [radius=0.1];
\draw [fill] (-1,4) circle [radius=0.1];
\draw [fill] (3,4) circle [radius=0.1];
\draw [very thick] (1,5) circle [radius=0.15];
\draw [very thin] (1,1.3) -- (1,2.8);
\draw [very thin] (1.2,3.1) -- (2.8,3.9);
\draw [very thin] (.8,3.1) -- (-.8,3.9);
\draw [very thin] (1.3,4.85) -- (2.8,4.1);
\draw [very thin] (.7,4.85) -- (-.8,4.1);

\draw [very thick,->] (1.2,2.7) to [out=-65, in=65] (1.3,1);
\draw [very thick,->] (3,3.8) to [out=-100, in=-15] (1.2,3);
\draw [very thick,->] (-1,3.8) to [out=-80, in=-165] (.8,3);
\end{tikzpicture}
\caption{A \SLOa{} showing that $\{\d p\imp p\}\not\vdash_{\SLO}\efun$.}\label{f:fivepointslo}
\end{figure}

$(iii)$
The correspondent of this \sptheory{} is $\Phi$ in \eqref{idarrow}. 
So it is easy to see that $\{\erefl,\efun\}\models_{\CA}\etrans$.
On the other hand, take the \SLOa\ $\Aa$ with 3 elements 
 $b\le a\le \top$, $\d b =a$ and $\d a=\d \top=\top$.
Then $\Aa\models\erefl$ and $\Aa\models\efun$, but $\d\d b\not\le \d b$.

$(iv)$ 
The correspondent of this \sptheory{} is again $\Phi$ in \eqref{idarrow}. 
So it is easy to see that $\Axsfour\cup\{\efun\}\models_{\CA}\d p\imp p$.
On the other hand, take the \SLOa\ $\Aa$ with 3 elements 
 $b\le a\le \top$, $\d b =b$ and $\d a=\d \top=\top$.
Then $\Aa\models \Axsfour \cup \{\efun\}$, but $\d a=\top\not\le a$.

$(v)$ 
It is easy to see that $\d\d p\imp p$ is valid in any symmetric and functional frame, and so
$\{\esym,\efun\}\models_{\CA}\d\d p\imp p$. 
On the other hand, in the \SLOa\ $\Aa$ from item $(iv)$, $\d\d a=\top\not\leq a$. 
\end{proof}

\begin{remark}\label{r:noJ}\em
Theorem~\ref{t:fun}~$(i)$ cannot be proved using the simpler embedding of \S\ref{Jembed}.
Indeed, take the \SLOa\ $\A = (A, \land, \top,\d)$ where 
$A=\{a_n\mid n< \omega\}$, $a_n\land a_m=a_n$ whenever $n\geq m$ (and so 
$\top=a_0$), 
and $\d a_n=\top$ for all $n<\omega$. Then $\Aa\models\efun$ clearly holds.
On the other hand, we claim that there is no functional $R^\F\subseteq A\times A$ satisfying \eqref{Jemb2}.
Indeed, suppose $R^\F$ satisfies \eqref{Jemb2}. Since for every $n<\omega$, we have
$\top\le\d a_n$, it follows from \eqref{Jemb2} that, for any $n<\omega$, there exists
$m\geq n$ such that $(\top,a_m)\in R^\F$, and so $R^\F$ is not functional.
\end{remark}

\subsection{Negative universal Horn correspondents.} 

Finally, we discuss \spequations{} with Horn correspondents of the form `false' and
`something implies false'\!.
Recall that Example~\ref{e:simple2} showed that 
the \sptheory{} $\SPi + \{\d p\imp \d q\}$ with the correspondent $R=\emptyset$---or $\forall x,y\,(R(x,y) \to \bot)$, to be more precise---is incomplete.
The next theorem gives an incomplete extension of this \sptheory:

\begin{theorem}\label{t:simple2}
The \sptheory{}
$\SPi + \{\erefl,\d p\imp \d q\}=\SPi+(\Axsfour \cup \{\d p\imp \d q\}) =\SPi+(\Axsfive \cup \{\d p\imp \d q\})$
is incomplete.  
\end{theorem}

\begin{proof}
It is easy to see that $\{\d p\imp \d q\}\vdash_{\SLO}\etrans$ and $\{\d p\imp \d q\}\vdash_{\SLO}\esym$,
and so all three \sptheories{} are the same.
As there is no frame validating $\d p\imp \d q$, we have
$\{\erefl,\d p\imp \d q\}\models_{\CA}\d \top\imp p$.
On the other hand, we have $\{\erefl,\d p\imp \d q\} \not\models_{\SLO} \d\top\imp p$, as the \SLOa\ $\Aa$ with 2 elements 
$a\le \top$ such that $\d a=\d \top=\top$ validates both $\erefl$ and $\d p\imp \d q$, but refutes $\d\top\imp p$, since $\d\top=\top\not\leq a$. 
\end{proof}

Of course, not every \sptheory{} without frames is incomplete.
We call an \sptheory{} $\SPL$ \emph{trivial} if $(p\imp q)\in \SPL$. 
Then we clearly have:

\begin{proposition}\label{p:trivial}
Every trivial \sptheory{} is \eqcomp.
\end{proposition}

%
The 
following two theorems imply that \sptheories{} axiomatised by \spequations{} of the form \mbox{$\dR\dS p\imp q$} 
(with negative universal Horn correspondent $\forall x,y,z\,\bigl((R(x,y)\land S(y,z) \to \bot\bigr)$)
behave differently in the uni- and multi-modal cases.

\begin{theorem}\label{t:unicomp}
 $\SPi + \{\d^n p\imp q\}$ is \complex, and so \eqcomp, for any $n\geq 1$. 
\end{theorem}

\begin{proof}
The correspondent of $\{\d^n p\imp q\}$ is `there is no $R$-chain of length $n$'\!.
Let $\A = (A, \land,  \top, \d)$ be a \SLOa\ with $\A\models\d^n p\imp q$. Then $\d^n\top$ is the $\le$-smallest
element in $\A$. If $|A|=1$ then $\A$ is clearly embeddable into $\Fslo$ of any frame $\F$.
So let $|A|>1$ and $A^-=A \setminus\{\d^n\top\}$.
For any $a,b\in A^-$, let $(a,b)\in R^\F$ iff $a\le\d b$, and let $\F=(A^-,R^\F)$. Then $R^\F$ clearly satisfies \eqref{Jemb1}. 
As $\{\d^n p\imp q\}\vdash_{\SLO}\d^{n+1}\top\imp\d^n\top$, we also have the following analogue of
 \eqref{Jemb2}: 
 \[
 \forall a\in A^-,b\in A\ \bigl[ a\le \dR b \   \Rightarrow \ \exists c\in A^-\, \bigl(c\le b \ \mbox{ and }\  (a,c)\in R^\F\bigr)\bigr].
 \]
A proof similar to the one in \S\ref{Jembed} shows that the map
 $\eta(a)=\{b\in A^-\mid b\le a\}$ embeds $\A$ into $\Fslo$. Now, suppose $\F$ contains an
 $R^\F$-chain of length $n$, that is, there are $a_0,a_1,\dots,a_n\in A^-$ with $(a_i,a_{i+1})\in R^\F$ for 
 all $i<n$. Then $a_0\le\d^n a_n$, and so $a_0=\d^n\top$, contrary to $a_0\in A^-$.
\end{proof}

\begin{theorem}\label{t:multiincomp}
Let $\sigma$ be an \spterm{} containing $\dS$ but not $\dR$,
and let $q$ be a propositional
variable not occurring in $\sigma$. Then $\SPi + \{\dR\sigma\imp q\}$ is incomplete.
\end{theorem}

\begin{proof}
It is easy to see that $\{\dR\sigma\imp q\}\models_{\CA}\dS\dR\sigma\imp\dR\sigma$ for any $\dS$
 in $\sigma$. 
On the other hand, take the \SLOa\ $\Aa$ with 2 elements $a\le \top$ such that $\dR a=\dR \top=a$
and $\dS a=\dS\top=\top$, for $S\ne R$.
Then $\A$ validates $\dR\sigma\imp q$ but refutes $\dS\dR\sigma\imp\dR\sigma$, and so  
$\SPi + \{\dR\sigma\imp q\}$ is incomplete.
\end{proof}


\section{Completeness of \sptheories{} with existential  correspondents}\label{sec:exist}
We now extend Theorem~\ref{Th:2} to \spequations{} whose 
correspondents contain existential quantifiers (but no disjunction) on the right-hand side of implication.

It is not hard to see, using distributivity of $\land$ over $\lor$, that the correspondent $\Psi_{\e}$
of an \spequation{} $\e=(\sigma\imp\tau)$ (see \eqref{eq:corr1} and \eqref{globalcorr}) can be equivalently rewritten as
\begin{multline}\label{eq:corr2}
\Psi_{\e} ~=~ \forall \varel{v}_0,\dots,\varel{v}_{n_\sigma} \, \Big( \bigwedge_{\stackrel{i,j\leq n_{\sigma},\,R\in\R}{(v_i,v_j)\in R_\sigma}} R(\varel{v}_i,\varel{v}_j)  \to \\
\exists \varel{u}_0,\dots,\varel{u}_{n_\tau} \, \big( (\varel{v}_0 = \varel{u}_0) \land \!\!\bigwedge_{\stackrel{i,j\leq n_\tau,\,R\in\R}{(u_i,u_j)\in R_\tau}}\!\! R(\varel{u}_i,\varel{u}_j) 
\land  \bigvee_{f\in Y_{\sigma,\tau}} \ \bigwedge_{\stackrel{(u_i,p)\in X_\tau}{f(u_i,p)=v_j}} \bigl(\varel{u}_i=\varel{v}_j \big)\Big),
\end{multline}
where $X_\tau=\{(u_i,p)\mid p\mbox{ is a variable and }u_i\in\V_\tau(p)\}$  and
\[
Y_{\sigma,\tau}=\{f\mid f\colon X_\tau\to W_\sigma, f(u_i,p)\in\V_\sigma(p)\mbox{ for all }(u_i,p)\in X_\tau\}.
\]
If the right-hand side of $\Psi_{\e}$ does not contain any disjunction, this means that $Y_{\sigma,\tau}$
consists of a single `choice' function $f$. 

\begin{theorem}\label{t:exist}
Any \sptheory{} axiomatised by \spequations{} $\sigma \imp \tau$ such that 
\begin{itemize}
\item[$(i)$] every variable in $\tau$ occurs in $\sigma$ exactly once,

\item[$(ii)$] 
$|W_\tau|\geq 2$ and all points in any $\V_\tau(p)$ are leaves of $\T_\tau$,

\item[$(iii)$] 
$\V_\tau(p)\cap\V_\tau(q)=\emptyset$ whenever $p\ne q$
\end{itemize}
is complex, and so \eqcomp.
\end{theorem}
\begin{proof}
Suppose that $\e = (\sigma \imp \tau)$ and the points of $W_\sigma= \{v_0, \dots, v_{n_\sigma}\}$ and 
$W_\tau = \{u_0, \dots, u_{n_\tau}\}$ are listed so
that $(v_i, v_j)\in R_\sigma$ or $(u_i, u_j)\in R_\tau$ imply
$i < j$ (and so $v_0 = r_\sigma$ and $u_0 = r_\tau$). 
By \eqref{eq:corr2} and $(i)$, 
\begin{multline*}
\Psi_{\e} ~=~ \forall \varel{v}_0,\dots,\varel{v}_{n_\sigma} \, \Big( \bigwedge_{\stackrel{i,j\leq n_{\sigma},\,R\in\R}{(v_i,v_j)\in R_\sigma}} R(\varel{v}_i,\varel{v}_j)  \to \\
\exists \varel{u}_0,\dots,\varel{u}_{n_\tau} \, \big( (\varel{v}_0 = \varel{u}_0) \land \!\!\bigwedge_{\stackrel{i,j\leq n_\tau,\,R\in\R}{(u_i,u_j)\in R_\tau}}\!\! R(\varel{u}_i,\varel{u}_j) 
\land  \bigwedge_{\stackrel{u_i\in \V_\tau(p)}{\V_\sigma(p)=\{v_j\}}} \bigl(\varel{u}_i=\varel{v}_j \big)\Big).
\end{multline*}
Let $\A = (A, \land, \top, \dR)_{R\in \R}$ be a \SLOa\  validating $\e$.
It is shown in \S\ref{Jembed} that $\Aa$ can be embedded
into $\Fslo$ for the frame $\F=(A,R^\F)_{R\in\R}$ with $R^\F$ defined by 
\eqref{JRclassic}. We claim that $\F\models\Psi_{\e}$.
Indeed, for each point $v_i$ in $\T_\sigma$, take some $a_{i}\in A$ such that
$(v_i,v_j)\in R_\sigma$ imply $(a_{i},a_{j})\in R^\F$, that is, $a_i\le\dR a_j$. We need to find $b_0,\dots,b_m\in A$ 
such that $b_0=a_0$ and the following properties hold, for $j=0,\dots,m$:
\begin{itemize}
\item[(a)] 
$b_j=a_k$ if $u_j \in\V_\tau(p)$ and $\V_\sigma(p)=\{v_k\}$, for some variable $p$;
\item[(b)] 
if $(u_j,u_k)\in R_\tau$ then $(b_j,b_k)\in R^\F$, that is, $b_j\le \dR b_k$.
\end{itemize}
We define inductively $b_m, \dots, b_0$ by taking:
\[
b_j=\left\{
\begin{array}{ll}
a_k, & \mbox{if $u_j\in\V_\tau(p)$, $\V_\sigma(p)=\{v_k\}$ for some $p$},\\
\top, & \mbox{if $u_j$ is a leaf  and there is no $p$} \mbox{ with $u_j\in\V_\tau(p)$},\\
\d_{R^1}b_{k_1}\land\dots\land\d_{R^\ell}b_{k_\ell}, 
& \mbox{if $j\ne 0$ and $u_j$ has $\ell>0$ successors}\\
& \hspace*{1.5cm}\mbox{$u_{k_1},\dots,u_{k_\ell}$ with $(u_j,u_{k_i})\in R^i_\tau$},\\
a_0, & \mbox{if $j=0$}.
\end{array}
\right.
\]
By $(i)$--$(iii)$, $b_j$ is well-defined. We clearly have (a) and (b), for $j\ne 0$. 
To show (b) for $j=0$,
take the following valuation $\vala$ in $\Aa$, for all variables $p$:
\[
\vala(p)=\left\{
\begin{array}{ll}
a_k, & \mbox{if $p$ occurs in $\sigma$ and $\V_\sigma(p)=\{v_k\}$},\\
\top, & \mbox{otherwise}.
\end{array}
\right.
\]
By (i), $\vala$ is well-defined.
Let $\tau_j=\trm_{u_j}^{\M_\tau}$, for $j=m,\dots,1$ (cf.\ \S\ref{mterms}). 
We prove  that 
\begin{equation}\label{righteq}
\tau_j[\vala]\le b_j,\qquad\mbox{for every $j=m,\dots,1$}.
\end{equation}
Indeed, if $u_j$ is a leaf, then either $\tau_j=\top=b_j$, or $\tau_j=p$ for some variable $p$,
and so $\tau_j[\vala]=a_k=b_j$ for $k$ with $\V_\sigma(p)=\{v_k\}$. 
Now suppose inductively that, for some $j\geq 1$, \eqref{righteq} holds for every $k$ with $j<k\leq m$.
If $u_j$ has $\ell>0$ successors $u_{k_1},\dots,u_{k_\ell}$ with $(u_j,u_{k_i})\in R^i_\tau$,
then each $\d_{R^i}\tau_{k_i}$ is a conjunct of $\tau_j$, and so, by IH and monotonicity,
\[
\tau_j[\vala]\le \d_{R^1}\tau_{k_1}[\vala]\land\dots\land \d_{R^\ell}\tau_{k_\ell}[\vala]\le
 \d_{R^1}b_{k_1}\land\dots\land  \d_{R^\ell}b_{k_\ell}=b_j,
\]
as required in \eqref{righteq}. 
Next, let $\sigma_i=\trm_{v_i}^{\M_\sigma}$, $i=0,\dots,n$.
%
%
We prove that
\begin{equation}\label{lefteq}
a_{i}\le\sigma_i[\vala],\qquad\mbox{for every $i=n,\dots,0$}.
\end{equation}
Indeed, if $v_i$ is a leaf in $\T_\sigma$, then either $\sigma_i=\top$ or $\sigma_i[\vala]=\vala(p)=a_i$
for some $p$.
Now suppose inductively that \eqref{lefteq} holds for every $\ell$ with $i<\ell\leq n$.
We have $a_{i}\le\dR a_{\ell}$ for every $v_\ell$ with $(v_i,v_\ell)\in R_\sigma$.
So, by IH and monotonicity, we have 
\[
a_{i}\le a_{i}\land \bigwedge_{(v_i,v_\ell)\in R_\sigma} \dR \sigma_\ell[\vala]\le \sigma_i[\vala],
\]
as required in \eqref{lefteq}. In particular, $a_0\le \sigma_0[\vala]=\sigma[\vala]$.
As $\Aa\models (\sigma\imp\tau)[\vala]$, 
\begin{equation}\label{rooteq}
 a_0\le \tau[\vala].
 \end{equation}
 %
 %
Finally, to prove (b) for $j=0$, suppose that $R_\tau(u_0,u_j)$ for some $j$.
Then $\dR\tau_j$ is a conjunct of $\tau$, therefore $ \tau[\vala]\le \dR\tau_j[\vala]$, and so
$b_0=a_0\le \dR\tau_j[\vala]\le \dR b_j$ by \eqref{rooteq}, \eqref{righteq} and monotonicity, thus establishing $(b_0,b_j)\in R^\F$.
\end{proof}

Theorem~\ref{t:exist} has
the following consequence about the \SP-fragments of modal 
\emph{grammar logics}~\cite{farinasdelCerroPenttonen88}:

\begin{corollary}\label{co:grammar}
Every \sptheory{} axiomatised by \spequations{} of the form 
\[
\d_{\!R_1\,} \dots \d_{\!R_n\,} p \imp  \d_{\!S_1\,} \dots \d_{\!S_m\,}p,\qquad\mbox{ for $n\ge 0$, $m>0$},
\]
is \complex, and so \eqcomp. 
\end{corollary}

In particular, the \sptheories{} $\SPi + \{\eden\}$ with
$\eden=(\d p\imp\d\d p)$ (defining \emph{density\/}) and 
$\SPi + \{\dR\dS p\imp\dS\dR p,\,\dS\dR p\imp\dR\dS p\}$ 
(defining \emph{commutativity}) are \complex\ and \eqcomp. On the other hand, Corollary~\ref{co:grammar} gives examples of \eqcomp\ but undecidable finitely axiomatisable  \sptheories~\cite{Tseitin58,Shehtman82,Chagrov-Shehtman95,Baader03,Beklemishev15b}, which clearly cannot have the finite frame property.

The following theorem will be used in \S\ref{sec:undec}.
\begin{theorem}\label{t:funcomm}
Suppose $R$, $S$ and $Z$ are distinct elements in some signature $\R$, and let 
$\Ec$ consist of the following \spequations:
$\efun$ for $\dR$, $\efun$ for $\dS$,
%
\begin{align}
\label{comm}
& \dR\dS p\imp\dS\dR p\ \mbox{ and }\ \dS\dR p\imp\dR\dS p,\\
\label{smallest}
& \dZ \top\imp p, \\
\label{allnormal}
& \dX\dZ\top\imp \dZ\top,\quad\mbox{for all $X\in\R$}.
\end{align}
Then the \sptheory{} $\SPi + \Ec$ is  \complex.
\end{theorem}

\begin{proof}
Let $\A = (A, \land, \top, \dX)_{X\in\R}$ be a \SLOa\ such that $\A \models \Ec$.
Then by \eqref{smallest}, $\dZ\top$ is the $\le$-smallest element in $A$.
We call a filter $U$ of $\A$ \emph{proper} if $\dZ\top\notin U$.
As shown in the proof of Theorem~\ref{t:fun} $(i)$, for $X\in\{R,S\}$ and any filter $U$ of $\A$,  $\d_X^{-1}[U]$ is 
either empty or a filter. By \eqref{allnormal},  
$\d_X^{-1}[U]$  is either empty or a proper filter whenever $U$ is proper.
Let $\PFA$ be the set of all proper filters of $\Aa$.
We set, for $U,V\in\PFA$,
\begin{align*}
(U,V)\in X^\G\quad& \Longleftrightarrow\quad V=\d^{-1}_X[U],\quad\mbox{for $X\in\{R,S\}$},\\
(U,V)\in X^\G\quad& \Longleftrightarrow\quad \dX[V]\subseteq U,\quad\mbox{for $X\in\R\setminus\{R,S,Z\}$}.
\end{align*}
Then $R^\G$ and $S^\G$ are functional. Moreover, every $X\in\R\setminus\{Z\}$ satisfies \eqref{emb1} and \eqref{emb2} as well (with respect to $\PFA$), and so the map
$f(a)= \{U \in \PFA \mid a \in U\}$, for $a\in A$,
embeds $\A$ into $\Gslo$ for $\G=(\PFA,R^\G,S^\G,\emptyset,X^\G)_{X\in\R\setminus\{R,S,Z\}}$.
Clearly, $\G$ validates \eqref{smallest} and \eqref{allnormal}. 
It remains to show that $\G$ validates \eqref{comm}, that is, $R^\G$ and $S^\G$ commute. Suppose that, say, $(U,V)\in R^\G$ and $(V,W)\in S^\G$, and let $Y=\d^{-1}_S[U]$. We claim that $Y\ne\emptyset$,
$(U,Y)\in S^\G$ and $(Y,W)\in R^\G$. Indeed, take some $a\in W$. Then $\dS a\in V$, and so $\dR\dS a\in U$.
By \eqref{comm}, $\dS\dR a\in U$, and so $\dR a\in Y$, whence $Y\ne\emptyset$ and 
$a\in \d^{-1}_R [Y]$. Therefore, $(U,Y)\in S^\G$ and $W\subseteq \d^{-1}_R [Y]$. 
The inclusion $\d^{-1}_R [Y]\subseteq W$ is similar, proving $(Y,W)\in R^\G$.
The other direction of \eqref{comm} is shown analogously.
\end{proof}


\section{Completeness of \sptheories{} with disjunctive correspondents}\label{sec:disj}
%
Finally, we consider \spequations{} whose  correspondents contain disjunction, starting 
with a simple example.

\begin{example}\label{e:pic5}\em 
The \sptheory{} $\SPi + \{\e\}$ with $\e = (p \land \dR p \imp \dS p)$ corresponding to the non-Horn, disjunctive first-order condition 
\begin{equation}\label{e1}
\Psi_{\e} = \ \forall x,y\, \bigl(R(x, y) \to S(x, x) \lor S(x, y)\bigr)
\end{equation}
is not \eqcomp. It is easy to see that $\{\e\}\models_{\CA}p \land \dR \dS p \imp \dS \dS p$. However, the \SLOa{} in Fig.~\ref{f:pic5} 
validates $\e$, but refutes $p \land \dR \dS p \imp \dS \dS p$ when $p$ is $\{1,4\}$.
\end{example}
\begin{figure}[ht]
\centering
\begin{tikzpicture}[scale=.6]
\tikzset{every loop/.style={thick,min distance=15mm,in=45,out=135,looseness=8}}
\tikzset{place/.style={circle,thin,draw=white,fill=white,scale=1.2}}

\draw [fill=gray] (3.5,1) circle [radius=.25];
\node at (3.5,1)  {\textcolor{white}{\bf\footnotesize 1}}; 
\node[place] (foo) at (6,3.5) {}; 
\draw [fill=gray] (6,3.5) circle [radius=.25];
\node at (6,3.5)  {\textcolor{white}{\bf\footnotesize 2}}; 
\draw [fill=gray] (3.5,3.5) circle [radius=.25];
\node at (3.5,3.5)  {\textcolor{white}{\bf\footnotesize 3}}; 
\draw [fill=gray] (1,3.5) circle [radius=.25];
\node at (1,3.5)  {\textcolor{white}{\bf\footnotesize 4}}; 
\draw [thick,->] (3.75,1.25) to [out=45, in=-135] (5.75,3.25);
\draw [thick,->] (3.5,1.35) to [out=90, in=-90] (3.5,3.15);
\draw [thick,->] (3.15,3.5) to [out=180, in=0] (1.35,3.5);
\path[->] (foo) edge [loop above] node {${}_S$} ();
\draw[thick,dotted,rounded corners=4] (2.7,4.3) -- (4.5,4.3) -- (4.5,2.8) -- (2.7,2.8) -- cycle;
\draw[thick,dotted,rounded corners=4] (2.7,1.6) -- (4.5,1.6) -- (4.5,.5) -- (2.7,.5) -- cycle;
\draw[thick,dotted,rounded corners=6] (.2,4.3) -- (1.7,4.3) -- (1.7,1.8) -- (4.8,1.8) -- (4.8,.2) -- (2.5,.2) -- (.2,1.9) -- cycle;
\draw[thick,dotted,rounded corners=6] (2.4,2) -- (3.9,2) -- (5.2,3.3) -- (5.2,5.3) -- (6.8,5.3) -- (6.8,1.8) -- (5,-.1) -- (2.4,-.1) -- cycle;
\draw[thick,dotted,rounded corners=10] (2,5.6) -- (7.2,5.6) -- (7.2,-.4) -- (2,-.4) -- cycle; 
\draw[thick,dotted,rounded corners=10] (-.2,5.9) -- (7.5,5.9) -- (7.5,-.7) -- (-.2,-.7) -- cycle; 

\node[right] at (2,3.7) {${}^S$};
\node[right] at (2.8,2.3) {${}^R$};
\node[right] at (4.8,2.1) {${}^S$};

\end{tikzpicture}
\caption{The \SLOa{} of Example~\ref{e:pic5}.}\label{f:pic5}
\end{figure}


\subsection{\spequationscap{} defining $n$-functionality.}

Let $P=\{p_0,\dots,p_{n}\}$,  for \mbox{$n\geq 1$}, and let 
%
\begin{equation}\label{efuneq}
\efun^n ~=~ \bigl ( \hspace*{-4mm}\bigwedge_{Q\subseteq P,\, |Q|=n} \hspace*{-4mm}\d \bigwedge Q ~\imp~ 
\d\bigwedge P\bigr).
\end{equation}
In particular, $\efun^1=\efun$. 
%
%
It is easy to see that $\efun^n$ corresponds to $n$-\emph{functionality\/}:
\[
\forall x,y_0,\dots,y_n\ \bigl(\bigwedge_{i\leq n}R(x,y_i)\to\bigvee_{i\ne j}(y_i=y_j)\bigr).
\]

\begin{theorem}\label{t:altncomplex}
None of $\SPi + \{\efun^n\}$, $\SPi + \{\erefl,\efun^n\}$, $\SPi + \{\etrans,\efun^n\}$, 
$\SPi+(\Axsfour \cup \{\efun^n\})$, and 
$\SPi+(\Axsfive \cup \{\efun^n\})$ is \complex, for $n\geq 2$.
\end{theorem}

\begin{proof}
Let $\Aa_n$ be the \SLOa{} in Fig.~\ref{f:altncomplex}.
It is easy to see that $\Aa_n\models\Axsfive \cup\{ \efun^n\}$ if $n\geq 2$.  
Now suppose  there is an \SPa-embedding 
$\eta:\Aa_n\to\Fslo$ for some frame $\F=(W,R^\F)$.
Then there is some $x\in W \setminus\eta(g)$.
As $W=\eta(\d a_i)=\d^+\eta(a_i)$ for all $i\leq n$, 
there exist
$y_0\in\eta(a_0)$, $\dots$, $y_n\in\eta(a_n)$ such that $(x,y_i)\in R^\F$ for all $i\leq n$. 
As $\eta(g)=\eta(\d g)=\d^+\eta(g)$, we have $y_i\notin\eta(g)$, for any $i\leq n$.
It follows that all the $y_i$ are distinct, showing that $\F$ is not $n$-functional.
\begin{figure}[ht]
\begin{center}
\begin{tikzpicture}[scale=.8]
\draw [fill] (1,2) circle [radius=0.1];
\node[left] at (1,2) {$a_0$};
\draw [fill] (2,2) circle [radius=0.1];
\node[left] at (2,2) {$a_1$};
\draw [fill] (5,2) circle [radius=0.1];
\node[right] at (5,2) {$a_n$};
\draw [fill] (1,2) circle [radius=0.1];
\draw [very thick] (3,1) circle [radius=0.15];
\node[below] at (3,.8) {$g$};
\draw [very thick] (3,3) circle [radius=0.15];
\node[above] at (3,3.2) {$\top$};

\draw [fill] (3.2,2) circle [radius=0.015];
\draw [fill] (3.4,2) circle [radius=0.015];
\draw [fill] (3.6,2) circle [radius=0.015];

\draw [very thin] (3.2,1.1) -- (5,2);
\draw [very thin] (2.8,1.1) -- (1,2);
\draw [very thin] (2.8,1.2) -- (2,2);
\draw [very thin] (1,2) -- (2.8,2.9);
\draw [very thin] (2,2) -- (2.8,2.8);
\draw [very thin] (5,2) -- (3.2,2.9);

\draw [very thick,->] (1,2) to [out=90, in=180] (2.8,3);
\draw [very thick,->] (2,2) to [out=30, in=-120] (3,2.8);
\draw [very thick,->] (5,2) to [out=90, in=0] (3.2,3);
\end{tikzpicture}
\end{center}
\caption{The \SLOa{} $\Aa_n$ in the proof of Theorem~\ref{t:altncomplex}.}\label{f:altncomplex}
\end{figure}
\end{proof}


\begin{theorem}\label{t:altn}
$\SPi + \{\efun^n\}$ is \eqcomp, for any $n\geq 1$.
\end{theorem}

\begin{proof}
For $n=1$, this is Theorem~\ref{t:fun}~$(i)$. For $n\geq 2$, we prove the theorem by the \method\ method
from \S\ref{sec:interp}.
We first define $\{\efun^n\}$-\emph{\nterm s} by induction: 
$(i)$ all propositional variables and $\top$ are $\{\efun^n\}$-\nterm s; 
$(ii)$ if $\tau_{1},\ldots,\tau_{n}$ are $\{\efun^n\}$-\nterm s, then so is $\d(\tau_{1} \land \cdots \land\tau_{n})$. 

\begin{lclaim}\label{c:altnormalform}
For any \spterm{} $\varrho$, there is a set $N_\varrho$ of $\{\efun^n\}$-\nterm s such that
\[
\{\efun^n\} \vdash_{\SLO} \varrho \approx \bigwedge N_\varrho.
\]
\end{lclaim}

\begin{proof}
The proof is by induction on the modal depth $d$ of $\varrho$. The basis $d=0$ is trivial. 
Suppose inductively that $\varrho$ is an \spterm{} of depth $d>0$. Then 
$\varrho = \bigwedge P_\varrho\land\d \varrho_1 \land \dots \land \d \varrho_k$, where 
$P_\varrho$ is a set consisting of propositional variables and $\top$, and each $\varrho_i$ is an \spterm{} of depth $\leq d-1$. By IH, $\{\efun^n\} \vdash_{\SLO} \varrho_i \approx \bigwedge A_i$, for some set $A_i$ of $\{\efun^n\}$-\nterm s and  $i=1,\dots,k$. Then 
$$
\{\efun^n\} \vdash_{\SLO} \varrho \approx (\bigwedge P_\varrho\land\bigwedge_{i=1}^k \d\bigwedge A_i).
$$
If $|A_i|\leq n$ for all $i$, then we are done. So fix some $i$ and suppose  $|A_i|=k>n$.
Then we always have 
$ \vdash_{\SLO} \d\bigwedge A_i \imp
\bigwedge_{Q\subseteq A_i,\,|Q|=n} \d \bigwedge Q$.
We show that
\begin{equation}\label{allfun}
\{\efun^n\} \vdash_{\SLO}
 \bigwedge_{\substack{Q\subseteq A_i,\\|Q|=n}} \d \bigwedge Q\imp
 \d\bigwedge A_i.
\end{equation}
In order to prove this, first we claim that
%
 $\{\efun^m\}\vdash_{\SLO}\efun^{m+1}$, for every $m$.
 %
Indeed,
\begin{multline*}
\{ \efun^m\}\vdash_{\SLO}
\bigwedge_{\substack{Q\subseteq\{p_0,\dots,p_{m+1}\}\\|Q|=m+1}}
\hspace*{-.5cm}\d\bigwedge Q\imp \hspace*{-.2cm}
\bigwedge_{\substack{Q\subseteq\{p_1,\dots,p_{m+1}\}\\|Q|=m}}
\hspace*{-.5cm}\d\bigl(p_0\land\bigwedge Q \bigr)\imp\\[5pt]
\bigwedge_{\substack{Q\subseteq\{p_1,\dots,p_{m+1}\}\\|Q|=m}}
\hspace*{-.3cm}\d \bigwedge_{q\in Q}(p_0\land q) 
\imp~ \d\hspace*{-.7cm}\bigwedge_{q\in\{p_1,\dots,p_{m+1}\}}(p_0\land q)\approx
\d(p_0\land\dots\land p_{m+1}).
\end{multline*}
Therefore, we have 
%
$\{\efun^n\}\vdash_{\SLO}\efun^{m}$, for every $m>n$.
%
Thus,
\begin{multline*}
\{\efun^n\}\vdash_{\SLO}
\bigwedge_{\substack{Q\subseteq \{p_0,\dots,p_{k-1}\}\\ |Q|=n}} \d \bigwedge Q  ~\imp~ 
\bigwedge_{\substack{Q\subseteq \{p_0,\dots,p_{k-1}\}\\ |Q|=n+1}} \d \bigwedge Q ~\imp~\dots\\
\dots ~\imp~ \d(p_0\land\dots\land p_{k-1}),
\end{multline*}
and so a substitution of the $k$ terms in $A_i$ for $p_0,\dots,p_{k-1}$ in $\efun^k$ gives \eqref{allfun}.
\end{proof}

\begin{lclaim}\label{c:altkrtoslo}
For any \spterm{} $\sigma$ and $\{\efun^n\}$-\nterm\ $\tau$,   
if $\{\efun^n\}\models_{\CA} \sigma \imp \tau$ then $\models_{\CA} \sigma \imp \tau$.
\end{lclaim}
\begin{proof}
The proof is by induction on the modal depth $d$ of $\tau$. The basis is again easy. 
Now assume inductively that the claim holds for $d$ and the depth of $\tau$ is $d+1$. Let 
$\sigma=\bigwedge P_\sigma \land \d\sigma_{1}\land \ldots \land\d\sigma_{k}
$, 
where $P_\sigma$ is some set of propositional variables and $\top$,
and each $\sigma_i$ is an \spterm.
Suppose $\tau= \d(\tau_{1}\wedge \ldots \wedge \tau_{n})$, where 
each $\tau_{i}$ is either a variable, or $\top$, or of the form 
$\d (\tau_{1}^{i}\wedge \cdots \wedge \tau_{n}^{i})$.
Let $\not\models_{\CA} \sigma \imp \tau$. Then, for every $j$ ($1 \le j \le k$), there is $i$ ($1 \le i\le n$) such that $\not\models_{\CA} \sigma_j \imp \tau_i$, and so
$\bigcup_{i=1}^{n} K_i=\{1,\dots,k\}$, for $
K_i=\{1\le j\le k \mid \ \not\models_{\SLO}\sigma_j\imp\tau_i\}$.
It is not hard to see that, for any $i$ with $K_i\ne\emptyset$, we have $\not\models_{\CA}(\bigwedge_{j\in K_{i}} \sigma_j) \imp \tau_{i}$.
By IH, for any such $i$, there is a Kripke model $\M_{i}$ based on an $n$-functional frame with root $r_{i}$ where $\bigwedge_{j\in K_{i}} \sigma_j$ holds, but $\tau_i$ does not. Now take a fresh node $r$, make $\bigwedge P_\sigma$ true in $r$, and connect $r$ to $r_i$ of each $\M_{i}$. The constructed Kripke model is based on an $n$-functional frame and refutes $\sigma\imp \tau$ at $r$, showing that 
$\{\efun^n\}\not\models_{\CA}\sigma\imp\tau$ as required.
\end{proof}

That $\SPi + \{\efun^n\}$ is \eqcomp\ follows now from Claims~\ref{c:altnormalform}, \ref{c:altkrtoslo}, \eqcomp ness of $\SPi$ (Theorem~\ref{t:valid}) and \eqref{birkhoff}.
\end{proof}


Now, we set 
%
\[
\Axfunn ~=~ \Axsfive \cup \{\efun^n\}, \quad \text{ for }1\le n<\omega.
\]
and $\Efunn=\SPi+\Axfunn$.
The correspondent $\Phi_n$ of $\Axfunn$ 
says that $R$ is an equivalence relation whose classes (clusters) are of size $\leq n$.

\begin{theorem}\label{t:s5alt} \textup{(}\textup{\cite{jackson2004}}\textup{)}%
\footnote{This result also follows from~\cite{jackson2004}, which showed (for the similarity type without $\top$) that $\Efunn\models_{\BAO}$ is conservative over 
$\Efunn\models_{\SLO}$.}
$\Efunn$ is \eqcomp,
%
%
 for every $n\geq 2$.
\end{theorem}

\begin{proof}
The proof is again by the method of \method\  from \S\ref{sec:interp}.
Now, $\Axfunn$-\emph{\nterm s\/} 
 are defined as propositional variables, $\top$ and \spterms{}  of the form
$\d(q_1\land\dots\land q_n)$, where the $q_i$ are propositional variables  
or $\top$.

\begin{lclaim}\label{c:s5altnormalform}
For any \spterm{} $\varrho$, there is a set $N_\varrho$ of $\Axfunn$-\nterm s with 
\[
\Axfunn \vdash_{\SLO} \varrho \approx \bigwedge N_\varrho.
\]
\end{lclaim}

\begin{proof}
As $\Axsfive\vdash_{\SLO}\d p\land\d q\approx\d (p\land\d q)$    (by $\Esfive=\Esfive'$ and $\eeuc$), 
it is easy to see that, for any $\{\efun^n\}$-\nterm\ $\alpha$ (as defined in the proof of Theorem~\ref{t:altn}), 
there is some $\Axfunn$-\nterm\ $\beta$ such that $\Axsfive\vdash_{\SLO}\alpha\approx\beta$.
Now the claim follows from Claim~\ref{c:altnormalform}.
In particular,
\[
N_\varrho=P_{r_\varrho}\cup\{\d\bigwedge Q\mid Q\subseteq P_x,\ |Q|\le n,\ \mbox{$x$ in $\M_\varrho$}\},
\]
where $P_x$ is the set of variables that are true at $x$ in the $\varrho$-tree model $\M_\varrho$.
\end{proof}

\begin{lclaim}\label{c:s5altkrtoslo}
For any \spterm{} $\sigma$ and $\Axfunn$-\nterm\ $\tau$, if
$\Axfunn\models_{\CA} \sigma \imp \tau$ then $\Axsfive\models_{\CA} \sigma \imp \tau$.
\end{lclaim}

\begin{proof}
Suppose $\Axsfive\not\models_{\CA} \sigma \imp \tau$. 
Let $R_\sigma^\forall=W_\sigma\times W_\sigma$ for the domain $W_\sigma$ of the $\sigma$-tree
model $\M_\sigma$. Consider the Kripke model $\M_\sigma^\forall=(\T_\sigma^\forall,\V_\sigma)$
over $\T_\sigma^\forall=(W_\sigma,R_\sigma^\forall)$.
As $\M_\sigma^\forall$ is the equivalence-closure of $\M_\sigma$, we have 
$\M_\sigma^\forall,r_\sigma\kmodels\sigma$  and
$\M_\sigma^\forall,r_\sigma\not\kmodels\tau$ 
by Proposition~\ref{p:Hclosure}, and so $\tau\ne\top$.
If $\tau$ is a propositional variable $p$, then take the following model
$\M$ based on the universal frame over $\{x,y\}$:
for each variable $q$, let $\M,x\kmodels q$ iff $r_\sigma\in\V_\sigma(q)$ and $\M,y\kmodels q$ iff $\V_\sigma(q)\setminus\{r_\sigma\}\ne\emptyset$. (That is, $\M$ is obtained from $\M_\sigma^\forall$ by `sticking together' all of its points different from $r_\sigma$.)
Then we clearly have  $\M,x\not\models(\sigma\imp p)$.
Finally, let  $\tau$ be of the form $\d(q_1\land\dots\land q_n)$.
If $W_\sigma$ contains $\leq n$ points, then $\Axfunn\not\models_{\CA} \sigma \imp \tau$. So suppose 
$W_\sigma=\{w_1,\dots,w_m\}$ for some $m> n$. 
We show that there is a Kripke model $\M$ based on a universal frame with $<m$ points and such that  
$\M\not\models(\sigma\imp\tau)$. Indeed, as 
$\M_\sigma^\forall,r_\sigma\not\kmodels\d(q_1\land\dots\land q_n)$, 
for every $1\leq i\leq m$ there is $Q_i\subseteq\{q_1,\dots,q_n\}$ such that $|Q_i|=n-1$ and
$\{q_k\mid 1\leq k\leq n,\ \M_\sigma^\forall,w_i\kmodels q_k\}\subseteq Q_i$.
So by the pigeonhole principle, there are $i\ne j$ with $Q_i=Q_j$.
Now let $\M$ result from $\M_\sigma^\forall$ by `sticking together' $w_i$ and $w_j$.
Then we have  $\M,r_\sigma\not\models\sigma\imp\d(q_1\land\dots\land q_n)$, and so 
$\M\not\models(\sigma\imp\tau)$, as required.
\end{proof}

That $\Efunn$ is \eqcomp\ follows now from Claims~\ref{c:s5altnormalform}, \ref{c:s5altkrtoslo}, \eqcomp ness of $\Esfive$ (Theorem~\ref{t:comp}~$(ii)$) and \eqref{birkhoff}.
\end{proof}

As a consequence of Claims~\ref{c:s5altnormalform} and \ref{c:s5altkrtoslo} we also obtain:

\begin{theorem}\label{t:s5altP}  
$\Efunn$ is decidable in {\sc PTime}, for every $n\geq 2$.
\end{theorem}

\begin{proof}
Follows from the tractability of $\Esfive$ 
(Theorem~\ref{thm:horncomplete})
and the fact that $|N_\varrho|$ in Claim~\ref{c:s5altnormalform}  is clearly polynomial in the size of $\M_\varrho$.
\end{proof}

Jackson \cite{jackson2004} also proves the following about extensions of $\Esfive$:

\begin{theorem}\label{t:jackson} \textup{(}\textup{\cite{jackson2004}}\textup{)}
Let $\SPL$ be any non-trivial \sptheory{} extending $\Esfive$. 
Then exactly one of the following cases holds: 
\begin{itemize}
\item[--]
$\SPL = \Esfive$,
\item[--]
$\SPL = \SPi+(\Axsfive\cup \{\d p\imp p\})$,
\item[--]
$\SPL = \SPi+(\Axsfive \cup \{\d p\imp \d q\})$,
\item[--]
$\SPL = \Efunn$, for some $n$ $(1\leq n<\omega)$.
\end{itemize}
\end{theorem}

Then Proposition~\ref{p:trivial}, Theorems~\ref{t:comp}~$(ii)$, \ref{t:simple1}, \ref{t:fun}~$(iv)$, \ref{t:simple2} and \ref{t:jackson}  give a full
classification of the extensions $\Esfive$ according to their completeness: 
the trivial \sptheory, $\Esfive$,  $\SPi+(\Axsfive \cup \{\d p\imp p\})$ and $\Efunn$, for $1<n<\omega$, are \eqcomp, 
while $\Efunone$  and $\SPi+(\Axsfive \cup \{\d p\imp \d q\})$ are incomplete.


By Theorem~\ref{t:simple1},
$\SPi+(\Axsfive \cup \{\d p\imp p\})$ is \complex. Theorems~\ref{t:simple2}, \ref{t:altncomplex} and \ref{t:jackson} imply that it is the only \complex\ non-trivial proper extension of $\Esfive$:

\begin{corollary}\label{c:s5complex}
Let $\SPL$ be any non-trivial \sptheory{} such that $L\supseteq\Esfive$,  $\SPL\neq \Esfive$ and 
$L\neq \SPi+(\Axsfive \cup \{\d p\imp p\})$.
Then $L$ is not \complex. 
\end{corollary}



Finally, we show that $\efun^n$ behaves differently when added to $\Esfour$.
A transitive frame $\Ff$ is said to be \emph{of depth\/} $n$, for $n\geq 1$, if $\Ff$ contains a chain of $n$ points from distinct
clusters but no longer chain of this sort. It is easy to see that, over  $\Esfour$, we can define the property
`$\F$ is of depth $\le n$' by the \spequation{} 
\[
\edep^n=\bigl( p\land \underbrace{\d\bigl(q\land \d(p\land\dots)\dots\bigr)}_{\text{$\d$ is used $n$ times}}\ \imp\ 
 \underbrace{\d\bigl(q\land \d(p\land\dots)\dots\bigr)}_{\text{$\d$ is used $n+1$ times}}\bigr).
\]
Then $\edep^1=\esym$ and $\edep^2$ has the following `disjunctive' correspondent:
\[
\forall x,y, z\,\bigl(R(x,y)\land R(y,z)\to R(y,x)\lor R(z,y)\bigr).
\]
Also, it is not hard to see that $\{\edep^n\}\models_{\SLO}\edep^{n+1}$, for all $n\ge 1$
(simply substitute $p\land\d q$ for $p$, and $q\land\d p$ for $q$ in $\edep^n$).

\begin{question}\label{q:depth}
Are $\SPi + \{\edep^n\}$, $\SPi + \{\etrans,\edep^n\}$ and  
$\SPi+(\Axsfour \cup \{\edep^n\})$ complete?
\end{question}

As an $n$-functional reflexive and transitive frame can have at most $n$ points, its depth must be $\leq n$.
Therefore, for any $n\geq 1$, we have 
\begin{equation}\label{alttodepth}
\Axsfour \cup \{\efun^n\}\models_{\CA}\edep^n.
\end{equation}

\begin{theorem}\label{t:altnotcomp}
$\Axsfour \cup \{\efun^n\}\not\models_{\SLO}\edep^n$, for any $n\geq 2$. 
\end{theorem}

\begin{proof}
Fix some $n\geq 2$ and take the \SLOa{} $\Aa_n$ in Fig.~\ref{f:bigslo}.
\begin{figure}[ht]
\begin{center}
\begin{tikzpicture}[scale=.6]
\draw [fill] (6.5,6) circle [radius=0.03];
\draw [fill] (6.5,6.2) circle [radius=0.03];
\draw [fill] (6.5,6.4) circle [radius=0.03];
\draw [fill] (3.85,3.65) circle [radius=0.03];
\draw [fill] (3.7,3.8) circle [radius=0.03];
\draw [fill] (3.55,3.95) circle [radius=0.03];
\draw [fill] (9.15,3.65) circle [radius=0.03];
\draw [fill] (9.3,3.8) circle [radius=0.03];
\draw [fill] (9.45,3.95) circle [radius=0.03];

\draw [fill] (4.5,3) circle [radius=0.1];
\node[left] at (4.5,3) {$d_1$};
\draw [fill] (3,4.5) circle [radius=0.1];
\node[left] at (3,4.5) {$d_{n-2}$};
\draw [fill] (2,5.5) circle [radius=0.1];
\node[left] at (2,5.5) {$d_{n-1}$};
\draw [fill] (1,6.5) circle [radius=0.1];
\node[left] at (1,6.5) {$d_n$};
\draw [fill] (8.5,3) circle [radius=0.1];
\node[right] at (8.5,2.8) {$e_1$};
\draw [fill] (10,4.5) circle [radius=0.1];
\node[right] at (10,4.3) {$e_{n-2}$};
\draw [fill] (11,5.5) circle [radius=0.1];
\node[right] at (11,5.3) {$e_{n-1}$};
\draw [fill] (12,6.5) circle [radius=0.1];
\node[right] at (12,6.3) {$e_{n}$};

\draw [very thick] (6.5,1) circle [radius=0.15];
\node[below right] at (6.5,1) {$g$};
\draw [very thick] (6.5,3) circle [radius=0.15];
\node[right] at (6.6,3) {$a_0$};
\draw [very thick] (6.5,5) circle [radius=0.15];
\node[right] at (6.6,5) {$a_1$};
\draw [very thick] (6.5,8) circle [radius=0.15];
\node[right] at (6.6,8) {$a_{n-2}$};
\draw [very thick] (6.5,10) circle [radius=0.15];
\node[right] at (6.6,10) {$a_{n-1}$};
\draw [very thick] (6.5,12) circle [radius=0.15];
\node[above right] at (6.5,12) {$\top$};
\draw [very thick] (5.5,2) circle [radius=0.15];
\node[left] at (5.4,2) {$d_0$};
\draw [very thick] (5.5,4) circle [radius=0.15];
\node[above left] at (5.5,4.1) {$b_0$};
\draw [very thick] (5.5,7) circle [radius=0.15];
\node[above left] at (5.7,7.1) {$b_{n-3}$};
\draw [very thick] (5.5,9) circle [radius=0.15];
\node[above left] at (5.7,9.1) {$b_{n-2}$};
\draw [very thick] (5.5,11) circle [radius=0.15];
\node[above left] at (5.7,11.1) {$b_{n-1}$};

\draw [very thick] (7.5,2) circle [radius=0.15];
\node[right] at (7.6,1.8) {$e_0$};
\draw [very thick] (7.5,4) circle [radius=0.15];
\node[above right] at (7.4,4.1) {$c_0$};
\draw [very thick] (7.5,7) circle [radius=0.15];
\node[above right] at (7.4,7) {$c_{n-3}$};
\draw [very thick] (7.5,9) circle [radius=0.15];
\node[above right] at (7.4,9) {$c_{n-2}$};
\draw [very thick] (7.5,11) circle [radius=0.15];
\node[above right] at (7.4,11) {$c_{n-1}$};

\draw [very thin] (6.7,1.2) -- (7.3,1.8);
\draw [very thin] (6.3,1.2) -- (5.7,1.8);
\draw [very thin] (7.7,2.2) -- (8.8,3.3);
\draw [very thin] (5.3,2.2) -- (4.2,3.3);
\draw [very thin] (9.7,4.2) -- (12,6.5);
\draw [very thin] (3.3,4.2) -- (1,6.5);
\draw [very thin] (1,6.5) -- (5.3,10.8);
\draw [very thin] (2,5.5) -- (5.3,8.8);
\draw [very thin] (3,4.5) -- (5.3,6.8);
\draw [very thin] (4.5,3) -- (5.3,3.8);
\draw [very thin] (5.7,2.2) -- (6.3,2.8);
\draw [very thin] (7.3,2.2) -- (6.7,2.8);
\draw [very thin] (8.5,3) -- (7.7,3.8);
\draw [very thin] (10,4.5) -- (7.7,6.8);
\draw [very thin] (11,5.5) -- (7.7,8.8);
\draw [very thin] (12,6.5) -- (7.7,10.8);
\draw [very thin] (6.3,3.2) -- (5.7,3.8);
\draw [very thin] (6.7,3.2) -- (7.3,3.8);
\draw [very thin] (5.7,4.2) -- (6.3,4.8);
\draw [very thin] (7.3,4.2) -- (6.7,4.8);
\draw [very thin] (5.7,7.2) -- (6.3,7.8);
\draw [very thin] (7.3,7.2) -- (6.7,7.8);
\draw [very thin] (5.7,9.2) -- (6.3,9.8);
\draw [very thin] (7.3,9.2) -- (6.7,9.8);
\draw [very thin] (5.7,11.2) -- (6.3,11.8);
\draw [very thin] (6.3,8.2) -- (5.7,8.8);
\draw [very thin] (6.7,8.2) -- (7.3,8.8);
\draw [very thin] (6.3,10.2) -- (5.7,10.8);
\draw [very thin] (6.7,10.2) -- (7.3,10.8);
\draw [very thin] (7.3,11.2) -- (6.7,11.8);

\draw [very thick,->] (4.5,3) to [out=90, in=180] (5.3,4);
\draw [very thick,->] (3,4.5) to [out=70, in=200] (5.3,7);
\draw [very thick,->] (2,5.5) to [out=70, in=200] (5.3,9);
\draw [very thick,->] (1,6.5) to [out=70, in=200] (5.3,11);

\draw [very thick,->] (8.5,3) to [out=90, in=0] (7.7,4);
\draw [very thick,->] (10,4.5) to [out=110, in=-20] (7.7,7);
\draw [very thick,->] (11,5.5) to [out=110, in=-20] (7.7,9);
\draw [very thick,->] (12,6.5) to [out=110, in=-20] (7.7,11);
\end{tikzpicture}
\end{center}
\caption{The \SLOa{} $\Aa_n$ in the proof of Theorem~\ref{t:altnotcomp}.}\label{f:bigslo}
\end{figure}
%
%
%
It is easy to check that $\Aa_n\models\Axsfour$. We claim that $\Aa_n\models \efun^n$. In fact, if $n\geq 3$ then $\Aa_n\models \efun^3$. We prove only the latter. Take any valuation $\vala$ in $\Aa_n$. If 
$\vala(p_0)$, $\dots$, $\vala(p_3)$ are not pairwise $\le$-incomparable, then 
$\Aa_n\models \efun^3[\vala]$ clearly holds. And if they are, then 
$d_i,e_j\in\{\vala(p_0),\dots,\vala(p_3)\}$ must hold for some $i,j\le n$. As $d_i\land e_j=g$, the left-hand side of
$\efun^3[\vala]$ evaluates to $g$.

On the other hand, we claim that $\Aa_n\not\models \edep^n[\vala]$ for the valuation $\vala(p)=d_n$ and $\vala(q)=e_n$. Indeed, for $1\le k\le n$, define \spterms{} $\tau_k$ and $\sigma_k$ by taking
$\tau_1=\d q$, $\sigma_1=\d p$, $\tau_k=\d (q\land\sigma_{k-1})$ and $\sigma_k=\d (p\land\tau_{k-1})$.
Then $\edep^n$ is $p\land\tau_n\imp\tau_{n+1}$. It is not hard to prove by parallel induction that
$\tau_k[\vala]= c_{n-k}$ and $\sigma_k[\vala]=b_{n-k}$ for all $1\le k\le n$. Therefore, 
the left-hand side of $\edep^n$ evaluates to $d_n\land c_0=d_0$, while the right-hand side to 
$\d(e_n\land b_0)=\d e_0=e_0$.
\end{proof}

As a consequence of \eqref{alttodepth}, Theorems~\ref{t:fun} and \ref{t:altnotcomp} we obtain:

\begin{corollary}
$\SPi+(\Axsfour \cup \{\efun^n\})$ is not \eqcomp, for any $n\geq 1$. 
\end{corollary}


\subsection{\spequationscap{} defining width above $\Esfour$.}
Consider the \spequation{}
\[
\ewc=\bigl(  \d(p\land q)\land \d(p\land r)\imp \d (p\land \d q\land \d r)\bigr),
\]
with the disjunctive correspondent
\begin{equation}\label{fowconn}
\forall x,y,z\, \bigl(R(x,y)\land R(x,z)\to 
\bigl(R(y,y)\land R(y,z)\bigr)\lor \bigl(R(z,z)\land R(z,y)\bigr)\bigr).
\end{equation}
Now let
\begin{equation}\label{ewconn}
\Axlin=\{\erefl,\etrans, \ewc\}.
\end{equation}
It is easy to see that $\Axlin$ defines the class of all linear quasiorders (frames for the modal logic 
$\mf{S4.3}$).
%
We set 
\[
\Elin=\SPi +\Axlin.
\]

\begin{theorem}\label{t:notdlocons}
Neither $\SPi + \{\ewc\}$ nor $\Elin$ is \complex.	
\end{theorem}

\begin{proof}
Take the \SLOa{} $\Aa$ in Fig.~\ref{f:notdlocons}.
\begin{figure}[ht]
\begin{center}
\begin{tikzpicture}[scale=.65]
\draw [fill] (1,3) circle [radius=0.1];
\node[left] at (1,3) {$a$};
\draw [fill] (5,3) circle [radius=0.1];
\node[right] at (5,3) {$c$};
\draw [very thick] (3,1) circle [radius=0.15];
\node[below left] at (3,1) {$g$};
\draw [very thick] (3,3) circle [radius=0.15];
\node[left] at (2.8,3) {$b$};
\draw [very thick] (2,4) circle [radius=0.15];
\node[above] at (2,4.1) {$d$};
\draw [very thick] (4,4) circle [radius=0.15];
\node[above] at (4,4.1) {$e$};
\draw [very thick] (3,5) circle [radius=0.15];
\node[above right] at (3,5) {$\top$};

\draw [very thin] (3,1.2) -- (3,2.8);
\draw [very thin] (2.8,1.2) -- (1.1,2.9);
\draw [very thin] (3.2,1.2) -- (4.9,2.9);
\draw [very thin] (1.1,3.1) -- (1.8,3.8);
\draw [very thin] (4.9,3.1) -- (4.2,3.8);
\draw [very thin] (2.8,3.2) -- (2.2,3.8);
\draw [very thin] (3.2,3.2) -- (3.8,3.8);
\draw [very thin] (2.2,4.2) -- (2.8,4.8);
\draw [very thin] (3.8,4.2) -- (3.2,4.8);

\draw [very thick,->] (1,3.1) to [out=90, in=180] (1.8,4);
\draw [very thick,->] (5,3.1) to [out=90, in=0] (4.2,4);
\end{tikzpicture}
\end{center}
\caption{The \SLOa{} $\Aa$ in the proof of Theorem~\ref{t:notdlocons}.}\label{f:notdlocons}
\end{figure}
It is not hard to check that $\Aa\models\Axlin$.
Now suppose we have an \SPa-embedding $\eta:\Aa\to\Fslo$, for some $\F=(W,R^\F)$.
Then there is $u\in\eta(b)\setminus\eta(g)$.
As $b\leq \d a$ and $b\leq \d c$, we have
$\eta(b)\subseteq\eta(\d a)=\d^+\eta(a)$ and $\eta(b)\subseteq\eta(\d c)=\d^+\eta(c)$. 
Then $u\in \d^+\eta(a)\cap \d^+\eta(c)$, and so there are $v\in \eta(a)$ and $w\in \eta(c)$ such 
that $(u,v)\in R^\F$ and $(u,w)\in R^\F$. 
As $\eta(g)=\eta(\d g)=\d^+\eta(g)$, we have $v\notin\eta(g)=\eta(a\land\d c)=\eta(a)\cap\d^+\eta(c)$,
and so $(v,w)\notin R^\F$. Similarly,
$w\notin\eta(g)=\eta(c\land\d a)=\eta(c)\cap\d^+\eta(a)$,
and so $(w,v)\notin R^\F$.
Therefore,  \eqref{fowconn} does not hold in $\F$.
\end{proof}

Now we use the \method\ method to prove the following:
\begin{theorem}\label{t:lincomp}
$\Elin$ is \eqcomp.
\end{theorem}
\begin{proof}
We define $\Axlin$-\emph{\nterm s} by induction: $(i)$ all finite conjunctions of propositional variables are $\Axlin$-\nterm s; $(ii)$ if $\tau$ is an $\Axlin$-\nterm\ and $P_\tau$ is a set of propositional variables, then $\bigwedge P_\tau\land\d\tau$ is an $\Axlin$-\nterm.


\begin{lclaim}\label{c:linnormalform}
For any \spterm{} $\varrho$, there is a set $N_\varrho$ of $\Axlin$-\nterm s with 
\[
\Axlin \vdash_{\SLO} \bigl(\varrho \approx \bigwedge N_\varrho\bigr).
\]
\end{lclaim}

\begin{proof}
Let $N_\varrho$ be the set of \nterm s describing the full linear branches of $\M_\varrho$ (from root to a leaf): if $\varrho=\bigwedge P_\varrho$ then $N_\varrho=\{\varrho\}$, and if $\varrho=\bigwedge P_\varrho\land\bigwedge_{i<k}\d\varrho_i$ then
$N_\varrho=\{\bigwedge P_\varrho\land\d\tau\mid \tau\in\bigcup_{i<k} N_{\varrho_i}\}$.
We clearly have $\vdash_{\SLO} \bigl(\varrho \imp \bigwedge N_\varrho\bigr)$.
Conversely, it is easy to see first that, for any $n$,
\begin{equation}\label{lincons}
\Axlin\vdash_{\SLO} \d(p\land \d q_1)\land\dots\land \d(p\land \d q_n)\imp 
\d (p\land \d q_1\land\dots\land \d q_n).
\end{equation}
Next, we prove by induction on $\varrho$ that
\begin{equation}\label{termind}
\Axlin \vdash_{\SLO} \bigwedge_{\tau\in N_\varrho} \hspace*{-2mm} \d\tau ~\imp~ \d\varrho.
\end{equation}
This is obvious if the depth of $\varrho$ is $0$. So suppose  
$\varrho=\bigwedge P_\varrho\land\bigwedge_{i<k}\d\varrho_i$. Then
\begin{align*}
\Axlin \vdash_{\SLO} \bigwedge_{\tau\in N_\varrho} \hspace*{-2mm} \d\tau  & \approx
\bigwedge_{i<k}\bigwedge_{\tau\in N_{\varrho_i}} \hspace*{-2mm}\d (\bigwedge P_\varrho\land\d \tau) ~\imp
\quad\mbox{(by \eqref{lincons})}\\
& \imp\ \bigwedge_{i<k} \d (\bigwedge P_\varrho \,\land \bigwedge_{\tau\in N_{\varrho_i}}\hspace*{-2mm}\d\tau) 
\quad\mbox{(by IH)}\\
& \imp\ \bigwedge_{i<k} \d (\bigwedge P_\varrho\land\d\varrho_i) 
\quad\mbox{(by \eqref{lincons})}\\
& \imp\ \d(\bigwedge P_\varrho\land\bigwedge_{i<k}\d\varrho_i) ~\approx~\d\varrho.
\end{align*}
Finally, by \eqref{termind}, for every $i<k$, 
\[
\Axlin \vdash_{\SLO} 
\bigwedge N_\varrho ~\imp~ \bigwedge_{\tau\in N_{\varrho_i}} \hspace*{-2mm}\d\tau ~\imp~ \d\varrho_i.
\]
Since $\Axlin \vdash_{\SLO} \bigwedge N_\varrho\imp\bigwedge P_\varrho$,
we have
$\Axlin \vdash_{\SLO} \bigwedge N_\varrho\imp\varrho$ as required.
\end{proof}


\begin{lclaim}\label{c:linnterms}
For any \spterm{} $\sigma$ and any $\Axlin$-\nterm\ $\tau$, 
$\Axlin\models_{\CA} \sigma \imp \tau$ implies $\Axsfour\models_{\CA} \sigma \imp \tau$.
\end{lclaim}

\begin{proof}
Suppose   
$\Axsfour\not\models_{\CA} \sigma \imp \tau$.
Take the $\sigma$-tree model $\M_\sigma$, and let 
$\M_\sigma^\ast=\bigl((W_\sigma,R_\sigma^\ast),\V_\sigma\bigr)$
for the reflexive and transitive closure $R_\sigma^\ast$ of $R_\sigma$.
By Proposition~\ref{p:Hclosure}, we have 
$\M_\sigma^\ast,r_\sigma\not\kmodels\tau$.
We call $\M=(W_\sigma,R,\V_\sigma)$ a \emph{linearisation} of $\M_\sigma^\ast$ 
if $R$ is a linear order%
\footnote{A \emph{linear order} is an antisymmetric linear quasiorder.}
containing $R_\sigma^\ast$.

It should be clear that 
$\M,r_\sigma\kmodels\sigma$ 
for any  linearisation $\M$ of $\M_\sigma^\ast$.
We show that
there is a linearisation $\M_\sigma^+$ of $\M_\sigma^\ast$ with 
$\M_\sigma^+,r_\sigma\not\kmodels\tau$, 
which means that $\Elin\not\models_{\CA} \sigma \imp \tau$.
We construct $\M_\sigma^+$ step-by-step by rearranging the points in $W_\sigma$.
We build a
binary tree $(\mathbb L_\sigma, \prec)$ of models $\M=(W_\sigma,R,\V_\sigma)$ by induction so that 
each $(W_\sigma,R)$ is a reflexive and transitive tree containing $R_\sigma^\ast$ and,  for each 
$\M$ in $\mathbb L_\sigma$, there is some $\M'$ with $\M\prec\M'$ and 
$\M',r_\sigma\not\kmodels\tau$.
Each leaf in $(\mathbb L_\sigma, \prec)$ will be a linearisation of $\M_\sigma^\ast$.
First, let $\M_\sigma^\ast$ be the root of $(\mathbb L_\sigma, \prec)$.
Suppose now inductively that $\M=(W_\sigma,R,\V_\sigma)$ in $\mathbb L_\sigma$ has been defined,
$\M,r_\sigma\not\kmodels\tau$, 
and $R$ is not a linear order.
We call a triple $(u,v_\Le,v_\Re)$ of distinct points in $W_\sigma$ an $R$-\emph{defect\/} 
if $(u,v_\Le)\in R$, $(u,v_\Re)\in R$,
$v_\Le\ne v_\Re$, but neither $(v_\Le,v_\Re)\in R$ nor $(v_\Re,v_\Le)\in R$ hold.
Take any $R$-defect $(u,v_\Le,v_\Re)$ with minimal $R$-distance between $r_\sigma$ and $u$.
 We define two relations
$R_{\Le} = \bigl(R \setminus \{(u,v_\Re)\}\bigr) \cup \{(v_\Le,v_\Re)\}$ and 
$R_{\Re} = \bigl(R \setminus \{(u,v_\Le)\}\bigr) \cup \{(v_\Re,v_\Le)\}$ (see Fig.~\ref{f:lr}), and 
add $\M \prec \M_{\Le}$ and $\M \prec \M_{\Re}$ to $(\mathbb L_\sigma, \prec)$, where
$\M_i = (W_\sigma,R_i,\V_\sigma)$ for $i=\Le,\Re$. 
\begin{figure}[ht]
\centering
\hspace*{1.5cm}
\begin{tikzpicture}[scale=.5]
\draw [fill] (4,1) circle [radius=.1];
\draw [thick,->] (4,1.2) to [out=90, in=-90] (4,1.8);
\draw [fill] (4,2) circle [radius=.1];
\draw [thick,->] (4,2.2) to [out=90, in=-90] (4,2.8);
\draw [fill] (4,3) circle [radius=.1];
\draw [thick,->] (4,3.2) to [out=90, in=-90] (4,3.8);
\draw [fill=gray] (4,4) circle [radius=.15];
\node[below right] at (4,4) {$\, u$};
\draw [fill=gray] (1,6) circle [radius=.15];
\node[below left] at (1,6) {$v_\Le$};
\draw[thin,fill=lightgray] (.9,6.3) -- (0,9) -- (1.8,9) -- (1.1,6.3);
\node at (1,8.4) {$T_\Le$};
\draw [fill=gray] (3,6) circle [radius=.15];
\node[below right] at (3,6) {$\;v_\Re$};
\draw[thin,fill=lightgray] (2.9,6.3) -- (2,9) -- (3.8,9) -- (3.1,6.3);
\node at (3,8.4) {$T_\Re$};
\draw [fill] (5.5,6) circle [radius=.1];
\node[right] at (5.5,6) {$v_1$};
\draw[thin] (5.4,6.2) -- (4.5,9) -- (6.3,9) -- (5.6,6.2);
\node at (5.5,8.2) {$T_1$};
\node at (6.8,7) {$\dots$};
\draw [fill] (8,6) circle [radius=.1];
\node[right] at (8,6) {$v_k$};
\draw[thin] (7.9,6.2) -- (7,9) -- (8.8,9) -- (8.1,6.2);
\node at (8,8.2) {$T_k$};

\draw [thick,->] (3.8,4.1) -- (1.2,5.8);
\draw [thick,->] (3.9,4.2) -- (3.1,5.75);
\draw [thick,->] (4.1,4.2) -- (5.4,5.75);
\draw [thick,->] (4.2,4.1) -- (7.8,5.9);

\node at (4,0) {$R$};
\end{tikzpicture}
\newline

\vspace*{-1cm}
\begin{tikzpicture}[scale=.5]
\draw [fill] (4,1) circle [radius=.1];
\draw [thick,->] (4,1.2) to [out=90, in=-90] (4,1.8);
\draw [fill] (4,2) circle [radius=.1];
\draw [thick,->] (4,2.2) to [out=90, in=-90] (4,2.8);
\draw [fill] (4,3) circle [radius=.1];
\draw [thick,->] (4,3.2) to [out=90, in=-90] (4,3.8);
\draw [fill=gray] (4,4) circle [radius=.15];
\node[below right] at (4,4) {$\, u$};
\draw[thin,fill=lightgray] (3.7,6.3) -- (1.4,9) -- (3.2,9) -- (3.8,6.3);
\node at (2.65,8.4) {$T_\Le$};
\draw [fill=gray] (4,6) circle [radius=.15];
\node[right] at (4,6) {$\,v_\Le$};
\draw [fill=gray] (4,8) circle [radius=.15];
\node[below right] at (4,8) {$v_\Re$};
\draw[thin,fill=lightgray] (3.9,8.3) -- (3,11) -- (4.8,11) -- (4.1,8.3);
\node at (4,10.4) {$T_\Re$};

\draw [fill] (7,6) circle [radius=.1];
\node[right] at (7,6) {$v_1$};
\draw[thin] (6.9,6.2) -- (6,9) -- (7.8,9) -- (7.1,6.2);
\node at (7,8.2) {$T_1$};
\node at (8.3,7) {$\dots$};
\draw [fill] (9.5,6) circle [radius=.1];
\node[right] at (9.5,6) {$v_k$};
\draw[thin] (9.4,6.2) -- (8.5,9) -- (10.3,9) -- (9.6,6.2);
\node at (9.5,8.2) {$T_k$};

\draw [thick,->] (4,4.25) -- (4,5.75);
\draw [thick,->] (4,6.25) -- (4,7.75);
\draw [thick,->] (4.1,4.2) -- (6.8,5.9);
\draw [thick,->] (4.2,4.1) -- (9.3,5.9);

\node at (4,0) {$R_\Le$};

\end{tikzpicture}
\hspace*{1cm}
\begin{tikzpicture}[scale=.5]
\draw [fill] (4,1) circle [radius=.1];
\draw [thick,->] (4,1.2) to [out=90, in=-90] (4,1.8);
\draw [fill] (4,2) circle [radius=.1];
\draw [thick,->] (4,2.2) to [out=90, in=-90] (4,2.8);
\draw [fill] (4,3) circle [radius=.1];
\draw [thick,->] (4,3.2) to [out=90, in=-90] (4,3.8);
\draw [fill=gray] (4,4) circle [radius=.15];
\node[below right] at (4,4) {$\, u$};
\draw[thin,fill=lightgray] (3.7,6.3) -- (1.4,9) -- (3.2,9) -- (3.8,6.3);
\node at (2.65,8.4) {$T_\Re$};
\draw [fill=gray] (4,6) circle [radius=.15];
\node[right] at (4,6) {$\,v_\Re$};
\draw [fill=gray] (4,8) circle [radius=.15];
\node[below right] at (4,8) {$v_\Le$};
\draw[thin,fill=lightgray] (3.9,8.3) -- (3,11) -- (4.8,11) -- (4.1,8.3);
\node at (4,10.4) {$T_\Le$};

\draw [fill] (7,6) circle [radius=.1];
\node[right] at (7,6) {$v_1$};
\draw[thin] (6.9,6.2) -- (6,9) -- (7.8,9) -- (7.1,6.2);
\node at (7,8.2) {$T_1$};
\node at (8.3,7) {$\dots$};
\draw [fill] (9.5,6) circle [radius=.1];
\node[right] at (9.5,6) {$v_k$};
\draw[thin] (9.4,6.2) -- (8.5,9) -- (10.3,9) -- (9.6,6.2);
\node at (9.5,8.2) {$T_k$};

\draw [thick,->] (4,4.25) -- (4,5.75);
\draw [thick,->] (4,6.25) -- (4,7.75);
\draw [thick,->] (4.1,4.2) -- (6.8,5.9);
\draw [thick,->] (4.2,4.1) -- (9.3,5.9);

\node at (4,0) {$R_\Re$};

\end{tikzpicture}
\caption{Linearising step-by-step.}\label{f:lr}
\end{figure}
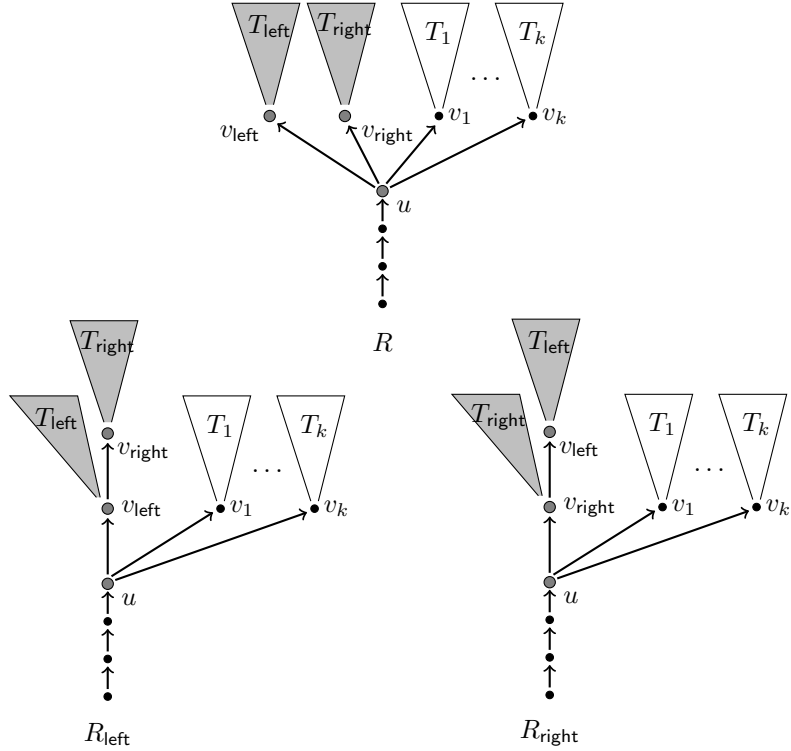
%

We claim that 
either $\M_\Le,r_\sigma\not\kmodels\tau$ or $\M_\Re,r_\sigma\not\kmodels\tau$.
Suppose otherwise. Then, by Proposition~\ref{p:hom}, there are two homomorphisms
$h_\Le\colon\M_\tau\to\M_\Le$ and $h_\Re\colon\M_\tau\to\M_\Re$ with $h_\Le(r_\tau)=h_\Re(r_\tau)=r_\sigma$.
If one of these is an $\M_\tau\to\M$ homomorphism, then 
$\M,r_\sigma\kmodels\tau$, 
contrary to IH.
If this is not the case, suppose that $\M_\tau$ is based on an irreflexive and intransitive unary tree
$x_0<x_1<\dots< x_n$.
It is not hard to see that 
\begin{itemize}
\item[--]
there is $i_\Le < n$ such that
%
$\bigl(h_\Le(x_i), v_\Le\bigr)\in R_\Le$ for every $i\le i_\Le$,
and $\bigl(v_\Re, h_\Le(x_{i})\bigr)\in R_\Le$
for every $i\geq i_\Le+1$;
%
\item[--]
there is $i_\Re < n$ such that
%
$\bigl(h_\Re(x_i), v_\Re\bigr)\in R_\Re$ for every $i\le i_\Re$,
and $\bigl(v_\Le,  h_\Re(x_{i})\bigr)\in R_\Re$
for every $i\geq i_\Re+1$.
%
\end{itemize}
Suppose $i_\Re\geq i_\Le$ (the other case is similar). Define $h$ by taking, for any $i\leq n$,
\[
h(x_i)=\left\{
\begin{array}{ll}
h_\Re(x_i), & \mbox{if } i\leq i_\Re,\\
h_\Le(x_i), & \mbox{else.}
\end{array}
\right.
\]
We prove that $h$ is an $\M_\tau\to\M$ homomorphism with $h(r_\tau)=r_\sigma$, from which we shall have 
$\M,r_\sigma\kmodels\tau$, 
contrary to IH.
Thus, we need to show that, for every $i<n$, we have $\bigl(h(x_i),h(x_{i+1})\bigr)\in R$.
There are three cases:

\emph{Case} 1: $i<i_\Re$. 
Then $h(x_i)=h_\Re(x_i)$, $h_\Re(x_{i+1})=h(x_{i+1})$ and 
$\bigl(h(x_i),h(x_{i+1})\bigr)\in R_\Re$. Since $i,i+1\leq i_\Re$,
we have $\bigl(h(x_i), v_\Re\bigr)\in R_\Re$ and $\bigl(h(x_{i+1}), v_\Re\bigr)\in R_\Re$, 
and so 
$\bigl(h(x_i),h(x_{i+1})\bigr)\in R$ follows from $\bigl( h(x_i), h(x_{i+1})\bigr)$ $\in R_\Re$.

\emph{Case} 2: $i>i_\Re$.
Then $h(x_i)=h_\Le(x_i)$, $ h_\Le(x_{i+1})=h(x_{i+1})$ and we also have $\bigl(h(x_i),h(x_{i+1})\bigr)\in R_\Le$.
As $i,\mbox{$i+1$}\geq i_\Le+1$,
we have $\bigl(v_\Re , h(x_i)\bigr)\in R_\Le$ and $\bigl(v_\Re, h(x_{i+1})\bigr)\in R_\Le $.
Therefore, we obtain $\bigl(h(x_i),h(x_{i+1})\bigr)\in R$ from $\bigl(h(x_i), h(x_{i+1})\bigr)\in  R_\Le$.

\emph{Case} 3: $i=i_\Re$.
Then $h(x_i)=h_\Re(x_i)$, and so $\bigl(h(x_i), v_\Re\bigr)\in R_\Re$.
Also, $h(x_{i+1})= h_\Le(x_{i+1})$ and, since $i+1= i_\Re+1\geq i_\Le+1$, we have  
$\bigl(v_\Re,  h(x_{i+1})\bigr)\in R_\Le$.
Therefore, $\bigl(h(x_i),h(x_{i+1})\bigr)\in R$,
as required.
\end{proof}

That $\Elin$ is \eqcomp\ follows now from Claims~\ref{c:linnormalform}, \ref{c:linnterms}, \eqcomp ness  of $\Esfour$ (Corollary~\ref{t:ss}) and \eqref{birkhoff}.
\end{proof}

As a consequence of Claims~\ref{c:linnormalform} and \ref{c:linnterms} we also obtain:

\begin{theorem}\label{t:linP}  
$\Elin$ is decidable in {\sc PTime}.
\end{theorem}

\begin{proof}
Follows from the {\sc PTime}-time decidability of $\Esfour$ 
\cite{stokkermans2008} (see also Theorem~\ref{thm:horncomplete})
and the fact that $|N_\varrho|$ in Claim~\ref{c:linnormalform} is the number of leaves in $\M_\varrho$.
\end{proof}


The completeness landscape for extensions of $\Elin$ is much more involved than for extensions of $\Esfive$. 
In \cite{aiml18}, all complete extensions of $\Elin$ are characterised, and
infinitely many incomplete extensions of $\SPi+(\Axlin \cup \{\efun^2\})$ are given. Here 
we prove the following:

\begin{theorem}\label{t:lindep}
$\SPi+(\Axlin \cup \{\efun^n\})$ is not \eqcomp, for any $n\geq 1$. 
\end{theorem}

\begin{proof}
For $n=1$, we reuse the proof of Theorem~\ref{t:fun}~(iii) since we clearly have 
$\Aa\models\ewc$.
Now, fix some $n\geq 2$. Observe that $\Elin=\SPi+(\Axsfour\cup\{\ewc'\})$, where
\begin{equation}\label{esv}
\ewc'=\bigl( \d(p\land \d q)\land \d(p\land \d r)\imp \d (p\land \d q\land \d r)\bigr).
\end{equation}
Let $\Aa_n$ be the \SLOa\ from the proof of Theorem~\ref{t:altnotcomp}. We claim
that $\Aa_n\models\ewc'$. Indeed, take a valuation $\vala$ in $\Aa_n$. If there are 
distinct $x,y\in\{\vala(p),\d q[\vala],\d r[\vala]\}$ such that $x\le y$, then $\Aa\models\ewc'[\vala]$
clearly holds. So we may assume that $\vala(p)$, $\d q[\vala]$, and $\d r[\vala]$ are pairwise
$\leq$-incomparable. Let, say, $\d q[\vala]=b_i$, $\d r[\vala]=c_i$, and $\vala(p)=d_j$, 
for some $i<n-1$ and $i+1<j\leq n$ (the other cases are similar). Then both sides of $\ewc'$
evaluate to $d_0$ if $i=0$, and to $b_{i-1}$ if $i>0$, proving that $\Aa\models\ewc'[\vala]$.
\end{proof}

\begin{question}
Is $\SPi+(\Axlin \cup \{\edep^n\})$ \eqcomp\ for $n>1$? 
\end{question}

The \spequation{} $\ewc'$ in \eqref{esv} was also used by Svyatlovsky~\cite{Svy18} who showed that $\{\etrans,\ewc'\}$ axiomatise the \SP-fragment of $\mf{K4.3}$---the modal logic of all transitive and weakly connected frames---and described the class of Kripke frames validating $\{\etrans,\ewc'\}$. As not all frames in this class are weakly connected, it follows that the class of $\mf{K4.3}$-frames is not \SP-definable. For a direct model-theoretic proof of this fact, see Proposition~\ref{p:weakly} below.  Svyatlovsky also proved that the \sptheory{}  $\SPi + \{\etrans,\ewc'\}$ is tractable.

We can generalise $\ewc$ to
\[
\ew^n=\Bigl(  \d\bigl(p\land\bigwedge P^n_0\bigr)\land\dots\land \d\bigl(p\land\bigwedge P^n_n\bigr)\imp \d
(p\land \d p_0\land\dots\land \d p_n)\Bigr).
\]
Then $\ewc^1=\ew$.

\begin{question}
Are $\SPi+(\Axsfour \cup \{\ew^n\})$ and $\SPi+(\Axsfour \cup \{\edep^m,\ew^n\})$ \eqcomp? 
\end{question}


\section{Undecidability of completeness}\label{sec:undec}

Having established quite a few \eqcomp ness and incompleteness results for 
\sptheories, we now show that an exhaustive and decidable classification of finitely axiomatisable 
\sptheories{} according to their \eqcomp ness (or \complex ity) is not possible.

\begin{theorem}\label{undec}
Given a finite set $\Ec$ of \spequations, 
it is undecidable whether the \sptheory{} $\SPi+\Sigma$ is \eqcomp{}; it is also undecidable whether it is complex. 
\end{theorem}

\begin{proof}
We encode the halting problem for deterministic Turing machines starting from an empty tape. 
Recall that a Turing machine is  a tuple
$$
M = (Q, \Gamma, \delta, q_{0}, q_{h}),
$$
where $Q$ is a non-empty finite set of states with an initial state $q_{0}$ and a halting state $q_{h}$,
$\Gamma$ is a finite tape alphabet with a special symbol $b\in \Gamma$ denoting the blank cell, and $\delta$ is a transition function that, for any pair $(q,a)\in Q\times \Gamma$, gives a triple $\delta(q,a)\in Q\times \Gamma\times \{\TMl,\TMr\}$, where $\TMl$ and $\TMr$ stand for
`move left' and `move right'\!, respectively. We use the standard 
definition of a computation of $M$ on an input word. Then the problem to decide whether the computation starting from an empty tape in state $q_{0}$ reaches the halting state $q_{h}$
is undecidable~\cite{DBLP:books/daglib/0032222}. 
We may assume that 
the initial state is not reachable from any state, 
the halting state has no successor state, and that the head never moves to the left of  its initial position. Now, suppose $M = (Q, \Gamma, \delta, q_{0}, q_{h})$ is such a Turing machine.

For the reduction, we encode the computation of $M$ starting from the empty tape by a grid with points 
$d_{n,m}$ for the $n$th cell of the $m$th configuration of the computation.
We use relations ${\sf next}$ and ${\sf step}$ such that $(d_{n,m},d_{n+1,m})\in {\sf next}$
and $(d_{n,m},d_{n,m+1})\in {\sf step}$. We encode that the $n$th cell contains symbol $a\in \Gamma$ in 
the $m$th configuration by introducing a relation $R_{a}$ and stating that
$d_{n,m}$ has some $R_{a}$-successor. Likewise, we encode that $M$ is in state $q\in Q$
in the $m$th configuration by introducing a relation $R_{q}$ and stating that all $d_{n,m}$ have 
an $R_{q}$-successor. In the same way, we use relations $R_{{\sf head}}$, $R_{{\sf left}}$, and $R_{{\sf right}}$
to encode the position of the head and, for technical reasons, all cells to the left and right of the
position of the head. 

Using the above intuition, we now construct a finite set $\AxM$ of \spequations{} such that $M$ halts on the empty tape  iff $\SPi+\AxM$ is \eqcomp{} iff $\SPi+\AxM$ is \complex.
Let $\dnext$ and $\dstep$ be modal operators interpreted by the relations ${\sf next}$ and ${\sf step}$ introduced
above. The following set of \spequations{} state that
the relations ${\sf next}$ and ${\sf step}$ are functional and commute:
\begin{align}
\label{nextfun}
& \dnext p \land \dnext q \imp \dnext(p \land q), \\
\label{stepfun}
& \dstep p\land \dstep q \imp \dstep(p \land q),\\
\label{nextstepcomm}
& \dnext\dstep p \imp \dstep \dnext p\quad\mbox{and}\quad \dstep \dnext p\imp  \dnext\dstep p.
\end{align}
To axiomatise the properties of $R_{a}$, $a\in \Gamma$, $R_{q}$, $q\in Q$, and $R_{\sf head}$, $R_{\sf left}$, and $R_{\sf right}$,
we introduce an operator $\dq$ for every state $q\in Q$, an operator $\da$ for every $a\in \Gamma$, and
operators $\dhead$, $\dleft$ and $\dright$. We say that $M$ does not halt by the \spequation{} 
%
\begin{equation}
\label{nohalt}
 \dqh \top\imp p.
\end{equation}
In order to show that what we have so far axiomatise a \complex{} \sptheory{} (see Theorem~\ref{t:funcomm}),
we also need to add the \spequations{}
\begin{equation}
\label{n2}
\d\dqh\top\imp \dqh\top,
\end{equation}
for all $\d=\dnext,\,\dstep,\,\dhead,\,\dleft,\,\dright,\,\dq,\,\da$, $q\in Q$, $a\in\Gamma$.
(Note that if the language contained a constant $\bot$, interpreted as `falsehood' in Kripke models and the $\leq$-smallest element in `normal' \SLOa{}s, then  $\dqh \top\imp\bot$ would suffice in place of  \eqref{nohalt}--\eqref{n2}, see \S\ref{bot}.) 
%
Let $\varXi$ be the set of \spequations{} comprising \eqref{nextfun}--\eqref{n2}. By Theorem~\ref{t:funcomm}, the \sptheory{} $\SPi + \varXi$ is \complex. 
To ensure that $\varXi$ together with the set of \spequations{} encoding the computation of $M$ on empty tape axiomatise
a \complex\ \sptheory, we apply Proposition~\ref{p:varfree}, and therefore represent states, tape symbols and tape positions using variable-free \spterms{} of the form $\dR\top$ for the operators $\dR$ introduced above. 
We first set ${\sf left}$ and ${\sf right}$ correctly, exploiting the assumed functionality of  ${\sf next}$:
\begin{align}
\label{leftgen1}
& \dnext\dleft \top \imp \dleft \top,\\
\label{leftgen2}
& \dnext\dhead \top\imp \dleft \top,\\
\label{rightgen1}
& \dhead \top \imp \dnext\dright\top, \\
\label{rightgen2}
& \dright\top \imp \dnext\dright\top.
\end{align}
Then we say that the state of each configuration is encoded in a uniform way over the tape:
for all $q\in Q$,
\begin{equation}
\label{qconfig}
\dq\top \imp  \dnext\dq\top
\ \mbox{ and }\ 
\dnext\dq\top\imp \dq\top.
\end{equation}
Exploiting that $q_0$ is not reachable from any state,
we can say that the tape is initially blank with
\begin{equation}
\label{startsblank}
\dqo\top \imp \d_b\top.
\end{equation}
Exploiting the commutativity and functionality of ${\sf next}$ and ${\sf step}$, for each transition 
$\delta(q,a)=(q',a',\TMl)$, we set
\begin{equation}
\label{movel}
\dnext(\dq\top\land \dhead \top\land \da\top) \imp \dstep (\dqp\top\land \dhead \top\land \dnext\dap\top),
\end{equation}
and for each transition $\delta(q,a)=(q',a',\TMr)$, we set
\begin{equation}
\label{mover}
\dq\top\land \dhead \top\land\da\top \imp \dstep (\dap\top\land\dqp\top\land \dnext\dhead \top).
\end{equation}
We also say that symbols not under the head do not change: for all
$a \in \Gamma$, put
\begin{align}
\label{nochange1}
& \da\top\land \dleft \top \imp \dstep \da\top,\\
\label{nochange2}
& \da\top\land \dright\top \imp \dstep \da\top.
\end{align}
%
%
%
%
Let $\AxM^0$ be $\varXi$ together with the \spequations{} \eqref{leftgen1}--\eqref{nochange2}. 
%
%
%
Finally, we obtain $\AxM$ from $\AxM^0$ by adding the following
\spequation{} that triggers incompleteness whenever $\dqo\top \land \Diamond_{{\sf head}}\top$ is satisfiable in
a frame for $\AxM^0$:
\begin{equation}\label{eqincomp}
\TMeq=(\dqo\top \land \Diamond_{{\sf head}}\top \land \dR p \imp p),
\end{equation}
where $R$ is a fresh relation. 

\begin{lclaim}\label{c:haltstoc}
If $M$ halts on the empty tape, then $\SPi+\AxM$ is \complex.
\end{lclaim}

\begin{proof}
Suppose $M$ halts on the empty tape in $H<\omega$ steps.
As $\SPi+\AxM^0$ is \complex{} by Theorem~\ref{t:funcomm} and Proposition~\ref{p:varfree},
it is enough to show that $\SPi+\AxM=\SPi+\AxM^0$.
We prove that 
\begin{multline}
\label{computeq}
\{w\mid \M,w\kmodels\dqo \top \land \dhead \top\}=\emptyset\\
\mbox{ for any model $\M$ over any frame $\F$ for $\AxM^0$.}
\end{multline}
Then $\AxM^0\models_{\CA} \TMeq$ would follow, and so $\TMeq\in\SPi+\AxM^0$ would hold by 
the \eqcomp ness of $\SPi+\AxM^0$.

To prove \eqref{computeq},
take any frame $\F\models\AxM^0$ and  
suppose to the contrary that there is some $d_{0,0}$ with
$\M,d_{0,0}\kmodels\dqo\top \land \d_{{\sf head}}\top$.
We show by induction on $m$ that, for any $m\leq H$
and $n<\omega$, there exists $d_{n,m}$ in $\F$ 
representing the $n$th cell in the $m$th configuration of the
computation of $M$ in the following sense: 
for all $q\in Q$ and $a\in\Gamma$,
\begin{itemize}
\item[(\emph{i})]
$(d_{n,m},d_{n+1,m})\in {\sf next}$;
\item[(\emph{ii})]
$(d_{n,m},d_{n,m+1})\in {\sf step}$ whenever $m<H$;
\item[(\emph{iii})]
 if the state in the $m$th configuration is $q$, then $\M,d_{n,m}\kmodels\dq\top$;
 \item[(\emph{iv})]
 if the $n$th cell contains $a$ in the  $m$th configuration, then $\M,d_{n,m}\kmodels\da\top$;
 \item[(\emph{v})]
 if the head is at the $n$th cell in the $m$th configuration, then $\M,d_{n,m}\kmodels\dhead \top$.
\end{itemize}
%
%
Indeed, for $m=n=0$, (\emph{iii})--(\emph{v}) follow from our assumption and \eqref{startsblank}. 
We have $d_{n,0}$ for all $n>0$ satisfying (\emph{i}), (\emph{iii}) and (\emph{iv}) by \eqref{qconfig} and \eqref{startsblank}.
Now suppose inductively
that we have $d_{n,m}$ for some $m<H$ and all $n<\omega$. Suppose that 
in the $m$th configuration the head is at the $n$th cell containing symbol $a$, $M$ is in state $q$
and $\delta(q,a)=(q',a',R)$. (The case when $\delta(q,a)=(q',a',L)$ is similar and left to the reader.)
Then, by IH, 
$\M,d_{n,m}\kmodels\dq\top\land\dhead \top\land\da\top$, 
and so, by \eqref{mover}, there exist $d_{n,m+1}$ and $d_{n+1,m+1}$ such that 
$(d_{n,m},d_{n,m+1})\in {\sf step}$, $(d_{n,m+1},d_{n+1,m+1})\in {\sf next}$,
$\M,d_{n,m+1}\kmodels\dap\top\land\dqp\top$ and $\M,d_{n+1,m+1}\kmodels\dhead \top$.
%
%
If $n>0$ then we have $d_{i,m+1}$ for all $i<n$ satisfying (\emph{i}), (\emph{ii}) and (\emph{iv}) by \eqref{nextstepcomm}, \eqref{leftgen1}, \eqref{leftgen2}, \eqref{nochange1} and \eqref{stepfun}.
We have $d_{i,m+1}$ for all $i>n+1$ satisfying (\emph{i}) and (\emph{ii}) by \eqref{qconfig}, \eqref{nextfun} and
\eqref{nextstepcomm}. Then  $d_{i,m+1}$ for all $i\geq n+1$ satisfy (\emph{iv}) by
\eqref{rightgen1}, \eqref{rightgen2}, \eqref{nextfun} and \eqref{nochange2}. 
Finally, we have (\emph{iii}) by \eqref{qconfig} and \eqref{nextfun}.

Thus, 
$\M,d_{n,H}\kmodels\dqh\top$ 
for some $n$, and so the relation $R_{q_h}$ in $\F$ interpreting $\dqh$ is not empty, contrary to 
$\F\models\eqref{nohalt}$. 
 This establishes   \eqref{computeq}.
\end{proof}

\begin{lclaim}\label{c:ctohalts} 
If $M$ does not halt on the empty tape, then $\SPi+\AxM$ is incomplete.
\end{lclaim}

\begin{proof}
Consider the \spequation{} 
%
\[
\e' = (\dqo\top \land \dhead \top \land p\land \dR\top  \imp \dR p).
\]
On the one hand, it is easy to see that $\{\TMeq\}\models_{\CA}\e'$ (cf.\ Example~\ref{e:simple}),
and so $\AxM\models_{\CA}\e'$.
On the other hand, take the
infinite computation of $M$ starting from the empty tape. Using this computation, we define a frame $\F$ 
with domain $W=\{d_{n,m} \mid n,m<\omega\}\cup \{g,g'\}$ by taking:
\begin{itemize}
\item[--] $(d_{n,m},d_{n+1,m})\in {\sf next}$ for all $n,m<\omega$;
\item[--] $(d_{n,m},d_{n,m+1})\in {\sf step}$ for all $n,m<\omega$;
\item[--] $(d_{n,m},g)\in R_{q}$ if the state of the $m$th configuration is $q$, for $q\in Q$;
\item[--] $(d_{n,m},g)\in R_{a}$ if the $n$th cell contains  $a$ in the  $m$th configuration, for $a\in\Gamma$;
\item[--] $(d_{n,m},g)\in {\sf head}$ if the head  is at the $n$th cell in the $m$th configuration;
\item[--] $(d_{n,m},g)\in {\sf left}$ if the head is to the right of the $n$th cell in the $m$th configuration;
\item[--] $(d_{n,m},g)\in {\sf right}$ if the head  is to the left of the $n$th cell in the $m$th configuration;
\item[--] $(d_{0,0},g')\in R$.
\end{itemize}
It is straightforward to check that $\F\models\AxM^0$.
Define an \sptype{} subalgebra 
$\A$ of $\Fslo$ by taking all subsets of $W$ except those that contain $g'$ but not $d_{0,0}$. 
Then $\A\models\AxM^0$. It is easy to see that $\A$ is a \SLOa\ and $\A\models\TMeq$, and 
so $\A\models\AxM$. However, $\A\not\models\e'$, witnessed by evaluating $p$ to $\{d_{0,0}\}$. Thus, $\AxM\not\models_{\SLO}\e'$, and so $\SPi+\AxM$ is incomplete.
\end{proof}

Now, Theorem~\ref{undec} follows from Claims~\ref{c:haltstoc} and \ref{c:ctohalts}.
\end{proof}

\begin{question}\label{q:undec}
Does Theorem~\ref{undec} hold in the unimodal case? Does it hold for 
\sptheories{} with Horn correspondents?
\end{question}


\section{Some related topics}\label{sec:other}


\subsection{\SPcap-definability}\label{eldef}

A class $\Cc$ of frames is called \emph{\SP-definable} if
$\Cc=\CA_{\Ec}$, for some set $\Ec$ of \spequations. 
In this section, we prove a necessary condition for \SP-definability and use it to give a few examples
of modally definable frame classes that are not \SP-definable. To keep the notation
simple, we formulate everything for the unimodal setting only, that is, for $\R=\{R\}$.

Suppose that 
$\F_i=(W_i,R_i)$, for $i\in I$, $\G=(W,R^\G)$, $\T=(T,R^\T)$ are frames, 
$w\in T$,
$g_i\colon\T\to\F_i$, for $i\in I$, and $h\colon \T\to \G$ are homomorphisms, and that 
$Z\subseteq \bigl(\prod_{i\in I}W_i\bigr)\times W$. We write 
\[
(\F_i,g_i)_{i\in I}>\!>_Z(\G,h,w)
\]
if the following conditions hold: 
\begin{itemize}
\item[{\bf (s1)}]
$\bigl((g_i(w))_{i\in I},h(w)\bigr)\in Z$;
\item[{\bf (s2)}]
for all $(\avec{x},y)\in Z$ and $\avec{x}'=(x'_i\in W_i \mid i\in I)$, if $(x_i,x'_i)\in R_i$ for all $i\in I$, then
there is $y'$ such that $(y,y')\in R^\G$ and  $(\avec{x}',y')\in Z$;
\item[{\bf (s3)}]
for all $(\avec{x},y)\in Z$ and $A\subseteq T$, if $x_i\in g_i[A]$  for all $i\in I$, then
$y\in h[A]$.
\end{itemize}
We write
\[
(\F_i)_{i\in I} >\!>\!\> \G
\]
if, for all finite trees $\T$ with root $w$ and all homomorphisms $h\colon \T\to \G$, there exist $(g_i)_{i\in I}$ and $Z$ such 
that  $(\F_i,g_i)_{i\in I}>\!>_Z(\G,h,w)$.

\begin{theorem}\label{t:sim}
For any \spequation{} $\e$, if $(\F_i)_{i\in I} >\!>\!\> \G$ and 
$\F_i\models\e$ for all $i\in I$, then $\G\models\e$.
\end{theorem}

\begin{proof}
Suppose $\e=(\sigma\imp\tau)$. It is enough to show that, for the correspondent $\Psi_{\e}$ of $\e$ from \eqref{eq:corr1}--\eqref{globalcorr}, $\G\models\Psi_{\e}$ holds whenever
$\F_i\models\Psi_{\e}$ for $i\in I$.
Recall the respective tree models $\M_\sigma$ and $\M_\tau$ from \S\ref{sec:treemodel}, and
let  $W_\sigma=\{v_0,\dots,v_{n_\sigma}\}$ with $v_0=r_\sigma$. Let
$x_0,\dots,x_{n_\sigma}$ be a sequence of points in $\G$ such that
$(x_k,x_\ell)\in R^\G$ whenever $(v_k,v_\ell)\in R_\sigma$.
%
%
Then $h^\sigma\colon\T_\sigma\to \G$ defined by $h^\sigma(v_k)=x_k$, for $k\leq n_\sigma$,
is a homomorphism. As $(\F_i)_{i\in I} >\!>\!\> \G$, there are $(g_i^\sigma)_{i\in I}$ and $Z$ such that
\begin{align}
\label{sigmahomo}
& \mbox{$g_i^\sigma\colon\T_\sigma\to\F_i$ are homomorphisms for all $i\in I$,}\\
\label{sigmasimul}
& (\F_i,g_i^\sigma)_{i\in I}>\!>_Z(\G,h^\sigma,r_\sigma).
\end{align}
%
%
%
Since $\F_i \models \Psi_{\e}$, it follows that $\F_i \models \Psi'_{\e}[g_i^\sigma(r_\sigma)/\varel{v}_0]$ for all $i\in I$, and so, by \eqref{sigmahomo}, there exist homomorphisms $g_i^\tau\colon \T_\tau \to \F_i$ such that
\begin{equation}\label{respect}
g_i^\tau(r_\tau)=g_i^\sigma(r_\sigma)\land
\bigwedge_{\substack{(u,p)\\ u \in\V_\tau(p)}} \bigvee_{\ \ v \in\V_\sigma(p)} \bigl(g_i^\tau(u)=g_i^\sigma(v)\bigr),\quad
\mbox{for all $i\in I$.}
\end{equation}
We define a homomorphism $h^\tau\colon \T_\tau \to \G$ such that
\begin{equation}\label{inZ}
\bigl((g_i^\tau(u))_{i\in I},h^\tau(u)\bigr) \in Z,\quad\mbox{for all $u \in W_\tau$}
\end{equation}
in a step-by-step manner, by constructing
its approximations $f_0, f_1, f_2, \dots$ with domains $B_0, B_1, \dots$ which are subsets of
$W_\tau$ and initial segments of $\T_\tau$. 
To begin with, let $B_0 = \{r_\tau\}$ and $f_0 = \bigl\{\bigl(r_\tau, h^\sigma(r_\sigma)\bigr)\bigr\}$. 
By  \eqref{respect} and {\bf(s1)} of \eqref{sigmasimul},  
$\bigl((g_i^\tau(r_\tau))_{i\in I},f_0(r_\tau)\bigr)=
\bigl((g_i^\sigma(r_\sigma))_{i\in I},h^\sigma(r_\sigma)\bigr)  \in Z$.
So suppose $B_l$ and $f_l$ are defined for some $l$,
and we have 
$\bigl((g_i^\tau(u))_{i\in I},f_l(u)\bigr) \in Z$ for all $u \in B_l$ (IH).
Take some $x \in B_l$ and $y \not\in B_l$ such that $(x,y)\in R_\tau$.
Since all $g_i^\tau$ are homomorphisms,
we have $\bigl(g_i^\tau(x), g_i^\tau(y)\bigr)\in R_i$.
By IH, $\bigl((g_i^\tau(x))_{i\in I},f_l(x)\bigr) \in Z$, and so, 
by {\bf(s2)} of \eqref{sigmasimul}, there is $z \in W$ such that $\bigl(f_l(x), z\bigr)\in R^\G$ and
\mbox{$\bigl((g_i^\tau(y))_{i\in I},z\bigr) \in Z$.}
Thus, we may extend $B_l$ and $f_l$ by setting
$B_{l+1} = B_l \cup \{y\}$ and $f_{l+1} = f_l \cup \{(y, z)\}$ while preserving IH.
Clearly, $h^\tau=\bigcup_{l<\omega}f_l$ is a homomorphism as required in \eqref{inZ}.

Suppose $W_\tau=\{u_0,\dots,u_{n_\tau}\}$ with $u_0=r_\tau$. We claim that 
\begin{multline}\label{tausimul}
\G\models
\bigl((\varel{v}_0 = \varel{u}_0) \land \!\!\bigwedge_{\stackrel{k,\ell\leq n_\tau,}{(u_k,u_\ell)\in R_\tau}}\!\! R(\varel{u}_k,\varel{u}_\ell) 
\land 
 \\[-12pt]
\bigwedge_{\stackrel{k\leq n_\tau,\,p\in\varset,}{u_k \in\V_\tau(p)}} \bigvee_{\stackrel{\ell\leq n_\tau,}{v_\ell\in\V_\sigma(p)}}\!\!(\varel{u}_k=\varel{v}_\ell) \big)
\bigl[h^\sigma(\avec{v})/\varel{\avec{v}},h^\tau(\avec{u})/\varel{\avec{u}}\bigr],
\end{multline}
proving $\G\models\Psi_{\e}$. Indeed,
$h^\sigma(r_\sigma)=h^\tau(r_\tau)$ and $h^\tau$ is a homomorphism, so it is enough to show
the second line in \eqref{tausimul}.
Fix $u_k$ and $p$ such that $u_k\in \V_{\tau}(p)$.
By \eqref{respect}, for any $i\in I$, there is $v_i\in \V_\sigma(p)$ with  
$g_i^\tau(u_k)=g_i^\sigma(v_i)$, and so $g_i^\tau(u_k)\in g_i^\sigma[\V_\sigma(p)]$ for all $i\in I$.
By \eqref{inZ} and {\bf(s3)} of \eqref{sigmasimul}, we have $h^\tau(u_k)\in h^\sigma[\V_\sigma(p)]$,
and so there is some $v_\ell\in\V_\sigma(p)$ with $h^\tau(u_k)= h^\sigma(v_\ell)$, as required in
\eqref{tausimul}.
\end{proof}

In certain cases, we may simplify the criterion of the previous theorem:

\begin{proposition}\label{p:simple}
If 
there exist homomorphisms \mbox{$f_i\colon \G\to \F_i$,} for $i\in I$, and $Z$ such that 
$(\F_i,f_i)_{i\in I}>\!>_Z(\G,\idmap,v)$, for all $v$ in $\G$ and the identity map $\idmap\colon\G\to\G$,
then $(\F_i)_{i\in I} >\!>\!\> \G$.
\end{proposition}

\begin{proof}
Suppose $h\colon\T\to\G$ is a homomorphism, for a finite tree $\T$ with root $w$.
Let $v=h(w)$. By our assumption, there are homomorphisms \mbox{$f_i\colon \G\to \F_i$,} 
for $i\in I$, and $Z$ such that 
$(\F_i,f_i)_{i\in I}>\!>_Z(\G,\idmap,v)$.
Define $g_i\colon\T\to \F_i$ by $g_i=f_i\circ h$, for $i\in I$.
Then it is not hard to check that 
$(\F_i,g_i)_{i\in I}>\!>_Z(\G,h,w)$.
\end{proof}


\smallskip
A relation $R$ is called \emph{pseudo-transitive} if
\[
\forall x,y,z\,\bigl(R(x,y)\land R(y,z)\to R(x,z)\lor (x=z)\bigr);
\]
$R$ is  \emph{pseudo-equivalence} if it is symmetric and pseudo-transitive.
Pseudo-equiva\-lence relations are the frames for the modal logic $\mf{Diff}$, also characterised
by the $\ne$ relation on nonempty sets.

\begin{proposition}\label{p:diff}
Neither the class of all pseudo-transitive nor the class of all pseudo-equivalence frames
is \SP-definable.
\end{proposition}

\begin{proof}
Take the frames $\F_1$, $\F_2$ and $\G$ in Fig.~\ref{f:defpic1}.
We show that the conditions of Proposition~\ref{p:simple} hold, and so $(\F_1, \F_2) >\!> \G$.
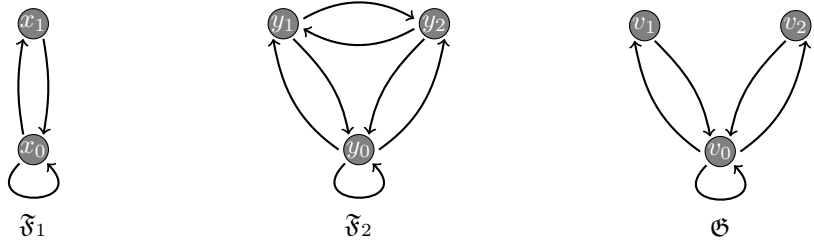
\begin{figure}[ht]
\centering
\begin{tikzpicture}[scale=.67]
\tikzset{every loop/.style={thick,min distance=6mm,in=-45,out=-135,looseness=6}}
\tikzset{place/.style={circle,thin,draw=white,fill=white,scale=1.5}}

\node[place] (foo) at (1,1.5) {}; 
\draw [fill=gray] (1,1.5) circle [radius=.3];
\node at (1,1.5)  {\textcolor{white}{$\boldmath{x_0}$}}; 
\draw [fill=gray] (1,4) circle [radius=.3];
\node at (1,4) {\textcolor{white}{$\boldmath{x_1}$}}; 
\draw [thick,->] (.8,1.8) to [out=100, in=-100] (.8,3.7);
\draw [thick,->] (1.2,3.7) to [out=-80, in=80] (1.2,1.8);
\path[->] (foo) edge [loop] node {} ();

\node at (1,0) {$\F_1$};
\end{tikzpicture}
\hspace*{2cm}
\begin{tikzpicture}[scale=.67]
\tikzset{every loop/.style={thick,min distance=6mm,in=-45,out=-135,looseness=6}}
\tikzset{place/.style={circle,thin,draw=white,fill=white,scale=1.5}}

\node[place] (foo) at (1.5,1.5) {}; 
\draw [fill=gray] (1.5,1.5) circle [radius=.3];
\node at (1.5,1.5)  {\textcolor{white}{$\boldmath{y_0}$}}; 
\draw [fill=gray] (0,4) circle [radius=.3];
\node at (0,4)  {\textcolor{white}{$\boldmath{y_1}$}}; 
\draw [fill=gray] (3,4) circle [radius=.3];
\node at (3,4)  {\textcolor{white}{$\boldmath{y_2}$}}; 
\draw [thick,->] (1.1,1.5) to [out=145, in=-80] (-.2,3.7);
\draw [thick,->] (.2,3.7) to [out=-45, in=100] (1.3,1.8);
\draw [thick,->] (1.9,1.5) to [out=35, in=-100] (3.2,3.7);
\draw [thick,->] (2.8,3.7) to [out=-135, in=80] (1.7,1.8);
\draw [thick,->] (.4,4.1) to [out=30, in=150] (2.6,4.1);
\draw [thick,->] (2.6,3.9) to [out=-150, in=-30] (.4,3.9);
\path[->] (foo) edge [loop] node {} ();

\node at (1.5,0) {$\F_2$};
\end{tikzpicture}
\hspace*{2cm}
\begin{tikzpicture}[scale=.67]
\tikzset{every loop/.style={thick,min distance=6mm,in=-45,out=-135,looseness=6}}
\tikzset{place/.style={circle,thin,draw=white,fill=white,scale=1.5}}

\node[place] (foo) at (1.5,1.5) {}; 
\draw [fill=gray] (1.5,1.5) circle [radius=.3];
\node at (1.5,1.5)  {\textcolor{white}{$\boldmath{v_0}$}}; 
\draw [fill=gray] (0,4) circle [radius=.3];
\node at (0,4)  {\textcolor{white}{$\boldmath{v_1}$}}; 
\draw [fill=gray] (3,4) circle [radius=.3];
\node at (3,4)  {\textcolor{white}{$\boldmath{v_2}$}}; 
\draw [thick,->] (1.1,1.5) to [out=145, in=-80] (-.2,3.7);
\draw [thick,->] (.2,3.7) to [out=-45, in=100] (1.3,1.8);
\draw [thick,->] (1.9,1.5) to [out=35, in=-100] (3.2,3.7);
\draw [thick,->] (2.8,3.7) to [out=-135, in=80] (1.7,1.8);
\path[->] (foo) edge [loop] node {} ();

\node at (1.5,0) {$\G$};
\end{tikzpicture}
\caption{Frames showing \SP-undefinability of pseudo-transitivity.}\label{f:defpic1}
\end{figure}
Consider the homomorphism $f_1:\G\to \F_1$ where $f_1(v_0)=x_0$, $f_1(v_1)=f_1(v_2)=x_1$,
and the homomorphism $f_2:\G\to \F_2$ where $f_2(v_i)=y_i$ for $i\le 2$, and let
%
\[
Z=\{(x_0,y_0,v_0),(x_1,y_1,v_1),(x_1,y_2,v_2),
(x_0,y_1,v_0),(x_0,y_2,v_0),(x_1,y_0,v_0)\}.
\]
We claim that, for all $i\le 2$, we have $\bigl((\F_1,f_1), (\F_2,f_2)\bigr) >\!>_Z (\G,\idmap,v_i)$.
Indeed, {\bf (s1)} clearly holds. It is easy to check that {\bf (s2)} holds, because
\begin{itemize}
\item[--] for all $(x,y)\in \F_1\times \F_2$, there is $v\in\G$ with $(x,y,v)\in Z$,
\item[--] for all $v$ in $\G$, we have $(v,v_0)\in R^\G$ and $(v_0,v)\in R^\G$.
\end{itemize}
Finally, we leave it to the reader to consider all 7 possible cases for  the non-empty  set $A\subseteq\{v_0,v_1,v_2\}$ and show {\bf (s3)}.
%
\end{proof}


\smallskip
Recall that a relation $R$ is called \emph{weakly connected} if
\[
\forall x,y,z\,\bigl(R(x,y)\land R(x,z)\to R(y,z)\lor R(z,y)\lor (y=z)\bigr).
\]
Transitive and weakly connected relations are the frames for the modal logic $\mf{K4.3}$. 
Note that the class of reflexive, transitive and weakly connected relations---linear
quasiorders, the frames for  
$\mf{S4.3}$---is
\SP-definable; see $\Axlin$ in \eqref{ewconn}. 

\begin{proposition}\label{p:weakly}
Neither the class of all weakly connected nor the class of all transitive and weakly connected frames is \SP-definable.
\end{proposition}

\begin{proof}
Take the frames $\F_1$, $\F_2$ and $\G$ in Fig.~\ref{f:defpic2}.
We show that the conditions of Proposition~\ref{p:simple} hold, and so $(\F_1, \F_2) >\!> \G$.
\begin{figure}[ht]
\centering
\begin{tikzpicture}[scale=.67]

\draw [fill=gray] (1,1) circle [radius=.3];
\node at (1,1)  {\textcolor{white}{$\boldmath{x_0}$}}; 
\draw [fill=gray] (1,3.5) circle [radius=.3];
\node at (1,3.5) {\textcolor{white}{$\boldmath{x_1}$}}; 
\draw [thick,->] (1,1.4) to [out=90, in=-90] (1,3.1);

\node at (1,0) {$\F_1$};
\end{tikzpicture}
\hspace*{2cm}
\begin{tikzpicture}[scale=.67]

\draw [fill=gray] (1.5,1) circle [radius=.3];
\node at (1.5,1)  {\textcolor{white}{$\boldmath{y_0}$}}; 
\draw [fill=gray] (0,3.5) circle [radius=.3];
\node at (0,3.5)  {\textcolor{white}{$\boldmath{y_1}$}}; 
\draw [fill=gray] (3,3.5) circle [radius=.3];
\node at (3,3.5)  {\textcolor{white}{$\boldmath{y_2}$}}; 
\draw [thick,->] (1.3,1.4) to [out=120, in=-60] (.2,3.1);
\draw [thick,->] (1.7,1.4) to [out=60, in=-120] (2.8,3.1);
\draw [thick,->] (.4,3.5) to [out=0, in=180] (2.6,3.5);

\node at (1.5,0) {$\F_2$};
\end{tikzpicture}
\hspace*{2cm}
\begin{tikzpicture}[scale=.67]
\draw [fill=gray] (1.5,1) circle [radius=.3];
\node at (1.5,1)  {\textcolor{white}{$\boldmath{v_0}$}}; 
\draw [fill=gray] (0,3.5) circle [radius=.3];
\node at (0,3.5)  {\textcolor{white}{$\boldmath{v_1}$}}; 
\draw [fill=gray] (3,3.5) circle [radius=.3];
\node at (3,3.5)  {\textcolor{white}{$\boldmath{v_2}$}}; 
\draw [thick,->] (1.3,1.4) to [out=120, in=-60] (.2,3.1);
\draw [thick,->] (1.7,1.4) to [out=60, in=-120] (2.8,3.1);

\node at (1.5,0) {$\G$};
\end{tikzpicture}
\caption{Frames showing \SP-undefinability of weak connectedness.}\label{f:defpic2}
\end{figure}
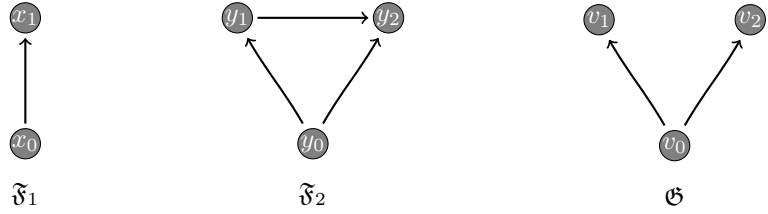
Consider the homomorphism $f_1\colon\G\to \F_1$, where $f_1(v_0)=x_0$, $f_1(v_1)=f_1(v_2)=x_1$,
and the homomorphism $f_2\colon\G\to \F_2$, where $f_2(v_i)=y_i$ for $i\le 2$, and let
\[
Z=\{ (x_0,y_0,v_0),(x_1,y_1,v_1),(x_1,y_2,v_2) \}.
\]
Then it is easy to check that
$\bigl((\F_1,f_1), (\F_2,f_2)\bigr) >\!>_Z$ $(\G,\idmap,v_i)$, for $i\le 2$.
\end{proof}


A relation $R$ is called \emph{confluent} if
\[
\forall x,y,z\,\bigl(R(x,y)\land R(x,z)\to \exists u\,\bigl(R(y,u)\land R(z,u)\bigr)\bigr).
\]
Transitive and confluent relations are the frames for the modal logic $\mf{K4.2}$. 

\begin{proposition}\label{p:chr}
Neither the class of all confluent nor the class of all transitive and confluent frames
is \SP-definable.
\end{proposition}

\begin{proof}
Take the frames $\F$ and $\G$ in Fig.~\ref{f:defpic3}.
We show that the conditions of Proposition~\ref{p:simple} hold, and so $(\F) >\!> \G$.
\begin{figure}[ht]
\begin{tikzpicture}[scale=.67]
\tikzset{every loop/.style={thick,min distance=6mm,in=45,out=135,looseness=6}}
\tikzset{place/.style={circle,thin,draw=white,fill=white,scale=1.5}}

\node[place] (foo1) at (0,3.5) {}; 
\node[place] (foo2) at (3,3.5) {}; 
\node[place] (foo3) at (1.5,6) {}; 
\draw [fill=gray] (1.5,1) circle [radius=.3];
\node at (1.5,1)  {\textcolor{white}{$\boldmath{x_0}$}}; 
\draw [fill=gray] (0,3.5) circle [radius=.3];
\node at (0,3.5)  {\textcolor{white}{$\boldmath{x_1}$}}; 
\draw [fill=gray] (3,3.5) circle [radius=.3];
\node at (3,3.5)  {\textcolor{white}{$\boldmath{x_2}$}}; 
\draw [fill=gray] (1.5,6) circle [radius=.3];
\node at (1.5,6)  {\textcolor{white}{$\boldmath{x_3}$}}; 
\draw [thick,->] (1.3,1.4) to [out=120, in=-60] (.2,3.1);
\draw [thick,->] (1.7,1.4) to [out=60, in=-120] (2.8,3.1);

\draw [thick,->] (1.5,1.4) to (1.5,5.6);
\draw [thick,->] (0.4,3.5) to [out=45, in=-100] (1.35,5.6);
\draw [thick,->] (2.6,3.5) to [out=135, in=-80] (1.65,5.6);
\path[->] (foo1) edge [loop] node {} ();
\path[->] (foo2) edge [loop] node {} ();
\path[->] (foo3) edge [loop] node {} ();

\node at (1.5,0) {$\F$};
\end{tikzpicture}
\hspace*{1.5cm}
\begin{tikzpicture}[scale=.67]
\tikzset{every loop/.style={thick,min distance=6mm,in=45,out=135,looseness=6}}
\tikzset{place/.style={circle,thin,draw=white,fill=white,scale=1.5}}

\node[place] (foo1) at (0,4.5) {}; 
\node[place] (foo2) at (3,4.5) {}; 
\draw [fill=gray] (1.5,2) circle [radius=.3];
\node at (1.5,2)  {\textcolor{white}{$\boldmath{v_0}$}}; 
\draw [fill=gray] (0,4.5) circle [radius=.3];
\node at (0,4.5)  {\textcolor{white}{$\boldmath{v_1}$}}; 
\draw [fill=gray] (3,4.5) circle [radius=.3];
\node at (3,4.5)  {\textcolor{white}{$\boldmath{v_2}$}}; 
\draw [thick,->] (1.3,2.4) to [out=120, in=-60] (.2,4.1);
\draw [thick,->] (1.7,2.4) to [out=60, in=-120] (2.8,4.1);
\path[->] (foo1) edge [loop] node {} ();
\path[->] (foo2) edge [loop] node {} ();

\node at (1.5,1) {$\G$};
\node at (1.5,0) {\ };
\end{tikzpicture}
\caption{Frames showing \SP-undefinability of confluence.}\label{f:defpic3}
\end{figure}
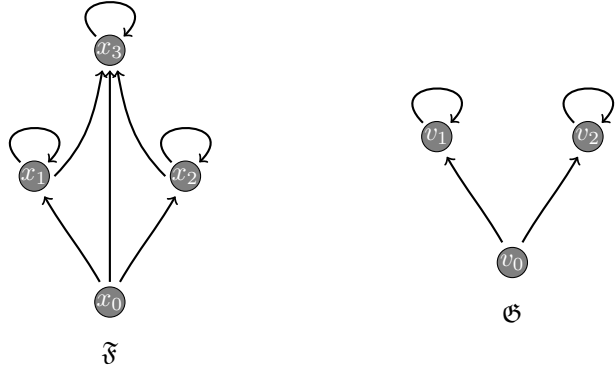
Consider the homomorphism $f\colon\G\to \F$, where $f(v_i)=x_i$ for $i\leq 2$,
and let
\[
Z=\{ (x_0,v_0),(x_1,v_1),(x_2,v_2),(x_3,v_1),(x_3,v_2) \}.
\]
Then it is easy to check that
$(\F,f) >\!>_Z$ $(\G,\idmap,v_i)$,  for $i\le 2$.
\end{proof}


We say that a relation $R$ \emph{has the McKinsey property} if
\[
\forall x \exists y \, \bigl(R(x,y)\land \forall z\,\bigl( R(y,z)\to (y=z)\bigr)\bigr).
\]
Transitive relations with this property are the frames for the modal logic $\mf{K4.1}$. 

\begin{proposition}\label{p:chr}
The class of all transitive frames with the McKinsey property 
is not \SP-definable.
\end{proposition}

\begin{proof}
Take the frames $\F$ and $\G$ in Fig.~\ref{f:defpic4}.
We show that the conditions of Proposition~\ref{p:simple} hold, and so $(\F) >\!> \G$.
\begin{figure}[ht]
\begin{tikzpicture}[scale=.67]
\tikzset{every loop/.style={thick,min distance=6mm,in=-45,out=-135,looseness=6}}
\tikzset{place/.style={circle,thin,draw=white,fill=white,scale=1.5}}

\node[place] (foo1) at (0,1.5) {}; 
\node[place] (foo2) at (3,1.5) {}; 
\draw [fill=gray] (1.5,4) circle [radius=.3];
\node at (1.5,4)  {\textcolor{white}{$\boldmath{x_2}$}}; 
\draw [fill=gray] (0,1.5) circle [radius=.3];
\node at (0,1.5)  {\textcolor{white}{$\boldmath{x_0}$}}; 
\draw [fill=gray] (3,1.5) circle [radius=.3];
\node at (3,1.5)  {\textcolor{white}{$\boldmath{x_1}$}}; 
\draw [thick,->] (.4,1.6) to [out=30, in=150] (2.6,1.6);
\draw [thick,->] (2.6,1.4) to [out=-150, in=-30] (.4,1.4);
\draw [thick,->] (0.2,1.9) to [out=60, in=-120] (1.3,3.6);
\draw [thick,->] (2.8,1.9) to [out=120, in=-60] (1.7,3.6);
\path[->] (foo1) edge [loop] node {} ();
\path[->] (foo2) edge [loop] node {} ();

\node at (1.5,0) {$\F$};
\end{tikzpicture}
\hspace*{1.5cm}
\begin{tikzpicture}[scale=.67]
\tikzset{every loop/.style={thick,min distance=6mm,in=-45,out=-135,looseness=6}}
\tikzset{place/.style={circle,thin,draw=white,fill=white,scale=1.5}}

\node[place] (foo1) at (0,2.5) {}; 
\node[place] (foo2) at (3,2.5) {}; 
\draw [fill=gray] (0,2.5) circle [radius=.3];
\node at (0,2.5)  {\textcolor{white}{$\boldmath{v_0}$}}; 
\draw [fill=gray] (3,2.5) circle [radius=.3];
\node at (3,2.5)  {\textcolor{white}{$\boldmath{v_1}$}}; 
\draw [thick,->] (.4,2.6) to [out=30, in=150] (2.6,2.6);
\draw [thick,->] (2.6,2.4) to [out=-150, in=-30] (.4,2.4);
\path[->] (foo1) edge [loop] node {} ();
\path[->] (foo2) edge [loop] node {} ();

\node at (1.5,1) {$\G$};
\node at (1.5,0) {\ };
\end{tikzpicture}

\caption{Frames showing \SP-undefinability of transitive frames with the McKinsey property.}\label{f:defpic4}
\end{figure}
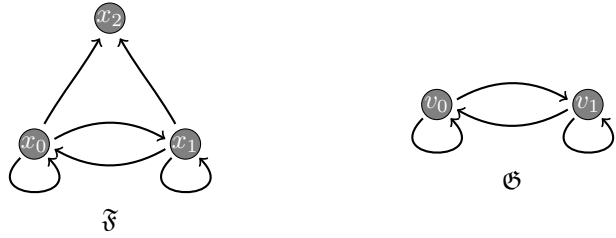
Consider the homomorphism $f\colon\G\to \F$, where $f(v_i)=x_i$ for $i\leq 1$,
and let
\[
Z=\{ (x_0,v_0),(x_1,v_1),(x_2,v_0),(x_2,v_1) \}.
\]
Then it is easy to check that
$(\F,f) >\!>_Z$ $(\G,\idmap,v_i)$,  for $i\le 1$. 
\end{proof}


As mentioned above, the class of linear quasiorders is \SP-definable. However, confluent quasiorders (the frames for the modal logic $\mf{S4.2}$) and quasiorders with the McKinsey property (the frames for the modal logic $\mf{S4.1}$) are not \SP-definable, which is a consequence of the following:

\begin{proposition}
Every unimodal \sptheory{} $\SPL \supseteq \Esfour$ is a subframe \sptheory. 
\end{proposition}
\begin{proof}
We show that, for every \spequation{} $\e = (\sigma \imp \tau)$, if $\F \not\models \e$ and $\F = (W,R)$ is a subframe of some quasiorder $\F' = (W',R')$, then $\F' \not\models \e$. Let $\M = (\F,\V)$ be such that 
$\M,w \not\models \e$, for some $w\in W$, and let $\M'= (\F',\V)$.
By induction on the construction of an \spterm{} $\varrho$, we show that 
$\{u\mid\M,u\kmodels\varrho\} =\{ u\mid\M',u\kmodels\varrho\}\cap W$, and so $\M',w \not\models \e$.
The basis of induction follows from the definition, and the cases of $\top$ and $\land$ are trivial. 
Let $\varrho = \d\varrho'$. By IH,  
$\{u\mid\M,u\kmodels\varrho\} \subseteq\{ u\mid\M',u\kmodels\varrho\}\cap W$.
For the converse inclusion, there are four cases. The case $\varrho' = \top$ is trivial as $R$ is reflexive.
Now, let 
$u\in W$ be such that $\M',u\kmodels\varrho$.
Then $\M',v\kmodels\varrho'$,  for some $v\in W'$ with $(u,v)\in R'$.  If $\varrho'$ is a variable, then $v \in W$ and $\M,v \kmodels\varrho'$ by the definition of $\M'$, and 
so $\M',u \kmodels\varrho$.
If $\varrho' = \d\pi$ then, by transitivity of $R'$, $\M',u \kmodels\d\pi$, and so, by IH, 
$\M,u \kmodels\varrho'$, from which, in view of reflexivity of $R$, we obtain $\M,u \kmodels\varrho$. 
Finally, let $\varrho' = \pi_1 \land\dots\land \pi_n$, where none of the $\pi_i$ is a conjunction or 
$\top$. If one of them is a variable, then $v \in W$ and we are done by IH. And if $\pi_i = \d\pi'_i$ for all $i$
then, by transitivity of $R'$, $\M',u \kmodels\d\pi'_i$ for all $i$, and we obtain $\M,u \kmodels\varrho$ by IH and reflexivity of $R$.
\end{proof}


\subsection{\SPcap-logics with $\bot$}\label{bot}

One can introduce a limited form of negation to the language of \spterms{}
by adding the `falsehood' constant $\bot$ (such that   
$\M,w\not\kmodels\bot$ for any point $w$ in any Kripke model $\M$). 
We call the \spterms{} of this extended language \emph{\SPb-formulas\/}, and define \emph{\SPb-implications\/} accordingly.
A class $\Cc$ of frames is \emph{\SPbi-definable} if
$\Cc=\CA_{\Ec}$, for some set $\Ec$ of \SPb-implications.
 
\begin{proposition}\label{p:defsame}
A class of frames is \SP-definable iff it is \SPbi-definable.
\end{proposition}

\begin{proof}
Suppose $\Cc=\CA_{\Ec}$, for some set $\Ec$ of \SPb-implications. As \SPb-implications $\sigma\imp\tau$
hold in all frames whenever $\sigma$ contains $\bot$, we may assume that $\bot$ only occurs in $\tau$.
Then it is easy to see that, for every frame $\F$, we have
$\F\models\sigma\imp\tau$ iff $\F\models\sigma\imp p$, where $p$ is a fresh variable not occurring in
$\sigma$.
\end{proof}

All the notions introduced above can be extended to \SPb. Thus, 
a  structure $\A =(A,\land,\bot,\top,\dR)_{R\in \R}$ is called an \emph{\SPb-type algebra} (\emph{of signature\/} $\R$). 
Given \SPb-type algebras $\Aa$ and $\Bb$ of the same signature, a function $\eta \colon \Aa\to\Bb$ is 
an \SPb-\emph{embedding} if it is an \SPa-embedding and $\eta(\bot)=\bot$.
We call $\A$ a \emph{bounded meet-semilattice with normal monotone operators\/} (or $\SLOba$) 
if $(A,\land,\top,\dR)_{R\in \R}$ is a \SLOa{} with $\le$-smallest element $\bot$, and $\dR\bot=\bot$ for
$R\in\R$. The set of \SPb-implications that are valid in all $\SLOba$s is denoted by $\SPi^\bot$.
For a set $\Ec$ of \SPb-implications, $\SLO^\bot_{\Ec}$ denotes the class of $\SLOba$s 
validating $\Ec$. We set 
\[
\Ec\models_{\SLOb}\e \quad\mbox{iff}\quad\mbox{$\A\models\e$\ \ for every  $\A\in\SLO^\bot_{\Ec}$.}
\]
(Note that $\Ec\models_{\SLOb}$ can be captured syntactically by adding the axioms 
$\bot\imp p$ and $\dR\bot\imp\bot$, for $R\in\R$,  to the calculus in \eqref{axioms}--\eqref{rules}.)
For any set $\Ec$ of \SPb-implications,
we define the \emph{\SPbi-logic 
\mbox{$\SPi^\bot + \Ec$}}
\emph{axiomatised by $\Ec$} as
%
\[
\SPi^\bot + \Ec =\{\e\mid \mbox{$\e$ is an \SPb-implication and }\Ec\models_{\SLOb}\e\}.
\]
Now one can
define the notions of completeness, complexity, finite frame property in the same way as in the \SPa-case. 
We give examples of incomplete \sptheories{} $\SPi+\Ec$ such that 
 $\SPi^\bot+\Ec$ is a \eqcomp{}
or even \complex{}  \SPbi-logic.

\begin{example}\em
By Theorem~\ref{t:simple2}, $\SPi+\Ec$ for $\Ec=\{p\imp\d p,\d p\imp \d q\}$ is an incomplete \sptheory. However, 
only the one-element $\SLOba$ can validate the \SPbi-logic $\SPi^\bot+\Ec$, and so $\Ec\models_{\SLOb}\e$ for every \SPb-implication $\e$. Thus, $\SPi^\bot+\Ec$ is a complete \SPbi-logic.
By Theorem~\ref{t:multiincomp}, $\SPi+\{\dR\dS p\imp q\}$ is an incomplete \sptheory. 
However, using a proof similar to that of Theorem~\ref{t:unicomp},
one can readily show that $\SPi^\bot+\{\dR\dS p\imp q\}$ is a \complex{} \SPbi-logic.
\end{example}

On the other hand, completeness and complexity do transfer from \SPa{} to \SPb: 

\begin{proposition}
Let $\Ec$ be a set of  \spequations.
\begin{itemize}
\item[$(i)$]
If the \sptheory{} $\SPi+\Ec$ is \eqcomp, then the \SPbi-logic $\SPi^\bot+\Ec$ is \eqcomp.

\item[$(ii)$]
If the \sptheory{} $\SPi+\Ec$ is \complex, then the \SPbi-logic $\SPi^\bot+\Ec$ is \complex.
\end{itemize}
\end{proposition}

\begin{proof}
$(i)$
Suppose $\Ec\models_{\CA}\e$ for some \SPb-implication $\e$ containing $\bot$. 
Then we may assume that $\e$ is of the form $\sigma\imp\bot$, in which case 
$\Ec\models_{\CA}\sigma\imp p$, for a fresh variable $p$. 
Also, $\Ec\models_{\CA}\dR\sigma\imp p$ for every $\dR$ occurring in $\Ec$, whence  
$\Ec\models_{\SLO}\sigma\imp p$ and $\Ec\models_{\SLO}\dR\sigma\imp p$.
So, in every $\A\in\SLO_{\Ec}$, there is a $\le$-smallest element $\bot$, 
for which $\dR\bot=\bot$ for every $\dR$ occurring in $\Ec$. This shows that $\Ec\models_{\SLOb}\sigma\imp\bot$.

$(ii)$ 
Suppose $\A\in\SLO^\bot_{\Ec}$. Then the \sptype{} reduct $\A^\downarrow$ of $\A$ is in
$\SLO_{\Ec}$, and so there is an \SPa-embedding $f\colon\A^\downarrow\to\Fslo$ for some 
$\F=(W,R^\F)_{R\in\R}$ with $\F\models\Ec$. Let $V=W\setminus f(\bot)$ and $R_V^\F=R^\F\cap(V\times V)$, for $R\in\R$. Then it is easy to see that the frame $\G=(V,R_V^\F)_{R\in\R}$ is a generated subframe
of $\F$ (and so $\G\models\Ec$), and the map $g\colon\A\to\G^{\star\bot}$ defined by
$g(a)=f(a) \setminus f(\bot)$ is an \SPb-embedding.
\end{proof}

A \eqcomp\ (\complex) \SPbi-logic can always be turned into a \eqcomp\ (\complex) \sptheory,
using a fresh diamond operator:

\begin{theorem}\label{thm:redbot}
Let $\Ec$ be a set of \SPbi-implications not using $\dR$. 
Let $\Ec'$ be obtained from $\Ec$ by
replacing each occurrence of $\bot$ by $\dR\top$ and adding $\dR \top\imp p$ and $\dS\dR\top\imp \dR\top$,
for each $\dS$ occurring in $\Ec$.
Then $\SPi^\bot+\Ec$ has property $P$ iff $\SPi+\Ec'$ has property $P$, where $P$ stands for any of the following: `is complete'\!, `is complex'\!, `has the finite frame property'\!, `is decidable'\!.
\end{theorem}

\begin{proof}
Let $\R_{\Ec}=\{S\mid \dS\mbox{ occurs in }\Ec\}$.
Given an \SPb-implication $\e$ using only $\dS$, for $S\in\R_{\Ec}$, denote by $\e^{\uparrow}$ the 
\spequation{} obtained by replacing each occurrence of $\bot$ in $\Ec$ by $\dR \top$.
Similarly, for any $\A\in\SLO^\bot_{\Ec}$, denote by $\A^{\uparrow}$ the \sptype{} reduct of $\A$ with
an additional operator $\dR$ for which $\dR a=\bot$ for all $a\in A$. 
Then $\A^\uparrow\in\SLO_{\Ec'}$, and $\A\models\e$ iff $\A^\uparrow\models\e^\uparrow$.
Conversely, given an \spequation{} $\e$ using only $\dS$, for $S\in\R_{\Ec}\cup\{R\}$, 
denote by $\e^{\downarrow}$ the \SPb-implication obtained by replacing each maximal subformula
 of the form $\dR\varrho$ in $\e$ with $\bot$.
Observe that in any $\A\in\SLO_{\Ec'}$, $\dR\top$ is the $\le$-smallest element with
$\dS\dR\top=\dR\top$ for all $S\in\R_{\Ec}$. Denote by $\A^{\downarrow}$ the result of removing $\dR$ from $\A$ and setting $\bot=\dR\top$. 
Then $\A^\downarrow\in\SLO^\bot_{\Ec}$, and $\A\models\e$ iff $\A^\downarrow\models\e^\downarrow$.
It remains to observe that, for any frame $\F=(W,S^\F)_{S\in\R_{\Ec}}$, we have 
$\F\models\Ec$ iff $(W,S^\F,\emptyset)_{S\in\R_{\Ec}}\models\Ec'$,
and $R^\F=\emptyset$ follows whenever $(W,S^\F,R^\F)_{S\in\R_{\Ec}}\models\Ec'$.
With these observations, all the statements of the theorem are straightforward.
\end{proof}


\subsection{\spquasiequationscap}\label{qeqcomp}

An \emph{\spquasiequation\/}, $\q$, takes the form $\frac{\e_{1},\ldots,\e_{n}}{\e}$, where $\e_{1},\ldots,\e_{n},\e$ are
\spequations. 
We identify the rule $\frac{\emptyset}{\e}$ with $\e$.
We say that an \spquasiequation{} 
$\q=\frac{\e_{1},\ldots,\e_{n}}{\e}$
\emph{holds} in a Kripke model $\M$ and write $\M \models \q$ if $\M \models \e$ whenever $\M \models \e_i$ for $1 \le i \le n$.
We say that $\q$ is \emph{valid} in a frame $\F$ and write $\F \models \q$ if $\q$ holds in every Kripke model based on $\F$. 
Given a set $\varTheta$ of \spquasiequations, we write
$\F\models\varTheta$ whenever $\F\models\q$ for every $\q\in\varTheta$ and set ${\sf Kr}_{\varTheta}= \{ \F \mid \F\models \varTheta\}$.

We say that $\q$ is \emph{valid} in an algebra $\A$ having an \sptype{} reduct and write  $\A\models \q$ if $\A$
validates the \sptype{} \emph{quasiequation} 
\[
(\e_{1}^{\ast}\mathop{\&} \dots \mathop{\&} \e_{n}^{\ast})\Rightarrow \e^{\ast},
\]
 where $(\sigma\imp \tau)^{\ast}= (\sigma \land \tau\approx\sigma)$:
for any valuation $\mathfrak{a}$ in $\A$, whenever $\A\models \e_{i}[\mathfrak{a}]$ for all $i$ $(1\leq i \leq n)$, then $\mathfrak{A}\models \e[\mathfrak{a}]$.
A set $\Qc$ of \spquasiequations{} is called an \emph{\qsptheory} if
$\Qc=\{\q\mid \Aa\models\q\mbox{ for every }\Aa\in\Cc\}$ for some class $\Cc$ of \SLOa{}s.
Given an \qsptheory{} $\Qc$, we write
$\A\models\Qc$ if $\A\models\q$ for any $\q\in\Qc$.
For a class $\Cc$ of algebras with \sptype{} reducts, let $\Cc_{\Qc}=\{ \Aa\in \Cc \mid \Aa \models \Qc\}$. We say 
that an \spquasiequation{} $\q$ \emph{follows from $\Qc$ over} $\Cc$ and write $\Qc \models_\Cc \q$ if 
$\A \models \q$, for any $\A\in \Cc_{\Qc}$.
We call $\Qc$ 
\begin{itemize}
\item[--]
$\Cc$-\emph{\embed} if every $\Aa\in\SLO_{\Qc}$ is embeddable into the \sptype{} reduct of some $\Bb\in\Cc_{\Qc}$;

\item[--]
$\Cc$-\emph{\qeqcons} if $\Qc \models_\Cc \q$ implies $\Qc \models_{\SLO} \q$, for every  
\spquasiequation{} $\q$;

\item[--]
$\Cc$-\emph{\eqcons} 
if $\Qc \models_\Cc \e$ implies $\Qc \models_{\SLO} \e$, for every \spequation{} $\e$.
\end{itemize}
In particular, let
$$
\CAa  =\{\Fslo\mid \mbox{$\F$ is a frame}\},\qquad
\BAO  = \{\A\mid\mbox{$\A$ is a \BAOa}\}.
$$
Extending the corresponding notions for \sptheories, we call 
an \qsptheory{} $\Qc$ 
\begin{itemize}
\item[--]
\emph{\complex} if it is $\CAa$-\embed;

\item[--]
\emph{\qeqcomp} if it is $\CAa$-\qeqcons;

\item[--]
\emph{\eqcomp} if it is $\CAa$-\eqcons.
\end{itemize}
As quasiequations are preserved under taking subalgebras,  we always have:
\begin{equation}\label{cproperties}
\Cc\mbox{-\embed}\ \Rightarrow\ \Cc\mbox{-\qeqcons}\ \Rightarrow\ \Cc\mbox{-\eqcons}.
\end{equation}
Also, 
%
since $\Fslo$ is the \sptype{} reduct of some \BAOa, we have:
\[
\begin{array}{llll}
\mbox{\complex}& \Rightarrow \ \mbox{\qeqcomp}\ &  \Rightarrow & \mbox{\eqcomp}\\
 \qquad\Downarrow  & \qquad\qquad\Downarrow  && \qquad\Downarrow\\
 \mbox{\BAO-\embed} & \Rightarrow  \ \mbox{\BAO-\qeqcons} &\Rightarrow & \mbox{\BAO-\eqcons}.
\end{array}
\]

\begin{lemma}\label{l:qeqconstoembed}
For any \qsptheory{} $\Qc$, if $\Qc$ is \BAO-\qeqcons, then
$\Qc$  is \BAO-\embed.
\end{lemma}
\begin{proof}
Suppose $\Qc$ is \BAO-\qeqcons{} and $\Aa\in\mbox{\SLO}_{\Qc}$.
To embed $\A$ into the \sptype{} reduct of some $\Bb\in\BAO_{\Qc}$, take 
the diagram $D_\Aa$ of $\Aa$, that is, the set all literals---equations and negated equations---that hold in $\Aa$ and are built  from the elements of $\Aa$ as constants using the \sptype{} operations. For any finite set $X$ of literals of this extended type, 
we write $X(a_1,\dots,a_n)$ to indicate that
the $\Aa$-type constants occurring in the literals in $X$ are among $a_1,\dots,a_n$.
If $X=\{\varphi\}$, we write $\varphi(a_1,\dots,a_n)$.
We write $\varphi(p_1/a_1,\dots,p_n/a_n)$ for the \SP-type literal where the constants $a_i$ in $\varphi$ are
simultaneously replaced by variables $p_i$.  

\begin{lclaim}\label{c:exists}
For any finite subset $X(a_1,\dots,a_n)$ of $D_\Aa$, there exist $\Bb^X\in\BAO$ and elements $a_1^X,\dots,a_n^X$ in $\Bb^X$ such that $\Bb^X\models\Qc$ and 
\begin{equation}\label{allx}
\Bb^X\models\bigwedge_{\varphi\in X}\varphi(p_1/a_1,\dots,p_n/a_n)[a_1^X,\dots,a_n^X].
\end{equation}
\end{lclaim}
\begin{proof}
If all literals in $X$ are equations, then we can take $\Bb^X$ to be the one-element \BAOa\ (for
which $\Bb^X\models\Qc$ for any $\Qc$) and set $a_i^X$ to be its only element, for $i=1,\dots,n$.
It is easy to see that~\eqref{allx} holds.

Now suppose $\e_1,\dots,\e_k$ are the equations in $X$ and
$\neg\e_{1}',\dots,\neg\e_{m}'$ are the negated equations in $X$, for $m\geq 1$ (we can always assume that $k\geq 1$).
For each $j$, $1\leq j\leq m$, take the \sptype{} quasiequation 
\[
\q_j ~=~ (\e_1\mathop{\&}\dots\mathop{\&}\e_k\Rightarrow \e_{j}')(p_1/a_1,\dots,p_n/a_n).
\] 
Then $\Aa\not\models\q_j$, and so, since $\Qc$ is \BAO-\qeqcons, there is some $\Bb_j\in\BAO$ with $\Bb_j\models\Qc$ and $\Bb_j\not\models\q_j$. 
Then there are
$b_1^j,\dots,b_n^j$ in $\Bb_j$ such that
\[
\Bb_j\models\Bigl(\bigwedge_{i=1}^k\e_i\land\neg\e_{j}'\Bigr)(p_1/a_1,\dots,p_n/a_n)[b_1^j,\dots,b_n^j].
\]
Now let $\Bb^X=\prod_{j=1}^m\Bb_j$ and $a_i^X = (b_i^1,\dots,b_i^m)$, for $i=1,\dots,n$. Then clearly we have~\eqref{allx}.
As the class $\BAO_{\Qc}$ 
is a quasivariety, it is closed under direct products, and so $\Bb^X\in\BAO_{\Qc}$ as required.
\end{proof}

Let $T_\Aa$ be the set of all finite subsets of $D_\Aa$.  For every $X\in T_\Aa$,
let 
\[
\J_X = \{Y\in T_\Aa \mid X\subseteq Y\}.
\]
As $X_1\cup\dots\cup X_m\in\J_{X_1}\cap\dots\cap\J_{X_m}$, the
collection $\{\J_x \mid X\in T_\Aa\}$ has the finite intersection property, and so there is an ultrafilter 
 $U$ over $T_\Aa$ extending  $\{\J_x \mid X\in T_\Aa\}$.
For $X\in T_\Aa$, take the \BAOa\ $\Bb^X$ given by Claim~\ref{c:exists}, and
let
\[
\Bb=\prod_{X\in T_\Aa}\Bb^X/U.
\]
As the class $\BAO_{\Qc}$  is a quasivariety, it is closed under ultraproducts, and so
$\Bb\in\BAO_{\Qc}$.
Define an $\eta\colon\Aa\to\Bb$ map by taking 
$\eta(a)= [(\hat{a}^X)_{X\in T_\Aa}]_U$, where for all $a$ in $\Aa$ and $X\in T_\Aa$,
\[
\hat{a}^X =\left\{
\begin{array}{ll}
a^X, & \mbox{if $a$ occurs in some literal in $X$},\\
\mbox{arbitrary element of $\Bb^X$}, & \mbox{otherwise}.
\end{array}
\right.
\]
By Claim~\ref{c:exists} and \L os' Theorem~\cite{Chang&Keisler73}, 
for every $\varphi(a_1,\dots,a_n)\in D_\Aa$, we have
\[
\Bb\models\varphi(p_1/a_1,\dots,p_n/a_n)[\eta(a_1),\dots,\eta(a_n)].
\]
Thus, $\eta$ is an \SPa-embedding from $\Aa$ into the \sptype{} reduct of $\Bb$.
\end{proof}

We call an \qsptheory{} $\Qc$ \emph{$\BAO$-complex} if the \sptype{} reduct of every 
$\Aa\in\mbox{\BAO}_{\Qc}$ is embeddable into some $\Fslo$ with $\F\in \CA_{\Qc}$.
Note that,
as \spequations{} correspond to Sahlqvist formulas in modal logic, any
\sptheory{} $\SPL$ is $\BAO$-complex. 
As a consequence of Lemma~\ref{l:qeqconstoembed} we obtain:

\begin{theorem}\label{t:qeq}
For every $\BAO$-complex \qsptheory{} $\Qc$, the following are equivalent:
\begin{enumerate}
\item[$(i)$]
$\Qc$ is \complex;
\item[$(ii)$]
$\Qc$ is \qeqcomp;
\item[$(iii)$]
$\Qc$  is \BAO-\qeqcons;
\item[$(iv)$]
$\Qc$ is \BAO-\embed.
\end{enumerate}
\end{theorem}

\begin{proof}
(i)~$\Rightarrow$~(ii) follows from~\eqref{cproperties}; (ii)~$\Rightarrow$~(iii) is trivial;
(iii)~$\Rightarrow$~(iv) follows from Lemma~\ref{l:qeqconstoembed}; and 
(iv)~$\Rightarrow$~(i) follows from the fact that $\Qc$ is $\BAO$-complex.
\end{proof}


\section{Conclusion}\label{concl}

In this article, we have started developing the completeness theory of \sptheories. Of course, many interesting and challenging problems remain to be explored.  A few concrete open questions have already been mentioned above, and there is a more or less standard list of problems regarding properties of modal logics and their lattices; see, e.g., \cite{Chagrov&Z97,Blackburnetal01,hand2001,mlhb2007}. 
Here, we briefly discuss few possible directions of follow-up research.

(1) In Boolean modal logic, the \emph{degree of Kripke incompleteness} of a normal modal logic $\Ll$---that is, the cardinality of the set of normal modal logics whose Kripke frames coincide with the Kripke frames of $\Ll$ \cite{Fine74}---has been used to analyse the position of Kripke incomplete logics within 
the lattice of all normal modal logics. Wim Blok~\cite{Blok78} established the following dichotomy: the degree of Kripke incompleteness of a consistent 
normal unimodal logic $\Ll$ is either $2^{\aleph_{0}}$ or $1$, in which case $\Ll$ is a  union of co-splitting logics; see also \cite{Litak08:au,hand2001,kracht07}. Given this complete classification, the question arises as to whether one can also characterise the degree of Kripke incompleteness of \sptheories{} and whether this is again linked to co-splittings (now in the lattice of \sptheories) and the existence of some analogue of Jankov-Fine formulas~\cite{Jankov63b,Fine74c}. 


(2) To prove undefinability of frame classes by 
\spequations, we developed a necessary condition for frame definability. In Boolean
modal logic, the Goldblatt--Thomason theorem~\cite{Goldblatt&Thomason74} provides necessary and sufficient conditions for frame definability
in terms of p-morphisms, generated subframes, disjoint unions, and ultrafilter extensions. Can one give 
natural necessary and sufficient conditions for frame definability by \spequations? 

(3) It is readily seen that \spquasiequations{} can define non-elementrary frame conditions and thus behave  differently from \spequations{}~\cite{islands10}.
We have also seen that \complex{} \qsptheories{} are exactly those that are \qeqcomp.
 Thus, it would be interesting to extend the completeness theory of \sptheories{} developed in this paper to \qsptheories.

(4) The embeddability of \SLOa{}s into full complex algebras of Kripke frames is shown by Sofronie-Stokkermans~\cite{Sofronie-Stokkermans01,stokkermans2008} using a method that is different from those in \S\S\ref{Jembed}--\ref{Tembed} and  involves \emph{distributive lattices with normal and $\lor$-additive 
operators} (\DLOa{}s). A given \SLOa{} $\Aa$ is first embedded into the DLO $\A^\lor$ of its downsets, which is then embedded into the full complex algebra of some frame $\F$ over the prime filters of $\A^\lor$ using 
Goldblatt's~\cite{Goldblatt89} extension of Priestley duality~\cite{Priestley70} to operators.
She also shows that validity of \spequations{} of the form 
\mbox{$\d_1\dots\d_n p\imp\d_0 p$} transfers from $\Aa$ to $\F$.
It would be interesting to study the boundaries of this method and its connections to  \S\S\ref{Jembed}--\ref{Tembed}. More generally, one can ask which
\spequations{} are \SLOa--to--\DLOa- and/or \DLOa--to--\BAOa-conservative?
The latter question can also be investigated for \SP$^\lor$-implications, that is, implications
between \spterms{} with disjunction. 

\smallskip
\noindent
{\bf Acknowledgements.} 
This work was supported by the U.K.~EPSRC grants EP/M012646 and EP/M012670 `iTract: Islands of Tractability in Ontology-Based Data Access'\!. 
We are grateful to the anonymous reviewer for encouraging us to improve the presentation and terminology, and to redesign the notation.

\end{document}